\documentclass[11pt]{amsart}
\usepackage[applemac]{inputenc}
\usepackage[titletoc]{appendix}

\usepackage[backend=biber, style=numeric, date=year, url=false, doi=false, isbn=false, eprint=true, firstinits=true, maxcitenames=5, maxbibnames=5]{biblatex}
\AtEveryBibitem{%
  \ifentrytype{online}{%
  }{%
    \clearfield{eprint}%
    \clearfield{urldate}%
  }%
}
\usepackage{blindtext}
\usepackage{amssymb}
\usepackage{amsfonts}
\usepackage{amsthm}
\usepackage{tikz}
\usepackage{caption}
\usepackage{color}
\usepackage{graphicx}
\usepackage{xcolor}
\usepackage{shadethm}
\usepackage{subfigure}
\usepackage{epigraph}
\usepackage{placeins}

\usepackage{enumerate}
\usepackage{verbatim}
\usetikzlibrary{trees}
\usetikzlibrary{decorations.markings}
\usetikzlibrary{arrows}
\usepackage{mathrsfs}
\usepackage[margin=1.0in]{geometry}
\usepackage{xparse}
\usepackage{diagrams}

\usetikzlibrary{positioning,arrows,patterns}
\usetikzlibrary{decorations.markings}
\usetikzlibrary{calc}
\tikzset{
  photon/.style={decorate, decoration={snake}, draw=black},
  fermion/.style={draw=black, postaction={decorate},decoration={markings,mark=at position .55 with {\arrow{>}}}},
  vertex/.style={draw,shape=circle,fill=black,minimum size=5pt,inner sep=0pt},
particle/.style={thick,draw=black},
particle2/.style={thick,draw=blue},
avector/.style={thick,draw=black, postaction={decorate},
    decoration={markings,mark=at position 1 with {\arrow[black]{triangle 45}}}},
gluon/.style={decorate, draw=black,
    decoration={coil,aspect=0}}
 }
\NewDocumentCommand\semiloop{O{black}mmmO{}O{above}}
{%
\draw[#1] let \p1 = ($(#3)-(#2)$) in (#3) arc (#4:({#4+180}):({0.5*veclen(\x1,\y1)})node[midway, #6] {#5};)
}

\usepackage[mathcal]{euscript}

\usepackage{url}
\usepackage{hyperref}
\hypersetup{colorlinks=true, citecolor=red, linktocpage, linkcolor=blue, urlcolor=cyan} 

\theoremstyle{plain}
\newtheorem{thm}{Theorem}[section]
\newtheorem{lem}[thm]{Lemma}
\newtheorem{prop}[thm]{Proposition}

\theoremstyle{definition}
\newtheorem{defn}[thm]{Definition}

\newtheorem*{thm*}{Theorem}
\newtheorem*{lem*}{Lemma}
\newtheorem*{prop*}{Proposition}
\newtheorem*{cor*}{Corollary}
\newtheorem*{exe*}{Exercise}
\newtheorem*{defn*}{Definition}

\theoremstyle{remark}
\newtheorem{rem}[thm]{Remark}

\newtheorem{ex}[thm]{Example}
\newtheorem{ass}[thm]{Assumption}

\newcommand{\R}{\mathbb{R}}
\newcommand{\Z}{\mathbb{Z}}

\newcommand{\E}{\mathbb{E}} 
\newcommand{\bbE}{\mathbb{E}} 
\newcommand{\bbX}{\mathbb{X}}

\newcommand{\calY}{\mathcal{Y}}

\newcommand{\ii}{{\mathrm{i}}}
\newcommand{\dd}{{\mathrm{d}}} 

\newcommand{\ddt}{\frac{\dd}{\dd t}
\Big\vert_{{t}=0}}

\newcommand{\C}{\mathbb{C}}

\newcommand{\Der}{{\mathrm{Der}}}

\DeclareMathOperator{\Aut}{Aut}
\DeclareMathOperator{\End}{End}

\newcommand{\T}{\textsf{T}}

\newcommand{\de}{\partial}

\newcommand{\calB}{\mathcal{B}}
\newcommand{\calH}{\mathcal{H}}
\newcommand{\calS}{\mathcal{S}}

\newcommand{\calL}{\mathcal{L}}
\newcommand{\calM}{\mathcal{M}}

\newcommand{\calP}{\mathcal{P}}
\newcommand{\calF}{\mathcal{F}}

\newcommand{\btpsi}{\boldsymbol{\widetilde{\psi}}}

\newcommand{\sfe}{{\mathsf{e}}}
\newcommand{\sfX}{{\mathsf{X}}}

\newcommand{\sfx}{{\mathsf{x}}}

\newcommand{\qtconn}{\nabla_\mathsf{G}}

\def\gpd{\,\lower1pt\hbox{$\longrightarrow$}\hskip-.24in\raise2pt
               \hbox{$\longrightarrow$}\,}       

\newcommand{\gm}{\Gamma }

\let\Tilde=\widetilde
\let\Bar=\overline

\let\Hat=\widehat

\newcommand{\boldeta}{\boldsymbol\eta}
\newcommand{\hateta}{\widehat{\boldsymbol\eta}}
\newcommand{\hatX}{\widehat{\mathsf{X}}}

\DeclareMathOperator{\dr}{d}
\DeclareMathOperator{\Map}{Map}

\newcommand{\I}{\mathrm{i}}

\newcommand{\ee}{\textnormal{e}}

\newcommand{\calV}{\mathcal{V}}

\newcommand{\romI}{\textnormal{I}}
\newcommand{\romII}{\textnormal{II}}
\usepackage{todonotes}

\makeatletter

\pgfdeclareshape{genus}{
 \anchor{center}{\pgfpointorigin}
\backgroundpath{
    \begingroup
    \tikz@mode
    \iftikz@mode@fill

         \pgfpathmoveto{\pgfqpoint{-0.78cm}{-.17cm}}
     \pgfpathcurveto %
            {\pgfpoint{-0.35cm}{-.44cm}}
        {\pgfpoint{0.35cm}{-.44cm}}
        {\pgfpoint{.78cm}{-0.17cm}} 
     \pgfpathmoveto{\pgfqpoint{-0.78cm}{-0.17cm}}
     \pgfpathcurveto %
            {\pgfpoint{-0.25cm}{.25cm}}
        {\pgfpoint{.25cm}{.25cm}}
        {\pgfpoint{0.78cm}{-0.17cm}}
        \pgfusepath{fill}
        \fi

        \iftikz@mode@draw
     \pgfpathmoveto{\pgfqpoint{-1cm}{0cm}}
     \pgfpathcurveto %
            {\pgfpoint{-0.5cm}{-.5cm}}
        {\pgfpoint{0.5cm}{-.5cm}}
        {\pgfpoint{1cm}{0cm}}

     \pgfpathmoveto{\pgfqpoint{-0.75cm}{-0.15cm}}
     \pgfpathcurveto %
            {\pgfpoint{-0.25cm}{.25cm}}
        {\pgfpoint{.25cm}{.25cm}}
        {\pgfpoint{0.75cm}{-0.15cm}}
              \pgfusepath{stroke}
        \fi
        \endgroup
    }
    }

    \makeatother

\begin{document}

\title[Glob. for Pert. Quant. of Nonl. Split AKSZ Sigma Models on Manifolds with Boundary]{Globalization for Perturbative Quantization of Nonlinear Split AKSZ Sigma Models on Manifolds with Boundary}
\author[A. S. Cattaneo]{Alberto S. Cattaneo}
\author[N. Moshayedi]{Nima Moshayedi}
\author[K. Wernli]{Konstantin Wernli}
\address{Institut f\"ur Mathematik\\ Universit\"at Z\"urich\\ 
Winterthurerstrasse 190
CH-8057 Z\"urich}
\email[A. S.~Cattaneo]{cattaneo@math.uzh.ch}
\address{Institut f\"ur Mathematik\\ Universit\"at Z\"urich\\ 
Winterthurerstrasse 190
CH-8057 Z\"urich}
\email[N.~Moshayedi]{nima.moshayedi@math.uzh.ch}
\address{Institut f\"ur Mathematik\\ Universit\"at Z\"urich\\ 
Winterthurerstrasse 190
CH-8057 Z\"urich}
\email[K.~Wernli]{konstantin.wernli@math.uzh.ch}
\thanks{
This research was (partly) supported by the NCCR SwissMAP, funded by the Swiss National Science Foundation, and by the
COST Action MP1405 QSPACE, supported by COST (European Cooperation in Science and Technology).  We acknowledge partial support of SNF grant No. 200020\_172498/1. K. W. acknowledges partial support by the Forschungskredit of the University of Zurich, grant no. FK-16-093.
}
\maketitle

\begin{abstract} 
We describe a covariant framework to construct a globalized version for the perturbative quantization of nonlinear split AKSZ Sigma Models on manifolds with and without boundary, and show that it captures the change of the quantum state as one changes the constant map around which one perturbs. This is done by using concepts of formal geometry.
Moreover, we show that the globalized quantum state can be interpreted as a closed section with respect to an operator that squares to zero. This condition is a generalization of the modified Quantum Master Equation as in the BV-BFV formalism, which we call the modified ``differential''  Quantum Master Equation.
\end{abstract}

\tableofcontents

\section{Introduction}

\subsection{Motivation}
The goal of this paper is to construct perturbative partition functions of certain AKSZ theories - on manifolds with and without boundary - that vary in a ``covariant fashion'' as one changes the point of expansion. This is achieved combining the BV-BFV formalism (Batalin--Vilkovisky and Batalin--Fradkin--Vilkovisky) (\cite{CMR2}) with methods of formal geometry (\cite{GK,B,CFT,BCM}). 
The globalization method in the case of a field theory on manifolds with boundary has been considered so far only in \cite{CMW2}, of which the current paper is a far reaching generalization. In \cite{CMW2} we performed this task for a particular example of an AKSZ theory, the Poisson Sigma model \cite{I,SS1,SS2} with constant Poisson structure. We briefly introduce the main players. \\ 

The \textsf{BV-BFV formalism} - briefly recalled in Section \ref{sec:BVBFV} - is a method for the perturbative quantization of gauge theories on manifolds with boundary compatible with cutting and gluing. It is named after Batalin, Fradkin and Vilkovisky, who introduced what are now known as the BV and the BFV formalisms in \cite{BV1,BV2,BF1,BF2,FV1,FV2}; see also \cite{S,Henneaux1994,Cost} and references therein. The classical framework for the BV-BFV formalism\footnote{The BV-BFV formalism is a consistent combination of the BV formalism for the bulk with the BFV formalism on the boundary.} was introduced in \cite{CMR1}. A classical BV-BFV theory associates to every manifold $\Sigma$ of a fixed dimension - possibly with boundary - the data of a ``BV-BFV manifold'' (\cite{CMR1}), the space of fields $\calF_\Sigma$ (plus extra data). Classical BV-BFV theories can be quantized by the construction in \cite{CMR2}. This procedure associates to $\Sigma$ a bi-complex $\calH_{\Sigma}$ with two commuting coboundary operators $\Delta_{\calV_\Sigma}$ (the BV Laplacian) and $\Omega_{\de\Sigma}$ (the BFV boundary operator). The \textsf{modified Quantum Master Equation} (mQME) is the statement that the partition function $\psi_\Sigma$ is closed with respect to the coboundary operator $\hbar^2\Delta_{\calV_\Sigma} + \Omega_{\de\Sigma}$, i.e.
\begin{equation}
(\hbar^2\Delta_{\calV_\Sigma} + \Omega_{\de\Sigma})\psi_{\Sigma} = 0.
\end{equation}  However, this construction works only if the space of fields is linear, i.e. a vector space. If the space of fields is nonlinear one has to linearize it, which amounts to working with a formal neighbourhood of a classical solution in the space of fields. In this paper we show how this can be done consistently for a large set of solutions at once for AKSZ theories. \\ 

\textsf{AKSZ theories} were introduced by Alexandrov, Kontsevich, Schwarz and Zaboronsky in \cite{AKSZ}. They form a large class of topological BV theories that naturally admit BV-BFV extensions, as was shown in \cite{CMR1} and is recalled in the present paper in Section \ref{AKSZ}. In AKSZ theories the space of fields $\calF_{\Sigma}$ is a space of graded maps with target a fixed graded manifold $\calM$. If the target is a vector space, then also the space of fields has a vector space structure, but in many examples one is interested in the case where the target is nonlinear (a prominent one being the Poisson Sigma Model, see \cite{CF4}). In this case, the quantization is constructed by linearizing around constant maps. \\

In this paper, we use methods of formal geometry, reviewed in Appendix \ref{app:formal_geometry} (see also \cite{CF3,CFT,BCM}, and \cite{LS} for the case where the moduli space of solutions is graded) to define a ``covariant partition function'' $\btpsi_\Sigma$. It is an inhomogenoeous differential form with values in the vector bundle $\Hat{\calH}_{\Sigma,tot}$ over (the body of) the target with fiber over $x$ the space of states of the BV-BFV quantization around $x$. In Section \ref{sec_mdQME} we show that it satisfies the following generalization of the mQME that we call ``mdQME'' (for \textsf{modified differential Quantum Master Equation}):
\begin{equation}
\left(\dd_x -\I\hbar \Delta_{\calV_\Sigma} + \frac{\I}{\hbar}\Omega_{\de\Sigma}\right)\btpsi_{\Sigma} = 0. \label{eq:intro_mdQMe}
\end{equation}
We also show that the \textsf{quantum Grothendieck BFV (GBFV) operator} $\qtconn := \dd_x -\I\hbar \Delta_{\calV_\Sigma} + \frac{\I}{\hbar}\Omega_{\de\Sigma}$ squares to zero. The operator $\nabla_\mathsf{G}$ can be thought of as a ``connection'' on the total space $\Hat{\calH}_{\Sigma,tot}$ and hence we can think of it as a flat connection on $\Hat{\calH}_{\Sigma,tot}$ (see Subsection \ref{flat}). If one interprets $\qtconn$ as a quantum version of the Grothendieck connection \eqref{AppA:Aconnection_G}, Equation \eqref{eq:intro_mdQMe} says that $\btpsi_\Sigma$ corresponds to the Taylor expansion of a globally defined object on $M$. \\

One of the goals of this construction is to go further towards the deformation quantization of the relational symplectic groupoid \cite{C,CC1,CC2}. The next step will be an extension of the results obtained here to the Poisson Sigma Model with alternating boundary conditions \cite{CMW3}. However, we also hope to deepen the understanding of how perturbative partition functions depend on the point of expansion. In AKSZ Sigma Models, there is a nice smooth part of the moduli space of classical solutions given by constant maps. But e.g. in Chern--Simons theory the body of the target is a point, and one is interested in expanding around points representing equivalence classes of flat connections. This will be the subject of further investigation.

\subsection{Main results}
Let us summarize the main results of the paper. One of the main theorems of this paper is the modified differential Quantum Master Equation for anomaly free, unimodular AKSZ theories:
\begin{thm*}[\ref{thm:mdQME}]
Consider the full covariant perturbative state $\btpsi_{\Sigma,x}$ as a quantization of an anomaly free and unimodular split AKSZ theory with target $T^*[d-1]M$, where $M$ is a graded manifold. Then
\begin{equation}
\left(\dr_x -\I\hbar \Delta_{\calV_{\Sigma,x}} + \frac{\I}{\hbar} \boldsymbol{\Omega}_{\de \Sigma}\right) \btpsi_{\Sigma,x}=0,
\end{equation}
where we denote by $\dr_x$ the de Rham differential on $\Bar{M}$, the body of the graded manifold $M$.
\end{thm*}
Another main result is that the quantum GBFV operator is a coboundary operator:
\begin{thm*}[\ref{thm:flatness}]
The operator $\nabla_\mathsf{G}$ squares to zero, i.e.
\begin{equation}
(\nabla_\mathsf{G})^2\equiv 0.
\end{equation}
\end{thm*}
We also show how the state and the BFV operator transform under change of data. This is captured in the following theorem:
\begin{thm}[\ref{thm:dep_choices}]
Let $\boldsymbol{\Omega}_t$ be defined as in Definition \ref{full_BFV} and let $\btpsi_t$ be defined as in \ref{full_state_2} for all $t\in[0,1]$. Then we have 
\begin{align}
\frac{\dd}{\dd t}\Big\vert_{t=0}\boldsymbol{\Omega}_t &= \dd_x\tau+[\boldsymbol{\Omega}_{t=0},\tau] \\
\frac{\dd}{\dd t}\Big\vert_{t=0}\btpsi_t &= \qtconn (\btpsi_{t=0} \bullet \varrho) - \tau\btpsi_{t=0} 
\end{align}
for some operator $\tau\in\Gamma(\End(\calH_{tot}))$ and a section $\varrho \in \gm(\calH_{tot})$. 
\end{thm}

\subsection{Summary}
Let us give a brief overview of the paper.\\

\begin{itemize}
\item{In Section \ref{sec:BVBFV} we review the main concepts of the BV-BFV formalism as in \cite{CMR2}. 
}\\
\item{In Section \ref{sec:AKSZ_Quantization} we recall the notion of an AKSZ Sigma Model and describe the split version. Moreover, we linearize the AKSZ Model using methods of formal geometry and explain the globalization construction by adding a \textsf{formal globalization} part to the action. We formulate everything according to the BV-BFV formalism at the classical as well as the quantum level, where we also introduce the \textsf{full covariant state}. 
}\\
\item{In Section \ref{sec_mdQME} we introduce the \textsf{quantum GBFV operator} and formulate the \textsf{modified differential Quantum Master Equation}. One of the main result there is the proof of the mdQME. Moreover, we prove that the quantum GBFV operator is a coboundary operator, such that we have a well-defined cohomology theory.
}\\
\item{In Section \ref{sec:dep_choices} we show how the state and the BFV boundary operator transform under change of propagator, residual fields and exponential maps. 
}
\end{itemize}
\vspace{0.3cm}

Various details are discussed in the appendices: 

\begin{itemize}
\item{In Appendix \ref{app:Conf} we recall the compactification of various configuration spaces and their boundary strata.}\\
\item{In Appendix \ref{app:formal_geometry} we recall some notions of fomal geometry and its extension to graded manifolds.}
\end{itemize}
\subsection*{Acknowledgements}
We thank I. Contreras for helpful comments. Moreover, we want to thank the referee for pointing out important and helpful comments.

\tikzset{internal/.style={draw, shape=circle, fill,black,inner sep=2pt}}
\tikzset{residual/.style={draw, shape=circle, black,inner sep=1pt}}
\tikzset{propagator/.style={thick, -latex}}

\section{The BV-BFV formalism}
\label{sec:BVBFV}

The BV-BFV formalism is a gauge fixing formalism for gauge theories on manifolds with boundary, both at the classical 
(\cite{CMR1}) and quantum (\cite{CMR2}) level. We briefly  recall the most important ideas. 
Readers already familiar with the BV-BFV formalism  as in \cite{CMR2} can skip this section. Another reference for learning about this formalism is \cite{CattMosh1}.

\subsection{Field theory}
We start with the following definition of a classical field theory. 
\begin{defn}[Classical field theory]
A $d$-dimensional \textsf{classical field theory} associates to every $d$-dimensional manifold $M$ a space of fields $F_M$ and an action functional $S_M\colon F_M \to \R$.
\end{defn}
Field theories are usually required to be \textsf{local}. For the purpose of the present paper, the following definition will suffice. 
\begin{defn}[Local field theory]
We say that a field theory $(F_M,S_M)$ is \textsf{local} if there is a fiber bundle $E \to M$ such that $F_M = \Gamma(E)$ and there is an integer $k$ such that 
\begin{equation}
S_M(\phi) = \int_M L[j^k(\phi)],\label{eq:Lagrangian_field_theory}
\end{equation}
where $j^k$ denotes $k$-th jet prolongation and $L \colon J^kE \to \mathrm{Dens}(M)$ is a function on the $k$-th jet bundle of $E$ with values in densities of $M$. $L$ is called the \textsf{Lagrangian} of the theory. 
\end{defn}
Let $(F_M,S_M)$ be a local  field theory. If $M\not=\varnothing$ and we don't fix any boundary conditions, there is a $1$-form $\alpha_{\de M}\in\Omega^1(F_{\de M})$ (the \textsf{Noether} $1$-form) such that the variation of the action $S_M$ is given by 
$$\delta S_M=\text{EL}_M+\pi^*_M\alpha_{\de M},$$
where $\pi_M\colon F_M\to F_{\de M}$ is the natural surjective submersion from the space of fields $F_M$ onto the space of fields $F_{\de M}$ on the boundary $\de M$. $F_{\de M}$ is given by restrictions of bulk fields and their normal jets to the boundary. We denote by $\text{EL}_M$ the $1$-form\footnote{$\text{EL}_M$ is the term that depends only on the variations of the fields but not on higher jets.} coming from the \textsf{Euler-Lagrange equations} (EL equations). The classical solutions are given by the critical points of $S_M$, i.e. by solutions of $\delta S_M=0$.
One can define a presymplectic form $\omega_{\de M}$ on $F_{\de M}$ by setting $\omega_{\de M}:=\delta\alpha_{\de M}$ (we think of $\delta$ as the de Rham differential on the space of fields). By techniques of symplectic geometry, such as \textsf{symplectic reduction}, one can obtain a symplectic manifold $(F^\de_{\de M},\omega^\de_{\de M})$.
Moreover, this construction is compatible with cutting and gluing (\cite{CMR1,CMR3}). This construction leads to a nice quantum formulation in the guise of path integrals after choosing a suitable polarization (\cite{CMR2}). We will discuss these issues in this section. 

\begin{rem}
Note that if $\de M=\varnothing$ we get the usual Euler-Lagrange equations from $\delta S_M=0$.
\end{rem}

\subsection{Finite dimensional BV theory}

Let $M$ be a closed manifold and let $F_M$ denote the space of fields associated to $M$. If we consider a regular\footnote{This means that the Hessian of the Lagrangian is weakly non degenerate.} local field theory $S_M\colon F_M\to\R$ the partition function in the path integral approach is
\begin{equation}
\label{state}
\psi_M=\int_{\phi\in F_M}\ee^{\frac{\I}{\hbar}S_M(\phi)}\mathscr{D}\phi.
\end{equation}
Usually, $F_M$ is infinite-dimensional, and one cannot define\footnote{Only in special situations, i.e. $\dim M = 1$, and some examples discussed in \cite{GJ}.} $\mathscr{D}\phi$. The way out is usually to translate the formal asymptotics as $\hbar \to 0$ of finite-dimensional integrals to the infinite-dimensional case. The terms in the asymptotic expansion are convenienetly labeled by Feynman diagrams \cite{Feynman1949,Feynman1950,P}. If the critical points of the action functional $S_M$ are degenerate, one needs to gauge-fix the theory before one can use the formal asymptotics .The most powerful gauge fixing formalism is the BV formalism. We briefly review its finite-dimensional version. Further references for gauge theories, different gauge fixing formalisms (including BV) and their perturbative quantization are
\cite{Mn,Mn2,R}. \\
The start is the following definition: 

\begin{defn}[BV manifold]
A \textsf{BV manifold} is a triple  $(\calF,\omega,\calS)$, where $\calF$ is a supermanifold with $\mathbb{Z}$-grading, $\omega$ an odd symplectic form of degree $-1$ on $\calF$, and $\calS$ is an even function of degree zero on $\calF$, such that 
\begin{equation}
(\calS,\calS) = 0.
\end{equation}
\end{defn}
Here, following Batalin and Vilkovisky (\cite{BV1,BV2}), we denote the Poisson bracket induced by the odd symplectic form with round brackets $(\enspace,\enspace)$.  
\begin{rem}[Grading on $\calF$]
Note that we have two different gradings on $\calF$, the $\Z_2$-grading from the supermanifold structure and an additional $\Z$-grading. In phyics, the $\Z$-grading is referred to as \textsf{ghost number} and the parity corresponds to bosonic and fermionic particles. Since we consider only bosonic theories, the $\Z_2$-grading coincides with the reduction of the $\Z$-grading.
\end{rem}
In a Darboux chart $(q^i,p_i)$, we can define the \textsf{BV Laplacian} by
\[
\Delta^{\text{loc}}=\sum_i(-1)^{\vert q^{i}\vert}\frac{\partial^2}{\partial q^{i}\partial p_i}.
\]
Then we get that $(\Delta^{\text{loc}})^2=0$ and for two functions $f,g$, $\Delta^{\text{loc}}(fg)=\Delta^{\text{loc}} fg\pm f\Delta^{\text{loc}} g\pm (f,g)$. 
This extends to a well-defined global operator $\Delta$ on half-densities (see \cite{Khudaverdian2004, Severa2006}). 

Moreover, given a half-density $f$ and a Lagrangian submanifold $\calL \subset \calF$, we can define a \textsf{BV integral} $
\int_{\mathcal{L}}f $
by restricting the half-density to the Lagrangian where it becomes a density and can be integrated. The main result in the Batalin--Vilkovisky formalism is the following Theorem. 
\begin{thm}[Batalin--Vilkovisky \cite{BV1}]
\label{thm:BV}
If we assume that the integrals converge, then
\begin{itemize} 
\item{If $f=\Delta g$, then $\int_{\mathcal{L}}f=0$,
}
\item{If $\Delta f=0$ and $(\calL_t)$ is a smoothly varying family of Lagrangians, then $\frac{\dd}{\dd t}\int_{\mathcal{L}_t}f = 0$.
}
\end{itemize}
\end{thm}

\begin{rem}
The second point of Theorem \ref{thm:BV} tells us that if we would have an ill-defined integral $\int_{\calL_0}f$ for some Lagrangian submanifold $\calL_0$, but we know that $\Delta f=0$, then we can define the value of the integral by a well-defined one $\int_{\calL_1}f$ for some Lagrangian submanifold $\calL_1$, and this does not depend on the choice of $\calL_1$ as long as we deform it continuously.
\end{rem}

This procedure is called \textsf{gauge-fixing}. This construction can be extended to any (super)manifold. Moreover, considering $f=\ee^{\frac{\I}{\hbar}S}$, two other conditions arise, which are the Master Equations for the classical and quantum level:
\begin{align}
(S,S)&=0,\tag{Classical Master Equation (CME)}\\
(S,S)-2\I\hbar\Delta S&=0.\tag{Quantum Master Equation (QME)}
\end{align}
The latter is equivalent to $\Delta \ee^{\frac{\I}{\hbar}S} = 0$. The former is the classical limit of the latter for $\hbar\to 0$, and motivates the definition of BV manifold as given above. 

\subsection{Classical BV-BFV formalism}
\label{classicalBV-BFV}
We now turn to the infinite-dimensional case and review the main definitions of references \cite{CMR1}. 
We first consider the classical BV formalism in field theory and its extension to manifolds with boundary.
\begin{defn}[BV theory]
A $d$-dimensional \textsf{BV theory} is the association of a BV manifold  $M\mapsto (\calF_M,\omega_M,\calS_M)$ to every closed $d$-manifold $M$. 
\end{defn}
\begin{rem}
These BV manifolds are typically infinite-dimensional. This means that neither the BV Laplacian nor the BV integral are defined (at least not without further work). 
\end{rem}
\begin{defn}[BV extension]
We say that a BV theory $(\calF_M,\omega_M,\calS_M)$ is a \textsf{BV extension} of a local field theory $M \mapsto (F_M,S_M)$ if for all closed $d$-manifolds $M$, we have that the degree 0 part $(\calF_M)_0$ of $\calF_M$ satisfies $(\calF_M)_0 = F_M$ and $\calS_M\big|_{(\calF_M)_0} = S_M$. Moreover, we want $\calF_M,\calS_M$ and $\omega_M$ to be local.
\end{defn}
To extend the BV formalism to manifolds with boundary one needs its Hamiltonian counterpart, the BFV formalism \cite{BF1,BF2, FV1, FV2}.

\begin{defn}[BFV manifold]
A \textsf{BFV manifold} is a triple \begin{equation}\calF^\partial=(\mathcal{F}^\partial,\omega^\partial,Q^\partial) \end{equation} where $\mathcal{F}^\partial$ is a graded manifold, $\omega$ an even symplectic form of degree $0$, and $Q^\partial$ a degree $1$ cohomological, symplectic vector field on $\calF^\partial$. If $\omega^\partial=\delta\alpha^\partial$ is exact, 
the BFV manifold is called \textsf{exact}.
\end{defn}
%
 Again, we denote by $\delta$ the de Rham differential on the space of fields. 
The notion of BV theory can be extended to manifolds with boundary as was shown in \cite{CMR1,CMR2}. On the boundary we will use the BFV formalism. The compatibility between the BV formalism and the BFV formalism is captured in the following definition. \begin{defn}[BV-BFV manifold]
 A \textsf{BV-BFV manifold} over a given exact BFV manifold $\calF^\partial=(\mathcal{F}^\partial,\omega^\partial=\delta \alpha^\partial,Q^\partial)$ is a quintuple \begin{equation}
 \calF = (\calF,\omega,\calS,Q,\pi),\end{equation} where \begin{itemize} 
 \item  $\calF$ is a graded manifold, 
 \item $\omega$ is an even symplectic form of degree $0$, 
 \item $\calS$ is an even function of degree $0$, 
 \item $Q$ is a degree $1$ cohomological vector field, 
 \item  $\pi\colon \calF \to \calF^\partial$ is a surjective submersion
 \end{itemize} such that \begin{equation}
 \label{mCME0}
 \iota_{Q}\omega=\delta \calS+\pi^*\alpha^\partial
 \end{equation}and $Q^\partial=\delta\pi Q$ where $\delta\pi$ denotes the differential of $\pi$. \end{defn}
 \begin{rem} 
\label{point} 
If $\calF^\partial$ is a point, we get that $(\calF_M,\omega_M,\calS_M)$ is a BV manifold. The shorthand notation for a BV-BFV manifold is $\pi\colon\calF \to \calF^\de$
\end{rem}

Note that by Remark \ref{point}, the following notion generalizes the one of a BV theory.
\begin{defn}[BV-BFV theory]
\label{defn:BVBFV_theory}
A $d$-dimensional \textsf{BV-BFV theory} associates to every closed $(d-1)$-dimensional manifold $\Sigma$ a BFV manifold $\calF^\partial_\Sigma$, and to a $d$-dimensional manifold $M$ with boundary $\de M$ a BV-BFV manifold $\pi_M\colon \calF_M \to \calF^\de_{\de M}$.
\end{defn}

\begin{rem}
Formally, for the Hamiltonian vector field $Q$ of $\calS$, one can write $(\calS,\calS)=\iota_Q\iota_Q\omega=Q(\calS)$. If we consider a BV-BFV theory for a manifold $M$ with boundary $\de M$, we get that 
$$Q(\calS)=\pi^*(2\calS^\de-\iota_{Q^\de}\alpha^\de).$$
Equivalently, we get
\begin{equation}
\label{mCME}
\iota_Q\iota_Q\omega=2\pi^*\calS^\de.
\end{equation}
We call \eqref{mCME} the \textsf{modified Classical Master Equation (mCME)}.

\end{rem}

It was shown in \cite{CMR1} that \textsf{abelian $BF$ theory} is an example of a BV-BFV theory. 
\begin{ex}[Abelian $BF$ theory]
\textsf{Abelian $BF$ theory} is given by the following data: 
\begin{align*}
\calF_M &= \Omega^\bullet(M)[1]\oplus \Omega^\bullet(M)[d-2]\ni \mathsf{X}\oplus \boldsymbol{\eta}\\
\omega_M&=\int_M\delta\mathsf{X}\land \delta\boldsymbol{\eta}\\
\calS_M&=\int_M \boldsymbol{\eta}\land\dd\mathsf{X}\\
Q_M&=(-1)^d\int_M\left(\dd\boldsymbol{\eta}\land\frac{\delta}{\delta\boldsymbol{\eta}}+\dd\mathsf{X}\land\frac{\delta}{\delta\mathsf{X}}\right)
\end{align*}
\end{ex}

\begin{defn}[$BF$-like theories]
We say that a BV-BFV theory is \textsf{$BF$-like} if
\begin{align}
\calF_M &= (\Omega^\bullet(M)\otimes V[1])\oplus (\Omega^\bullet (M)\otimes V^*[d-2])\\
\calS_M &= \int_M\left(\langle\boldsymbol{\eta},\dd\mathsf{X}\rangle+\calV(\mathsf{X},\boldsymbol{\eta})\right),
\end{align}
where $V$ is a graded vector space, $\langle\enspace,\enspace\rangle$ denotes the pairing between $V^*$ and $V$, and $\calV$ denotes some density-valued function of the fields $\mathsf{X}$ and $\boldsymbol{\eta}$,
such that $\calS_M$ satisfies the Classical Master Equation for $M$ without boundary.
\end{defn}

\begin{ex}[Quantum mechanics]
Consider $M$ to be a $1$-dimensional manifold, i.e. $d=1$ and $V=W[-1]$ with $W$ concentrated in degree zero. Denote by $P$ and $Q$ the degree-zero form components of $\mathsf{X}$ and $\boldsymbol{\eta}$, respectively. Choose a volume form $\dd t$ on $M$ and a function $H$ on $T^*W$. Set $\calV(\mathsf{X},\boldsymbol{\eta}):=H(\mathsf{X},\boldsymbol{\eta})\dd t=H(Q,P)\dd t$. Then 
\begin{equation}
\calS_M=\int_M\left(\sum_i P_i\dot{Q}^{i}+H(Q,P)\right)\dd t,
\end{equation}
is the action of classical 
mechanics in the Hamiltonian formalism.
\end{ex}

\begin{ex}[$BF$-like AKSZ theories \cite{AKSZ}]
Assume we are given a function $\Theta$ on $T^*[d-1](V[1])=V[1]\oplus V^*[d-2]$ that is of degree $d$ such that $\{\Theta,\Theta\}=0$, where $\{\enspace,\enspace\}$ is the canonical Poisson structure on the shifted cotangent bundle. Set $\calV(\mathsf{X},\boldsymbol{\eta})$ to be the top degree part of $\Theta(\mathsf{X},\boldsymbol{\eta})$. 
\end{ex}


\subsection{Quantum BV-BFV formalism}
In \cite{CMR2} the notion of a quantum BV-BFV theory was given and it was shown how to perturbatively quantize a classical BV-BFV theory\footnote{We have to assume certain condtions which are in particular satisfied for $BF$-like theories}. Let us briefly review this\footnote{We slighty changed the definition of quantum BV-BFV theory so that in principle it does not depend on a classical BV-BFV theory.}. 
\begin{defn}[Quantum BV-BFV theory]
A $d$-dimensional \textsf{quantum BV-BFV theory} associates 
\begin{itemize}
\item To every closed $(d-1)$-dimensional manifold $\Sigma$ a graded $\C[[\hbar]]$-module $\mathcal{H}_{\Sigma}$, 
\item To every $d$-dimensional manifold (possibly with boundary) $M$ a finite-dimensional BV manifold $\mathcal{V}_M$, a degree 1 coboundary operator $\Omega_{\de M}$ on $\mathcal{H}_{\de M}$  and a homogeneous element\footnote{Typically, $\psi$ will have degree 0. This is the case when the gauge-fixing Lagrangian (see below) has degree zero, in the sense that its Berezinian bundle has degree zero. This is the case in all examples we consider.} $$\psi_M \in \widehat{\mathcal{H}}_M := \mathrm{Dens}^{\frac12}(\mathcal{V}_M) \otimes \mathcal{H}_{\de M},$$ where $\mathrm{Dens}^\frac{1}{2}(\calV_M)$ denotes the space of half-densities on $\calV_M$,
\end{itemize}
such that 
\begin{equation}(\hbar^2\Delta_{\mathcal{V}_M} + \Omega_{\de M})\psi_M = 0.
\label{mQME}
\end{equation}
\end{defn}
The shorthand notation for a quantum BV-BFV theory is $M \mapsto (\widehat{\mathcal{H}}_M,\psi_M,\Delta_{\mathcal{V}_M},\Omega_{\de M})$.  Let us introduce some terminology: We call $\mathcal{V}_M$ the \textsf{space of residual fields}, $\mathcal{H}_{\de M}$ the \textsf{space of boundary states} and $\psi_M$ the \textsf{state}. $\Delta_{\mathcal{V}_M}$ denotes the canonical BV Laplacian on half-densities on the BV manifold $\mathcal{V}_M$. Recall that $\Delta_{\mathcal{V}_M}^2 =0$.
Hence, $\widehat{\mathcal{H}}_M$ carries the two commuting differentials $\Delta_{\mathcal{V}_M}$ and $\Omega_{\de M}$ which gives it the structure of a bicomplex.  We call $\Omega_{\de M}$ the quantum BFV boundary operator.
The condition \eqref{mQME} is called the \textsf{modified Quantum Master Equation}.
\begin{defn}[Equivalence]
We say that two quantum BV-BFV theories $(\widehat{\mathcal{H}}_M,\Delta_{\mathcal{V}_M},\Omega_{\de M},\psi_M)$ and $(\widehat{\mathcal{H}}_M',\Delta_{\mathcal{V}'_M},\Omega'_{\de M},\psi'M)$ are equivalent if  for every manifold $M$ with boundary $\partial M$ there is a quasi-isomorphism of bicomplexes \begin{equation}I_M\colon (\widehat{\mathcal{H}}_M,\Delta_{\mathcal{V}_M},\Omega_{\de M}) \to  (\widehat{\mathcal{H}}_M',\Delta_{\mathcal{V}'_M},\Omega'_{\de M})\end{equation}
such that $I_M(\psi_M)= \psi'_M$.
\end{defn}
\begin{defn}[Change of data]
\label{change_of_data}
We say that two quantum BV-BFV theories $(\widehat{\mathcal{H}}_M,\Delta_{\mathcal{V}_M},\Omega_{\de M},\psi_M)$ and $(\widehat{\mathcal{H}}_M,\Delta_{\mathcal{V}'_M},\Omega'_{\de M},\psi'_M)$
are related by \textsf{change of data} if there is an operator $\tau$ of degree 0 on $\mathcal{H}_{\de M}$ and an element $\chi \in \widehat{\mathcal{H}}_M$ with $\deg(\chi)=\deg(\psi)-1$ such that
\begin{align}
\begin{split}
\Omega'_{\de M} &= [\Omega_{\de M},\tau] \\
\psi'_M &=  (\hbar^2\Delta_{\calV_M} + \Omega_{\de M})\chi_M - \tau \psi_M
\end{split}
\end{align} 
\end{defn}
Let us now explain how to produce a quantum BV-BFV theory by perturbative quantization of a classical BV-BFV theory. Fix a classical BV-BFV theory $\pi\colon \calF \to \calF^\de$. For simplicity we shall assume that $\calF$ and $\calF^\de$ are always \emph{vector spaces}, which is sufficient for the present paper. For a general discussion see \cite{CMR2}.
\subsubsection{The space of states}
Consider a $(d-1)$-dimensional manifold $\Sigma$. Then the BV-BFV theory associates a symplectic vector space $(\calF^\de_\Sigma,\omega^\de_\Sigma,Q_\Sigma^\de)$.   Morally, we want to construct $\mathcal{H}_{\Sigma}$ and $\Omega_{\Sigma}$ as a geometric quantization of this symplectic vector space. More precisely, the construction proceeds as follows.
We require the data of a polarization\footnote{We have only considered the case of real polarizations so far.} $\calP$ of this symplectic vector space. For our purposes, a splitting 
\begin{equation}
\calF^\de_\Sigma = \mathcal{B}^\calP_{\Sigma} \oplus \mathcal{K}^\calP_{\Sigma}
\end{equation}
of $\calF^\de_\Sigma$ into Lagrangian subspaces is sufficient. Here $\mathcal{K}^\calP_{\Sigma}$ is thought of as the Lagrangian distribution on $\calF^\de_\Sigma$ and $\mathcal{B}^\calP_\Sigma$ is identified with the leaf space of the polarization. Given a polarization $\mathcal{P}$ the associated space of states $\calH_{\de M}$ is a certain space of functionals  on $\mathcal{B}^\calP_\Sigma$. We will discuss the space of states for $BF$-like theories in \ref{sec:BF_like}.

\subsubsection{Splitting the space of fields} To define the quantum state we proceed with the following constructions. Consider a $d$-manifold $M$ (possibly with boundary) and the associated BV-BFV manifold $(\calF_M,\omega_M,\calS_M,Q_M,\pi_M)$ over the exact BFV manifold $(\calF^\de_{\de M},\omega^\de_{\de M}=\delta\alpha^\de_{\de M},Q^\de_{\de M})$. Then, choosing a polarization $\calP$ on $\de M$, we choose a splitting
\begin{equation}
\calF_M\cong\calB_{\de M}^\calP\oplus\calY,
\end{equation}
where $\calY$ denotes some complement. This splitting is subject to the following assumption\footnote{This assumption forces one to choose singular extensions of boundary fields}. 
\begin{ass}[\cite{CMR2}]\label{ass:singularsplitting}
There is a weakly symplectic form $\omega_\calY$ on $\calY$ such that $\omega_M$ is the extension of $\omega_\calY$ to $\calF_M$. 
\end{ass} Formally, we can think of $\calB^\calP_{\de M}$ as the space of \textsf{boundary fields} and $\calY$ the space of \textsf{bulk fields}. Depending on the boundary polarization, we split $\calY$ into residual fields and some complement, i.e. we choose a splitting 
\begin{equation} 
\calY=\calV^\calP_M\oplus \calY'
\end{equation}
subject to the following assumption\footnote{This assumption is rather strong but can be slightly relaxed to the notion of \emph{hedgehog fibration}.} 
\begin{ass}\label{ass:symplecticproduct}
We assume the following hold:
\begin{enumerate}
\item $\calV^\calP_M, \calY'$ are BV manifolds, 
\item $\calV^\calP_M$ is finite-dimensional
\item $\omega_\calY = \omega_{\calV^\calP_M} + \omega_{\calY'}$.
\end{enumerate}
\end{ass}
We call the complement $\calY'$ the space of \textsf{fluctuation fields}. 
Residual fields are also called \textsf{low energy fields} or \textsf{slow fields} and fluctuation fields are also called \textsf{high energy fields} or \textsf{fast fields}.  Typically we choose $\calV^\calP_M$ as the solutions of $\delta \calS_M^{0}=0$ modulo gauge transformations, where $\calS_M^0$ denotes the quadratic part of the action $\calS_M$. This is the minimal choice, and is typically called the space of \textsf{zero modes}. Other choices are related by the equivalence relations above. 
\begin{defn}
A splitting 
\begin{equation}
\label{split}
\calF_M \cong \calB_{\de M}^\calP \oplus \calV^\calP_M\oplus \calY'
\end{equation}
is called good if it satisfies Assumptions \ref{ass:singularsplitting} and \ref{ass:symplecticproduct}
\end{defn}
\begin{rem}[Connection to Atiyah's TQFT formulation]
From the point of view of topological quantum field theories (TQFTs) as functors $\textbf{Cob}_n\to \textbf{Vect}_\C$ from the $n$-cobordism category (objects are $(n-1)$-manifolds bounding an $n$-manifold and morphisms are exactly the bounding $n$-manifolds connecting the objects) to the category of vector spaces over the complex numbers, it is clear that the quantum state should depend on the bulk. This can be seen by using the fact that the state represents exactly the bounding manifold between the objects and thus a morphism of the cobordism category. This also makes sense for manifolds without boundary, in which case the state is given by a partition function $Z\colon \C\to\C$, where as a morphism in $\textbf{Cob}_n$ it represents any closed $n$-manifold, seen as a bounding manifold connecting the empty $(n-1)$-manifold, i.e. as a morphism $\varnothing\to\varnothing$.
\end{rem}

\subsubsection{The quantum state in $BF$-like theories}
\label{sec:BF_like}
The quantum state in $BF$-like theories is defined perturbatively in terms of Feynman graphs by considering integrals defined on the configuration space of these graphs. In $BF$-like theories there are two preferred polarizations, namely the $\frac{\delta}{\delta \mathsf{X}}$- and $\frac{\delta}{\delta\boldsymbol{\eta}}$-polarization. We specify a polarization by splitting the boundary $\de M$ of the manifold $M$ into two parts $\de_1M$ and $\de_2M$, where we choose the $\frac{\delta}{\delta\boldsymbol{\eta}}$-polarization on $\de_1M$ and the $\frac{\delta}{\delta\mathsf{X}}$-polarization on $\de_2M$. We denote the $\mathsf{X}$-leaf by $\mathbb{X}\in\calB^{\frac{\delta}{\delta\boldsymbol{\eta}}}_{\de M}$ and the $\boldsymbol{\eta}$-leaf by $\E\in \calB^\frac{\delta}{\delta\mathsf{X}}_{\de M}$.

For $BF$-like theories, the polarization determines the first splitting as 
\begin{align*}
\calB^\calP_{\de M} &= (\Omega^\bullet(\de_1M)\otimes V[1]) \oplus (\Omega^\bullet(\de_2M) \otimes V^*[d-2]) \\
\calY &= (\Omega^{\bullet}(M,\de_1M)\otimes V[1]) \oplus (\Omega^\bullet(M,\de_2M)\otimes V^*[d-2])
\end{align*} 
The minimal space of residual fields is isomorphic to
\begin{equation}
\calV^\calP_M\cong(H^\bullet(M,\de_1M)\otimes V[1] )\oplus (H^\bullet(M,\de_2M)\otimes V^*[d-2]),
\end{equation}
for some graded vector space $V$. A good splitting is then determined by an splitting of the complex of de Rham forms with relative boundary conditions into a subspace $\calV^\calP_M$ isomorphic to cohomology and a complement $\calY'$ in a way compatible with the symplectic structure. One possibility to do so is to use a Riemannian metric and embed the cohomology as harmonic forms. \\
Before we can introduce the quantum state we need to introduce 
the concept of \textsf{composite fields}, which we denote by square brackets $[\enspace]$, e.g. for a boundary field $\mathbb{A}$ we will write $[\mathbb{A}^{i_1}\dotsm \mathbb{A}^{i_k}]$. They can be understood as a \textsf{regularization} of higher functional derivatives: the higher functional derivative $\frac{\delta^k}{\delta\mathbb{A}^{i_1}\dotsm \delta\mathbb{A}^{i_k}}$ gets replaced by a first order functional derivative $\frac{\delta}{\delta[\mathbb{A}^{i_1}\dotsm \mathbb{A}^{i_k}]}$. Concretely, this corresponds to introducing additional boundary vertices as in Figure \ref{fig:composite_field_vertices}. 

\begin{rem}
In fact, this concept will not be needed for the definition of the principal part of the quantum state.
We will use this concept to define the full part of the quantum state where we need to make sure that it will be compatible with the quantum BFV boundary operator, where higher functional derivatives do indeed appear as we will see.
\end{rem}

\begin{defn}[Regular functional]
\label{regular_fun}
A \textsf{regular functional} on the space of base boundary fields is a linear combination of expressions of the form 
\begin{equation}
\label{regular_functional}
\int_{\mathsf{C}_{m_1}(\de_1M)\times \mathsf{C}_{m_2}(\de_2M)}L^{J_1^1...J_1^{\ell_1}J_2...J_2^{\ell_2}...}_{I_1^1....I_1^{r_1}I_2^1...I_2^{r_2}...}\land \pi_1^*\prod_{j=1}^{r_1}\left[\mathbb{X}^{I_1^j}\right]\land\dotsm \land \pi_{m_1}^*\prod_{j=1}^{r_{m_1}}\left[\mathbb{X}^{I_{m_1}^j}\right]\land\pi_1^*\prod_{j=1}^{\ell_1}\left[\mathbb{E}_{J_1^j}\right]\land\dotsm\land \pi_{m_1}^*\prod_{j=1}^{\ell_{m_2}}\left[\mathbb{E}_{J_{m_2}^j}\right],
\end{equation}
where $I_i^j$ and $J_i^j$ are (target) multi-indices and $L^{J_1^1...J_1^{\ell_1}J_2...J_2^{\ell_2}...}_{I_1^1....I_1^{r_1}I_2^1...I_2^{r_2}...}$ is a smooth differential form on the direct product of compactified configuration spaces (see Appendix \ref{app:Conf}) $\mathsf{C}_{m_1}(\de_1M)\times \mathsf{C}_{m_2}(\de_2M)$ depending on residual fields. A regular functional is called \textsf{principal} if all multi-indices have length one.
\end{defn}

\begin{defn}[Full space of boundary states]
\label{space_of_boundary_states}
The \textsf{full space of boundary states} $\calH^\calP_{\de M}$ is given by the linear combinations of regular functionals of the form \eqref{regular_functional}.
\end{defn}

\begin{defn}[Principal space of boundary states]
We define the \textsf{principal space of boundary states} $\calH^{\calP,\textnormal{princ}}_{\de M}$ as the subspace of $\calH^\calP_{\de M}$, where we only consider principal regular functionals.
\end{defn}
The state is defined in terms of Feynman graphs and rules. We briefly explain what these terms mean in the BV-BFV context (for perturbations of abelian $BF$ theory).
\begin{defn}[($BF$) Feynman graph]
A \emph{($BF$) Feynman graph} is an oriented graph with three types of vertices $V(\Gamma) = V_{bulk}(\Gamma) \sqcup V_{\de_1} \sqcup V_{\de_2}$, called bulk vertices and type 1 and 2 boundary vertices, such that 
\begin{itemize}
\item bulk vertices can have any valence, 
\item type 1 boundary vertices carry any number of incoming half-edges (and no outgoing half-edges), 
\item type 2 boundary vertices carry any number of outgoing half-edges (and no incoming half-edges),
\item multiple edges and loose half-edges (leaves) are allowed but not short loops (tadpoles).
\end{itemize}
A \emph{labeling} of a Feynman graph is a function from the set of half-edges to $\{1,\ldots,\dim V\}$.
\end{defn}
\begin{defn}[Principal graph]
A Feynman graph is called \emph{principal} if all boundary vertices (type 1 and type 2) are univalent or zero valent.  
\end{defn}
For a set $S$ and a manifold $M$, the open configuration space of $S$ in $M$ is 
$$\mathsf{Conf}_S(M) := \{\iota\colon S\hookrightarrow M|\iota \text{ injection}\}.$$ 
Let $\Gamma$ be a Feynman graph and $M$ a manifold with boundary $\de M = \de_1 M \sqcup \de_2M$ and denote 
\begin{equation}
\mathsf{Conf}_{\Gamma}(M) := \mathsf{Conf}_{V_{bulk}}(M) \times \mathsf{Conf}_{V_{\de_1}}(\de_1M) \times \mathsf{Conf}_{V_{\de_2}}(\de_2M)
\end{equation}
The Feynman rules are a map that associate to a Feynman graph $\Gamma$ a differential form $\omega_{\Gamma} \in \Omega^\bullet(\mathsf{Conf}_\Gamma(M))$. 
\begin{defn}[($BF$) Feynman rules]
Let $\Gamma$ be a labeled Feynman graph, and choose a configuration $\iota\colon V(\Gamma) \to \mathsf{Conf}(\Gamma)$ (that respects the decompositions). We decorate the graph according to the following rules (called \emph{Feynman rules}):
\begin{itemize}
\item{Bulk vertices in $M$ decorated by ``vertex tensors'' $$\calV^{j_1\ldots j_t}_{i_1 \ldots i_s}:=\frac{\de^{s+t}}{\de \mathsf{X}_{i_1}\dotsm \de \mathsf{X}_{i_s}\de\boldsymbol{\eta}^{j_1}\dotsm \de\boldsymbol{\eta}^{j_t}}\big\vert_{\mathsf{X}=\boldsymbol{\eta}=0}\calV(\mathsf{X},\boldsymbol{\eta}),$$ where $s,t$ are the out- and in- valencies of the vertex and $i_1,\ldots,i_s$ and $j_1,\ldots,j_t$ are the labels of the out (resp. in-)oriented half-edges.
}
\item{ Boundary vertices $v \in V_{\de_1}(\Gamma)$ with  incoming half-edges labeled $i_1,\ldots,i_k$ and no out-going half-edges are decorated by a composite field $[\mathbb{X}_{i_1}\ldots\mathbb{X}_{i_k}]$ evaluated at the point (vertex location) $\iota(v)$ on $\partial_1M$.
}
\item{Boundary vertices $v \in V_{\de_2}(\Gamma)$ on $\partial_2M$ with  outgoing half-edges labeled $j_1,\ldots,j_l$ are decorated by $[\E^{j_1}\ldots\E^{j_l}]$ evaluated at the point on $\de_2M$.
}
\item{Edges between vertices $v_1,v_2$ are decorated with the propagator $\zeta(\iota(v_1),\iota(v_2))\cdot \delta_j^{i}$, where $\zeta$ is the propagator induced by $\calL\subset\calY'$, the chosen gauge-fixing Lagrangian.
}
\item{Loose half-edges (leaves) attached to a vertex $v$ and labeled $i$ are decorated with the residual fields $\mathsf{x}_i$ (for out-orientation), $\mathsf{e}^{i}$ (for in-orientation) evaluated at the point $\iota(v)$. 
}
\end{itemize}
We denote the differential forms given by the decorations collectively by $\omega_d$. The differential form $\omega_{\Gamma}$ at $\iota$ is then defined by multiplying all decorations and summing over all labelings: 
\begin{equation}
\omega_{\Gamma} = \sum_{\text{labelings}\atop \text{of }\Gamma}\prod_{\text{decorations}\atop \text{$d$ of } \Gamma} \omega_d
\end{equation}
\end{defn}
The Feynman rules are summarized in Figures \ref{fig:FeynmanRules0} and \ref{fig:composite_field_vertices}.
\begin{rem}[Configuration spaces]
We will work with the Fulton--MacPherson/Axelrod--Singer compactification of configuration spaces on manifolds with boundary and corners (FMAS compactification, see Appendix \ref{app:Conf}). It is a non-trivial analytic statement (proven first by Axelrod and Singer \cite{AS2}) that the propagator, \emph{a priori} defined only on the open configuration space $\mathsf{Conf}_2(M)$, extends to the  compactification $\mathsf{C}_2(M)$. It follows that also $\omega_\Gamma$, for all Feynman graphs $\Gamma$, extends to the compactification $\mathsf{C_{\Gamma}}(M)$ of $\mathsf{Conf}_\Gamma(M)$. Since integrals remain unchanged by adding strata of lower codimension, this immediately proves that all integrals in Equation \eqref{eq:def_state_pert} below are finite. Moreover, the combinatorics of the stratification can be used for various computations using Stokes' theorem. 
\end{rem}
\begin{defn}[Principal quantum state]
\label{principal_part}
Let $M$ be a manifold, possibly with boundary. 
Given a $BF$-like BV-BFV theory $\pi_M\colon\calF_M \to \calF^\de_{\de M}$, a polarization $\calP$ on $\calF^\de_{\de M}$, a good splitting $\calF_M = \mathcal{B}^\calP_{\de M}\oplus \calV^\calP_M \oplus \calY'$, and a gauge-fixing Lagrangian $\calL \subset \calY'$, we define the \textsf{principal part of the quantum state} by the  formal power series 
\begin{equation}
\psi_M(\mathbb{X},\E;\mathsf{x},\mathsf{e}):=T_M\exp\left(\frac{\I}{\hbar}\sum_{\Gamma}\frac{(-\I\hbar)^{\textnormal{loops}(\Gamma)}}{\vert\textnormal{Aut}(\Gamma)\vert}\int_{\mathsf{C}_\Gamma(M)}\omega_\Gamma(\mathbb{X},\E;\mathsf{x},\mathsf{e})\right),\label{eq:def_state_pert}
\end{equation}
where we denote for an element $\mathsf{X}\oplus \boldsymbol{\eta}\in \calF_M$ the split by 
\begin{align}
\mathsf{X}&= \mathbb{X}\oplus \mathsf{x}\oplus \mathscr{X},\\
\boldsymbol{\eta}&=\E\oplus \mathsf{e}\oplus \mathscr{E}.
\end{align}
Here the sum is taken over all \text{connected}, oriented, \emph{principal} BF Feynman graphs $\Gamma$, $\textnormal{Aut}(\Gamma)$ denotes the set of all automorphisms of $\Gamma$, and $\textnormal{loops}(\Gamma)$ denotes the number of all loops of $\Gamma$. 
\end{defn}
  The coefficient $T_M$ is related to the Reidemeister torsion of $M$, but its precise nature is irrelevant for the purpose of a present paper. For a definition see \cite{CMR1}.
  \begin{rem} The formal power series \eqref{eq:def_state_pert} is our definition of the  
  formal perturbative expansion of the BV integral
\begin{equation}
\psi_M = \int_{\calL \subset \calY'}\ee^{\frac{\I}{\hbar}\calS_M[(\mathsf{X},\boldsymbol{\eta})]} \in \Hat{\calH}_M^\calP:=\Hat{\calH}_{\de M}^\calP\otimes \textnormal{Dens}^\frac{1}{2}(\calV_M^\calP).\end{equation}
It was observed in \cite{CMR2} that, given a good splitting of the form \eqref{split}, one can decompose the action as  
\begin{equation}
\calS^\calP_M = \widehat{\calS}_{M,0} + \widehat{\calS}_{M,\textnormal{pert}} + \calS^\textnormal{res} + \calS^{\textnormal{source}}
\end{equation} 
with 
\begin{align*}
\widehat{\calS}_{M,0} &= \int_M \langle\mathscr{E}, \dr \mathscr{X}\rangle\\
 \widehat{\calS}_{M,\textnormal{pert}} &= \int_M \calV(\mathscr{X},\mathscr{E}) \\
 \calS^\textnormal{res} &= (-1)^{d-1}\left(\int_{\de_1 M}\langle \E, \mathsf{x}\rangle +\int_{\de_2 M}\langle \mathbb{X}, \mathsf{e}\rangle \right) \\
 \calS^{\textnormal{source}} &= (-1)^{d-1}\left(\int_{\de_1 M}\langle \E, \mathscr{X}\rangle +\int_{\de_2 M}\langle \mathbb{X}, \mathscr{E}\rangle \right)
\end{align*}
In that way we can rewrite 
$$\psi_M = T_M\left \langle \ee^{\frac{\ii}{\hbar}(\calS^\textnormal{res} + \calS^{\textnormal{source}})}\right\rangle$$ 
where $\langle \enspace \rangle$ denotes the expectation value with respect to the bulk theory ($\widehat{\calS}_{M,0} + \widehat{\calS}_{M,\textnormal{pert}}$). 
\end{rem}

\begin{rem}
Note that we sum over connected graphs, such that the sum is given by the \textsf{effective action}. 
\end{rem}

%

\begin{figure}
\centering
\subfigure[Interaction vertex]{
\centering
\begin{tikzpicture}
\node[vertex] (o) at (0,0) {};
\node[coordinate, label=below:{$i_1$}] at (30:1) {$i_1$}
edge[fermion] (o);
\node[coordinate, label=above:{$i_2$}] at (60:1) {$i_2$}
edge[fermion] (o);
\node[coordinate, label=below:{$i_s$}] at (145:1) {$i_s$}
edge[fermion] (o);
\draw[dotted] (90:0.5) arc (90:130:0.5);
\node[coordinate, label=below:{$j_1$}] (j1) at (-30:1) {};
\node[coordinate, label=below:{$j_2$}] (j2) at (-60:1) {};
\node[coordinate, label=below:{$j_t$}] (j3) at (-145:1) {};
\draw[dotted] (-90:0.5) arc (-90:-145:0.5); 
\draw[fermion] (o) -- (j1);
\draw[fermion] (o) -- (j2);
\draw[fermion] (o) -- (j3);
\node[coordinate,label=right:{$\leadsto\quad\calV^{j_1\ldots j_t}_{i_1 \ldots i_s}$}] at (2,0) {};
\end{tikzpicture}
}
%
\hspace{2cm}
\subfigure[Residual fields]{
\centering
\begin{tikzpicture}
\node[circle,draw,inner sep=1pt] (x) at (0,1) {$\mathsf{x}^i$};
\node[coordinate, label=below:{$i$}] (x2) at (30:2) {$i_1$}
edge[fermion] (x);
\node[coordinate, label=below:{$j$}] (e2)at (-30:2) {$i_1$};
\node[circle,draw,inner sep=1pt] (e) at (0,-1) {$\mathsf{e}_j$}
edge[fermion] (e2);
\node[coordinate, label={${}$}] at (4,0) {};
\end{tikzpicture}
}
\hspace{3cm}
\subfigure[Boundary vertices]{
\begin{tikzpicture}
\draw (-1,1) -- (1,1);
\node[vertex, label=above:{$\mathbb{X}$}] (x) at (0,1) {};
\node[coordinate] (b1) at (0,0) {}; 
\draw[fermion] (x) -- (b1);
\draw (1,-1) -- (3,-1);
\node[vertex, label=below:{$\mathbb{E}$}] (e) at (2,-1) {};
\node[coordinate] (b2) at (2,0) {}; 
\draw[fermion] (b2) -- (e);

\end{tikzpicture}
\hspace{2cm}
}
\caption{Summary of Feynman graphs and rules}\label{fig:FeynmanRules0}
\end{figure}
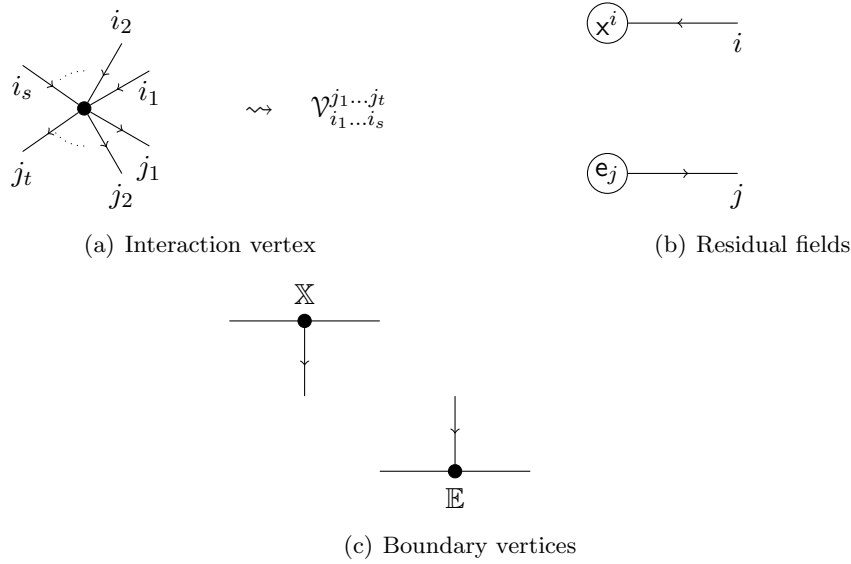

\begin{figure}[h!]
\subfigure[Boundary vertex on $\de_1\Sigma$] {
\begin{tikzpicture}
\draw (-3,0) -- (3,0); 
\node[vertex] (o) at (0,0) {};
\node[coordinate, label=below:{$[\mathbb{X}^{i_1}\cdots\mathbb{X}^{i_k}]$}] at (o.south) {};
\node[coordinate, label=below:{$i_1$}] (b1) at (30:1) {$i_1$};
\node[coordinate, label=above:{$i_2$}] (b2) at (60:1) {$i_2$};
\node[coordinate, label=below:{$i_k$}] (b3) at (145:1) {$i_k$};
\draw[dotted] (90:0.5) arc (90:130:0.5); 
\draw[fermion] (o) -- (b1);
\draw[fermion] (o) -- (b2);
\draw[fermion] (o) -- (b3);
\end{tikzpicture}
}
\hspace{2cm}
\subfigure[Boundary vertex on $\de_2\Sigma$] {
\begin{tikzpicture}
\draw (-3,0) -- (3,0); 
\node[vertex] (o) at (0,0) {};
\node[coordinate, label=below:{$[\mathbb{E}_{i_1}\cdots\mathbb{E}_{i_k}]$}] at (o.south) {};
\node[coordinate, label=below:{$i_1$}] at (30:1) {$i_1$}
edge[fermion] (o);
\node[coordinate, label=above:{$i_2$}] at (60:1) {$i_2$}
edge[fermion] (o);
\node[coordinate, label=below:{$i_k$}] at (145:1) {$i_k$}
edge[fermion] (o);
\draw[dotted] (90:0.5) arc (90:130:0.5);
\end{tikzpicture}
}
\caption{Composite field vertices.}\label{fig:composite_field_vertices}
\end{figure}
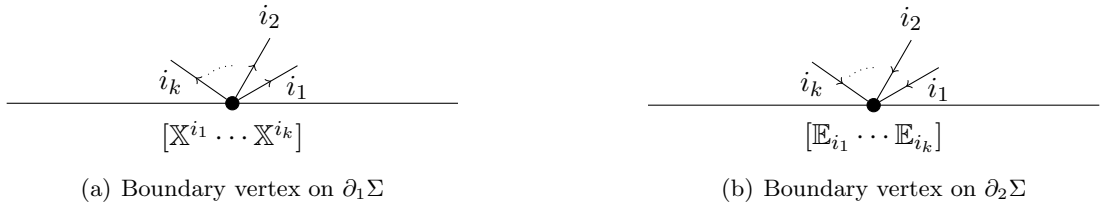 

We want to construct a product on the full state space using composite fields as in \cite{CMR2}. We construct the \textsf{bullet product} by 
\begin{multline}
\label{bullet_prod}
\int_{\de_1M}u_i\land \mathbb{X}^{i}\bullet \int_{\de_1M}v_j\land\mathbb{X}^j:=\\
(-1)^{\vert \mathbb{X}^{i}\vert (d-1+\vert v_j\vert)+\vert u_i\vert(d-1)}\left(\int_{\mathsf{C}_2(\de_1M)}\pi_1^*u_i\land \pi_2^*v_j\land \pi_1^*\mathbb{X}^{i}\land \pi_2^*\mathbb{X}^j+\int_{\de_1M}u_i\land v_j\land [\mathbb{X}^{i}\mathbb{X}^j]\right),
\end{multline}
where $u$ and $v$ are smooth differential forms depending on the bulk and residual fields.

\begin{rem}
Consider an operator $\int_{\de_1M}F^{ij}\frac{\delta^2}{\delta\mathbb{X}^{i}\delta\mathbb{X}^j}$. Such an operator can be interprated as $\int_{\de_1M}F^{ij}\frac{\delta}{\delta[\mathbb{X}^{i}\mathbb{X}^j]}$, so one gets 
\begin{equation}
\int_{\de_1M}F^{ij}\frac{\delta^2}{\delta\mathbb{X}^{i}\delta\mathbb{X}^j}\left(\int_{\de_1M}u_{i}\land\mathbb{X}^{i}\bullet \int_{\de_1M}v_j\land \mathbb{X}^j\right)=\int_{\de_1M}u_iv_jF^{ij},
\end{equation}
in accordance with our naive expectation.
\end{rem}

\begin{defn}[Full quantum state]
\label{full_covariant_state1}
Let $M$ be a manifold, possibly with boundary. 
Given a $BF$-like BV-BFV theory $\pi_M\colon\calF_M \to \calF^\de_{\de M}$, a polarization $\calP$ on $\calF^\de_{\de M}$, a good splitting $\calF_M = \mathcal{B}^\calP_{\de M}\oplus \calV^\calP_M \oplus \calY'$, and a gauge-fixing Lagrangian $\calL \subset \calY'$, we define the \textsf{full quantum state} by the formal
power series 
\begin{equation}
\boldsymbol{\psi}_M(\mathbb{X},\E;\mathsf{x},\mathsf{e})=T_M\exp\left(\frac{\I}{\hbar}\sum_{\Gamma}\frac{(-\I\hbar)^{\textnormal{loops}(\Gamma)}}{\vert\textnormal{Aut}(\Gamma)\vert}\int_{\mathsf{C}_\Gamma(M)}\omega_\Gamma(\mathbb{X},\E;\mathsf{x},\mathsf{e})\right),\label{eq:def_state_pert_full}
\end{equation}
\end{defn}
\begin{rem}
The full state can be interpreted as an expectation value with help of the bullet product: 
\begin{equation}
\boldsymbol{\psi}_M = T_M \left\langle \ee_{\bullet}^{\frac{\ii}{\hbar}(\calS^\textnormal{res} + \calS^{\textnormal{source}})}\right\rangle
\end{equation}
where $\ee_\bullet$ denotes the exponential with respect to the bullet product. 
\end{rem}

\subsubsection{The BFV boundary operator}
We want to define the quantum BFV boundary operator for $BF$-like theories according to \cite{CMR2}. Similarly to the state, we will express at first its principal part and then extend it to a regularization using the notion of composite fields. The quantum BFV boundary operator is constructed as a quantization of the BFV action such that Theorem \ref{thm:CMR2} below holds.

\begin{defn}[Principal part of the BFV boundary operator]
The \textsf{principal part} of the BFV boundary operator is given by 
\begin{equation}
\Omega^{\textnormal{princ}}=\underbrace{\Omega_0^\mathbb{X}+\Omega_0^\E}_{=:\Omega_0}+\underbrace{\Omega_{\textnormal{pert}}^\mathbb{X}+\Omega_{\textnormal{pert}}^\E}_{=:\Omega_{\textnormal{pert}}^{\textnormal{princ}}},
\end{equation}
where 
\begin{align}
\Omega_0^\mathbb{X}&:=(-1)^d\I\hbar\int_{\de_1M}\left(\dd \mathbb{X}\frac{\delta}{\delta\mathbb{X}}\right),\\
\Omega_0^\E &:=(-1)^d\I\hbar\int_{\de_2 M}\left(\dd\mathbb{E}\frac{\delta}{\delta\mathbb{E}}\right),\\
\Omega_{\textnormal{pert}}^\mathbb{X}&:=\sum_{n,k\geq 0}\sum_{\Gamma_1'}\frac{(\I\hbar)^{\textnormal{loops}(\Gamma_1')}}{\vert \textnormal{Aut}(\Gamma'_1)\vert}\int_{\de_1M}\left(\sigma_{\Gamma_1'}\right)_{i_1....i_n}^{j_1...j_k}\land\mathbb{X}^{i_1}\land \dotsm \land\mathbb{X}^{i_n} \left((-1)^d\I\hbar\frac{\delta}{\delta\mathbb{X}^{j_1}}\right)\dotsm \left((-1)^d\I\hbar\frac{\delta}{\delta\mathbb{X}^{j_k}}\right),\\
\Omega_{\textnormal{pert}}^\mathbb{E}&:=\sum_{n,k\geq 0}\sum_{\Gamma_2'}\frac{(\I\hbar)^{\textnormal{loops}(\Gamma_2')}}{\vert \textnormal{Aut}(\Gamma'_2)\vert}\int_{\de_2M}\left(\sigma_{\Gamma_2'}\right)_{i_1....i_n}^{j_1...j_k}\land \mathbb{E}^{i_1}\land\dotsm\land \mathbb{E}^{i_n}\left((-1)^d\I\hbar\frac{\delta}{\delta\mathbb{E}^{j_1}}\right)\dotsm \left((-1)^d\I\hbar\frac{\delta}{\delta\mathbb{E}^{j_k}}\right),
\end{align}
where, for $\mathbb{F}_1=\mathbb{X}$ and $\mathbb{F}_2=\E$ and $\ell\in\{1,2\}$, $\Gamma_\ell'$ runs over graphs with 
\begin{itemize}
\item{$n$ vertices on $\de_\ell M$ of valence 1 with adjacent half-edges oriented inwards and decorated with boundary fields $\mathbb{F}^{i_1}_\ell,...,\mathbb{F}^{i_n}_\ell$ all evaluated at the point of collapse $p\in \de_\ell M$,}
\item{$k$ outward leaves if $\ell=1$ and $k$ inward leaves if $\ell=2$, decorated with variational derivatives in boundary fields
$$(-1)^d\I\hbar\frac{\delta}{\delta\mathbb{F}^{j_1}_\ell},...,(-1)^d\I\hbar\frac{\delta}{\delta\mathbb{F}^{j_k}_\ell}$$
at the point of collapse,
}
\item{
no outward leaves if $\ell=2$ and no inward leaves if $\ell=1$  (graphs with them do not contribute).}
\end{itemize}
The form $\sigma_{\Gamma_\ell'}$ is obtained as the integral over the compactified configuration space $\Tilde{\mathsf{C}}_{\Gamma_\ell'}(\mathbb{H}^d)$, where $\mathbb{H}^d$ denotes the $d$-dimensional upper half plane, given by
\begin{equation}
\sigma_{\Gamma_\ell'}=\int_{\Tilde{\mathsf{C}}_{\Gamma_\ell'}(\mathbb{H}^d)}\omega_{\Gamma_\ell'},
\end{equation}
with $\omega_{\Gamma_\ell'}$ being the product of limiting propagators at the point $p$ of collapse and vertex tensors. 
\end{defn}

We want to roughly describe the construction of the BFV boundary operator with composite fields (see \cite{CMR2} for a more detailed discussion). First, we need to define the following notion.

On a regular functional as in \eqref{regular_functional}, we get a term $L$ replaced by $\dd L$ plus all the terms corresponding to  the boundary of the configuration space. As $L$ is smooth, its restriction to the boundary is also smooth and can be integrated on the fibers yielding a smooth form on the base configuration space; for example
\[
\Omega_0\int_{\de_1M}L_{IJ}\land[\mathbb{X}^{I}]\land [\mathbb{X}^J]=\pm \I\hbar\int_{\de_1M}\dd L_{IJ}\land [\mathbb{X}^{I}]\land [\mathbb{X}^J],
\]
\begin{multline*}
\Omega_0\int_{\mathsf{C}_2(\de_1M)}L_{IJK}\land \pi_1^*([\mathbb{X}^{I}]\land[\mathbb{X}^J])\land\pi_2^*[\mathbb{X}^K]\\=\pm\I\hbar \int_{\mathsf{C}_2(\de_1M)}\dd L_{IJK}\land \pi_1^*([\mathbb{X}^{I}]\land [\mathbb{X}^J])\land \pi_2^*[\mathbb{X}^K]\pm\I\hbar \int_{\de_1M}\underline{L_{IJK}}\land [\mathbb{X}^{I}]\land [\mathbb{X}^J]\land[\mathbb{X}^K],
\end{multline*}
with $\underline{L_{IJK}}=\pi^\de_*L_{IJK}$, where $\pi^\de\colon \de\mathsf{C}_2(\de_1M)\to \de_1M$ is the canonical projection.


Notice that for any two regular functionals $S_1$ and $S_2$ we have $$\Omega_0(S_1\bullet S_2)=\Omega_0(S_1)\bullet S_2\pm  S_1\bullet \Omega_0(S_2).$$ The other generators that we allow are products of expressions of the form 
\begin{align}
\int_{\de_1M}&L^J_{I^1...I^r}\left[\mathbb{X}^{I_1}\right]\land\dotsm \land\left[\mathbb{X}^{I_{r}}\right]\frac{\delta^{\vert J\vert}}{\delta [\mathbb{X}^J]}\\
\int_{\de_2M}&L_I^{J^1...J^\ell}\left[\mathbb{E}_{J_1}\right]\land\dotsm \land\left[\mathbb{E}_{J_{\ell}}\right]\frac{\delta^{\vert I\vert}}{\delta [\mathbb{E}_I]}.
\end{align}

\begin{defn}[Full BFV boundary operator]
\label{full_BFV}
The \textsf{full BFV boundary operator} is given by 
\begin{equation}
\boldsymbol{\Omega}_{\de M}:=\Omega_0+\underbrace{\boldsymbol{\Omega}_{\textnormal{pert}}^\mathbb{X}+\boldsymbol{\Omega}_{\textnormal{pert}}^\E}_{\boldsymbol{\Omega}_{\textnormal{pert}}},
\end{equation}
where 
\begin{equation}
\boldsymbol{\Omega}_{\textnormal{pert}}^\mathbb{X}:=\sum_{n,k\geq 0}\sum_{\Gamma_1'}\frac{(\I\hbar)^{\textnormal{loops}(\Gamma_1')}}{\vert \textnormal{Aut}(\Gamma'_1)\vert}\int_{\de_1M}\left(\sigma_{\Gamma_1'}\right)_{I_1....I_n}^{J_1...J_k}\land\left[\mathbb{X}^{I_1}\right]\land \dotsm \land\left[\mathbb{X}^{I_n}\right] \left((-1)^{kd}(\I\hbar)^k\frac{\delta^{\vert J_1\vert+\dotsm +\vert J_k\vert}}{\delta\left[\mathbb{X}^{J_1}\dotsm \mathbb{X}^{J_k}\right]}\right),
\end{equation}
\begin{equation}
\boldsymbol{\Omega}_{\textnormal{pert}}^\mathbb{E}:=\sum_{n,k\geq 0}\sum_{\Gamma_2'}\frac{(\I\hbar)^{\textnormal{loops}(\Gamma_2')}}{\vert \textnormal{Aut}(\Gamma'_2)\vert}\int_{\de_2M}\left(\sigma_{\Gamma_2'}\right)^{I_1....I_n}_{J_1...J_k}\land \left[\mathbb{E}_{I_1}\right]\land\dotsm\land \left[\mathbb{E}_{I_n}\right] \left((-1)^{kd}(\I\hbar)^k\frac{\delta^{\vert J_1\vert+\dotsm +\vert J_k\vert}}{\delta\left[\mathbb{E}_{J_1}\dotsm \mathbb{E}_{J_k}\right]}\right)\end{equation}
where, for $\mathbb{F}_1=\mathbb{X}$ and $\mathbb{F}_2=\E$ and $\ell\in\{1,2\}$, $\Gamma_\ell'$ runs over graphs with 
\begin{itemize}
\item{$n$ vertices on $\de_\ell M$, where vertex $s$ has valence $\vert I_s\vert\geq 1$, with adjacent half-edges oriented inwards and decorated with boundary fields $[\mathbb{F}^{I_1}_\ell],...,[\mathbb{F}^{I_n}_\ell]$ all evaluated at the point of collapse $p\in \de_\ell M$,}
\item{$\vert J_1\vert+\dotsm +\vert J_k\vert$ outward leaves if $\ell=1$ and $\vert J_1\vert+\dotsm +\vert J_k\vert$ inward leaves if $\ell=2$, decorated with variational derivatives in boundary fields
$$(-1)^d\I\hbar\frac{\delta}{\delta[\mathbb{F}^{J_1}_\ell]},...,(-1)^d\I\hbar\frac{\delta}{\delta[\mathbb{F}^{J_k}_\ell]}$$
at the point of collapse,
}
\item{
no outward leaves if $\ell=2$ and no inward leaves if $\ell=1$ (graphs with them do not contribute).}
\end{itemize}
The form $\sigma_{\Gamma_\ell'}$ is obtained as the integral over the compactified configuration space $\Tilde{\mathsf{C}}_{\Gamma_\ell'}(\mathbb{H}^d)$, where $\mathbb{H}^d$ denotes the $d$-dimensional upper half plane, given by
\begin{equation}
\sigma_{\Gamma_\ell'}=\int_{\Tilde{\mathsf{C}}_{\Gamma_\ell'}(\mathbb{H}^d)}\omega_{\Gamma_\ell'},
\end{equation}
with $\omega_{\Gamma_\ell'}$ being the product of limiting propagators at the point $p$ of collapse and vertex tensors. 
\end{defn}

\begin{figure}[h!]
\begin{tikzpicture}

\draw[thick] (-6,0) -- (6,0); 
\node[vertex] (bdry1) at (-1,0) {};
\node[coordinate, label=below:{}] at (bdry1.south) {};
\node[vertex] (bdry2) at (1,0) {};
\node[coordinate, label=below:{}] at (bdry2.south) {};


\node[vertex] (bulk1) at (-0.5,1) {}; 
\node[] (pi1) at (-0.8,1){};
\node[vertex] (bulk2) at (0.8,1.3) {};
\node[] (pi2) at (1.1,1.3){};
\node[vertex] (bulk3) at (0,0.5) {};
\node[] (pi3) at (0.3,0.5){};


\draw[] (0:2) arc (0:180:2); 
\draw[fermion] (-1,3.5) -- (bulk1);
\draw[fermion] (-0.5,3.5) -- (bulk2);
\draw[fermion] (0,3.5) -- (bulk2);
\draw[fermion] (bulk2) -- (bdry2);
\draw[fermion] (bulk1) -- (bdry1);
\draw[fermion] (bulk1) -- (bulk2);
\draw[fermion] (bulk2) -- (bulk3);
\draw[fermion] (bulk3) -- (bdry1);
\draw[fermion] (bulk3) -- (bulk1);
\draw[fermion] (5,3) -- (bdry2);
\draw[fermion] (5,2) -- (bdry2);
\draw[fermion] (5,1) -- (bdry2);
\draw[fermion] (-5,3) -- (bdry1);
\draw[fermion] (-5,2) -- (bdry1);
\draw[fermion] (-5,1) -- (bdry1);
\end{tikzpicture} 
\caption{Example of a graph collapsing to the boundary with three bulk and two boundary vertices. The semicircle represents the collapsing of the graph.}
\label{fig:collapsing_boundary_op}
\end{figure}
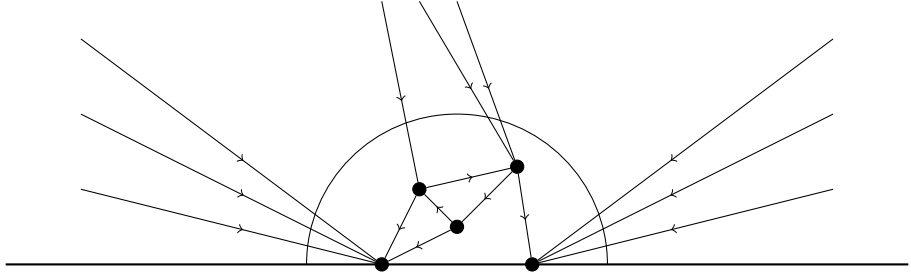

\begin{thm}[\cite{CMR2}]
\label{thm:CMR2}
Let $M$ be a smooth manifold (possibly with boundary). Then the following hold:
\begin{enumerate}
\item{The full covariant state $\boldsymbol{\psi}_M$ satisfies the mQME:
\begin{equation}
\label{full_mQME}
(\hbar^2\Delta_{\calV_M}+\boldsymbol{\Omega}_{\de M})\boldsymbol{\psi}_M=0.
\end{equation}
}
\item{The full BFV boundary operator $\boldsymbol{\Omega}_{\de M}$ squares to zero: 
\begin{equation}
\label{flatness_full_boundary_op}
(\boldsymbol{\Omega}_{\de M})^2=0.
\end{equation}
}
\item{A change of propagator or residual fields leads to a theory related by change of data as in \ref{change_of_data}.
}
\end{enumerate}
\end{thm}

\FloatBarrier
\section{Perturbative Quantization of AKSZ Sigma Models}\label{sec:AKSZ_Quantization}
\label{AKSZ}

\subsection{Review of AKSZ Sigma Models in the BV-BFV formalism}
We will begin with a brief review of AKSZ Sigma Models (\cite{AKSZ}) as described in \cite{CMR1}. 

\begin{defn}[Differential graded symplectic manifold]
A \textsf{dg symplectic manifold} of degree $d$ is a graded manifold $\calM$ endowed with a symplectic form $\omega=\dd \alpha$ of degree $d$ and a Hamiltonian function $\Theta$ of degree $d+1$ satisfying $\{\Theta,\Theta\}=0$, where $\{\enspace,\enspace\}$ is the Poisson bracket induced by $\omega$. 
\end{defn}

\begin{rem}
This is also called a \textsf{Hamiltonian manifold}.
\end{rem}

\begin{defn}[AKSZ Sigma Model]
The \textsf{AKSZ Sigma Model} with target a Hamiltonian manifold $(\calM,\omega=\dd\alpha,\Theta)$ of degree $d-1$ is the BV theory, which associates to a $d$-manifold $\Sigma$ the BV manifold $(\calF_\Sigma,\omega_\Sigma,\calS_\Sigma)$, where $\calF_\Sigma=\Map(T[1]\Sigma,\calM)$, $\omega_\Sigma$ is of the form $\omega_\Sigma=\int_\Sigma\omega_{\mu\nu}\delta \textsf{A}^\mu\land \delta\textsf{A}^\nu$,
and $\calS_\Sigma[\textsf{A}]=\int_\Sigma\left(\alpha_\mu(\textsf{A}^\mu)\dd \textsf{A}^\mu+\Theta(\textsf{A})\right)$, where $\textsf{A}\in\calF_\Sigma$,
$\omega_{\mu\nu}$ are the components of the symplectic form $\omega$, $\alpha_\mu$ are the components of $\alpha$ and $\textsf{A}^\mu$ are the components of $\textsf{A}$ in local coordinates. 
\end{defn}

Here ``$\Map$'' denotes the right adjoint functor to the Cartesian product of graded manifolds (with a fixed factor). On objects we have $\text{Hom}(X,\Map(Y,Z))=\text{Hom}(X\times Y,Z)$, where $\text{Hom}$ denotes the set of graded manifold morphisms.

\subsection{Split AKSZ Sigma Models} 
In this paper we will especially be interested in the case where $\mathcal{M} = T^*[d-1]M$, with $M$ a graded manifold, such that the symplectic form is given by the canonical symplectic structure $\omega_0$. 

\begin{defn}[Split AKSZ Sigma Model]
We call an AKSZ Sigma Model \textsf{split}, if the target is of the form 
\begin{equation}
\calM=T^*[d-1]M
\end{equation}
with canonical symplectic structure, where $M$ is a graded manifold.
\end{defn}

Coordinates on the space of fields can be considered as a pair $(\sfX,\boldeta)$, where $\sfX$ and $\boldeta$ are the base and fiber components of the map respectively. The action can be written as 
\begin{equation}
\calS_{\Sigma}[(\sfX,\boldeta)] = \calS^{\text{kin}}_\Sigma[(\sfX,\boldeta)] + \calS^{\text{int}}_\Sigma[(\sfX,\boldeta)]
\end{equation}
where the kinetic and interaction terms are given by 
\begin{align*}
\calS^{\text{kin}}_\Sigma[(\sfX,\boldeta)] &:= \int_\Sigma \langle \boldeta, \dr \sfX\rangle, \\
\calS^{\text{int}}_\Sigma[(\sfX,\boldeta)] &:= \int_\Sigma \Theta(\sfX,\boldeta),
\end{align*}
where $\langle\enspace , \enspace \rangle $ denotes the canonical pairing between tangent and cotangent bundle of $M$ and we think of elements of $C^{\infty}(T[1]\Sigma)$ as differential forms in the usual way, i.e. of elements in $\Omega^\bullet(\Sigma)$. In \cite{CMR1} it was shown that these data define a BV-BFV theory as in Definition \ref{defn:BVBFV_theory} in Section \ref{sec:BVBFV}. 

\subsection{Coordinatization of split AKSZ theories}
In this paper we want to quantize split AKSZ theories as perturbations of abelian $BF$ theory. This can be done by ``coordinatizing the target'', i.e. replacing the space of fields with the formal neighbourhood of a constant field. Using methods of \textsf{formal geometry} as in \cite{GK,B,CF3,BCM,CMW2} one can do this consistently for all constant solutions at once. In Appendix \ref{app:formal_geometry} we recall this procedure and its extension to graded manifolds, which is discussed in \cite{LS}.  For more details we refer to \cite{CF3,BCM}. 
\subsubsection{Coordinatizing the AKSZ construction} 
The idea now is to expand the theory around critical points of the kinetic part of the action. Denote by $\Bar{M}$ the body of the graded manifold $M$, and let $x \in \Bar{M}$. We will work in formal neighbourhoods of constant maps $$x = (\sfX , \boldsymbol\eta) \equiv (x,0) \in \Map(T[1]\Sigma,\calM)$$ 
   Let $\varphi$ be a formal exponential map (see Appendix \ref{App:formal_exp}) on $M$. This induces a map 
$$\varphi_x \colon T_xM \to M$$ which lifts to a map   
\begin{align*}
\widetilde{\varphi}_x \colon \mathcal{F}_{\Sigma,x}:=\Map(T[1]\Sigma,T^*[d-1]T_xM) & \to \Map(T[1]\Sigma,\calM) \\ 
(\hatX,\hateta) &\mapsto (\sfX,\boldeta)
\end{align*}
by taking post-composition with the cotangent lift. Notice that $\widetilde{\varphi}$ is a local symplectomorphism and that 
\begin{equation}
\calS^{\textnormal{kin}}_{\Sigma,x}:=\T\widetilde{\varphi}_x^*\calS^{\text{kin}}_\Sigma = \int_\Sigma \langle \hateta, \dr \hatX\rangle,
\end{equation}
where $\textsf{T}$ denotes the Taylor expansion as in \ref{App:formal_exp}. If we define
\begin{equation}
\calS^{\textnormal{int}}_{\Sigma,x}:=\T\widetilde{\varphi}_x^*\calS^{\text{int}}_\Sigma= \int_\Sigma\T\widetilde{\varphi}^*_x\Theta(\sfX,\boldeta)
\end{equation} 
and 
\begin{equation}
\calS_{\Sigma,x} := \calS^{\textnormal{kin}}_{\Sigma,x} + \calS^{\textnormal{int}}_{\Sigma,x}=\int_\Sigma\left(\langle \hateta,\dd\hatX\rangle+\mathsf{T}\Tilde{\varphi}^*_x\Theta(\mathsf{X},\boldsymbol{\eta})\right)
\end{equation}
then the pair $(\calF_{\Sigma,x},\calS_{\Sigma,x})$ is a $BF$-like theory in the sense of \cite{CMR2}, i.e. the kinetic part of the action is a sum of copies of the kinetic part of abelian $BF$-theory and for every $x\in\overline{M}$ it satisfies the mCME (see Equation \eqref{mCME} in Section \ref{sec:BVBFV}) $$\iota_{Q_{\Sigma,x}}\omega_{\Sigma,x} - \delta \calS_{\Sigma,x} = \pi_\Sigma^*\alpha^{\de}_{\de \Sigma,x}.$$ 
where $Q_{\Sigma,x}$ is the Hamiltonian vector field of $\calS_{\Sigma,x}$. Moreover, $\omega_{\Sigma,x}$ and $\alpha^\de_{\de\Sigma,x}$ are the corresponding symplectic form and boundary 1-form of the BV-BFV manifold associated to the space of fields $\calF_{\Sigma,x}$.
Notice that it could be obtained from the AKSZ construction with target $T^*[d-1]T_xM$ and Hamiltonian function $\Theta_x := \T\Tilde{\varphi}^*_x\Theta$. We regard $\Theta_x$ as a formal function on $T^*[d-1]T_xM$ and we will write 

\begin{equation}
\Theta_x(y^1,\ldots,y^r,\xi_1,\ldots\xi_r) = \sum_{k,l = 0}^{\infty}\Theta^{j_1\ldots j_l}_{i_1\ldots i_k}(x)y^{i_1}\cdots y^{i_k}\xi_{j_1}\cdots\xi_{j_l}\label{eq:Theta_expansion} 
\end{equation}
where $r = \dim M$ and the $\xi_i$ are the cotangent coordinates of the coordinates $y^i$.

\subsubsection{Varying the classical background}
We now define the map $\Hat{\calS}_\Sigma$ to be given by $\Hat{\calS}_\Sigma\colon x \mapsto \calS_{\Sigma,x}$. In local coordinates $(x^i)$ on $\Bar{M}$, we  define 
\begin{equation}
\calS_{\Sigma,x,R} := \int_\Sigma Y^j_i(x;\hatX)\hateta_j\wedge \dd x^i,
\end{equation}
where $Y \in \Gamma(\Bar{M},\Hat{S}T^*M)$ is defined in Appendix \ref{app:formal_geometry}, which is also a formal power series in the second argument, hence we can express $Y$ as 
\begin{equation}
Y^j_i(x;y) = \sum_{k=0}^{\infty} Y^j_{i;i_1,\ldots,i_k}(x)y^{i_1}\cdots y^{i_k}.\label{eq:Yexpansion}
\end{equation}  Notice that here we pull back to the body $\Bar{M}$ of $M$ via the zero section of $M \to \Bar{M}$. Moreover, on a closed manifold we have $$\dd_x\Hat{\calS}_\Sigma = (\calS_{\Sigma,x,R},\Hat{\calS}_\Sigma).$$ 
\begin{defn}[Formal global action]
The \textsf{formal global action} for a split AKSZ theory is defined by 
\begin{equation}
\label{formal_global_action}
\Tilde{\calS}_{\Sigma,x}:=\int_\Sigma\left(\hateta_{i}\land\dd\hatX^{i}+\mathsf{T}\Tilde{\varphi}_x^*\Theta(\mathsf{X},\boldsymbol{\eta})+Y^{j}_i(x;\hatX)\hateta_j\land \dd x^{i}\right)=\calS_{\Sigma,x}+\calS_{\Sigma,x,R}.
\end{equation}
\end{defn}

Using the formal global action, we get
\begin{equation}
\dd_x\Tilde{\calS}_{\Sigma,x} + \frac{1}{2}(\Tilde{\calS}_{\Sigma,x},\Tilde{\calS}_{\Sigma,x}) = 0. \label{dCME}
\end{equation}
This condition is called the \textsf{differential Classical Master Equation} (dCME) (see \cite{BCM,CMR1,CMR2,CMW2}). On a manifold with boundary, we get the cohomological vector field $\Tilde{Q}_{\Sigma,x}$ from the BV-BFV theory on $\mathcal{F}_{\Sigma,x}$. Recall the construction of a vector field $R$ in the setting of formal geometry as in Appendix \ref{sec_GrConn}. For a section $\sigma\in\Gamma(\Hat{S}T^*M)$ we have 
$$R(\sigma)=-\dd_y\sigma\circ (\dd_y\varphi)^{-1}\circ \dd_x\varphi.$$
Indeed, we can lift the vector field $R$ to a vector field $\Tilde{R}$ on $\mathcal{F}_{\Sigma}$ and define $\Tilde{Q}_{\Sigma,x} = \Hat{Q}_\Sigma + \Tilde{R}$, where $\Hat{Q}_\Sigma$ is the Hamiltonian vector field for $\Hat{\calS}_\Sigma$. Then we have 
\begin{equation}
\iota_{\Tilde{Q}_{\Sigma,x}}\omega_{\Sigma,x} - \delta \Tilde{\calS}_{\Sigma,x} = \pi_\Sigma^*\alpha^{\de}_{\de \Sigma}, 
\end{equation}
the \textsf{modified differential Classical Master Equation (mdCME)}. 

\begin{rem}
A similar approach to globalization for closed manifolds was done by Grady--Gwilliam, Costello, Grady--Li--Li (\cite{GG,Cost1,GLL}).  Their construction is based on the idea that one can replace the target by an $L_\infty$ equivalent one, whereas the one introduced in \cite{BCM} before was based on the idea of using formal geometry to define a symplectomorphism on a neighborhood of each solution in the space of fields to start the perturbation theory. The two approaches are essentially equivalent. However, in \cite{GG,Cost1,GLL} they only get $BF_\infty$ theories since they start with theories of a particular simple type.
We consider more general theories that do not fit into this. Here $BF_\infty$ means that one of the two fields appears at most linearly, but this is not the case in our setting (e.g., in the Poisson Sigma Model for a nonlinear Poisson structure). Moreover, in principle one should work around more general solutions than just the constant ones. In principle, one should do formal geometry on the moduli space of solutions. Note also that this construction can in principle be generalized to non AKSZ models. 
\end{rem}

\subsection{Quantization}
We now have a bundle of $BF$-like theories over the body $\Bar{M}$ of $M$. In every fiber we can apply a perturbative BV-BFV quantization as in \cite{CMR2}. That is, we define a splitting of the space of fields 
$$\mathcal{F}_{\Sigma,x} = \mathcal{B}^{\calP}_{\de\Sigma} \oplus \mathcal{V}_{\Sigma,x} \oplus  \mathcal{Y}' $$ as in \eqref{split} and split the fields accordingly as 
\begin{align*}
\hatX &= \mathbb{X} \oplus \sfx \oplus \mathscr{X}, \\
\hateta &= \E \oplus \sfe \oplus \mathscr{E},
\end{align*}
where $\calB^{\calP}_{\de\Sigma}$ is the base space of a polarization $\mathcal{P}$ of boundary fields, $$\calV_{\Sigma,x} = H^{\bullet}(\Sigma,\de_1\Sigma) \otimes T_xM \oplus H^{\bullet}(\Sigma,\de_2\Sigma) \otimes T_x^*M$$ is the space of residual fields and $\mathcal{Y}$ is a symplectic complement of $\calB^\calP_{\de\Sigma} \oplus \calV_{\Sigma,x}$. The polarizations that we consider are defined by splitting the boundary $\de \Sigma = \de_1\Sigma \sqcup \de_2\Sigma$ and choosing the $\mathbb{X}$-representation on $\de_1\Sigma$ and the $\mathbb{E}$-representation on\footnote{This simply means that we choose the $\frac{\delta}{\delta \mathbb{X}}$-polarization on $\de_2\Sigma$ and the $\frac{\delta}{\delta\E}$-polarization on $\de_1\Sigma$} $\de_2\Sigma$. Let us denote by $\mathcal{H}^{\mathcal{P}}_{\partial\Sigma,x}$ the boundary state space as in Definition \ref{space_of_boundary_states}. 
Using the definition of the formal global action $\Tilde{\calS}_{\Sigma,x}$ and Definition \ref{principal_part}, we can define a covariant version of the principal part of the quantum state.

\begin{defn}[Principal covariant quantum state]
The \textsf{principal covariant quantum state} $\Tilde{\psi}_{\Sigma,x}$ is defined as in Definition \ref{principal_part},
using the Feynman rules given in Figure \ref{fig:FeynmanRules} coming from the formal global action $\Tilde{\calS}_{\Sigma,x}$.
\end{defn}

\subsubsection{Feynman graphs and rules}\label{sec:FeynmanRules}
The Feynman graphs and rules are the same as in \cite{CMR2,CMW,CMW2}, but there are additional interaction vertices given by $\calS_{\Sigma,x,R}$. Namely, to an interaction vertex with $k$ incoming and $l$ outgoing half-edges labeled by $i_1,\ldots, i_k$ and $j_1,\ldots ,j_l$ respectively we associate 
$\Theta_{i_1 \ldots i_k}^{j_1 \ldots j_l}(x)$ as defined in \eqref{eq:Theta_expansion}. To a vertex labeled by $R$, with $k$ incoming half-edges labeled $i_1,\ldots,i_k$ and one outgoing edge labeled $j$, we associate $Y^j_{i;i_1,\ldots,i_k}(x)$ as in  \eqref{eq:Yexpansion}. Half-edges can start at $\mathsf{e}$ zero modes and boundary vertices on $\partial_1\Sigma$ and end at $\mathsf{x}$ zero modes or boundary vertices on $\partial_2 \Sigma$. See Figure \ref{fig:FeynmanRules}.  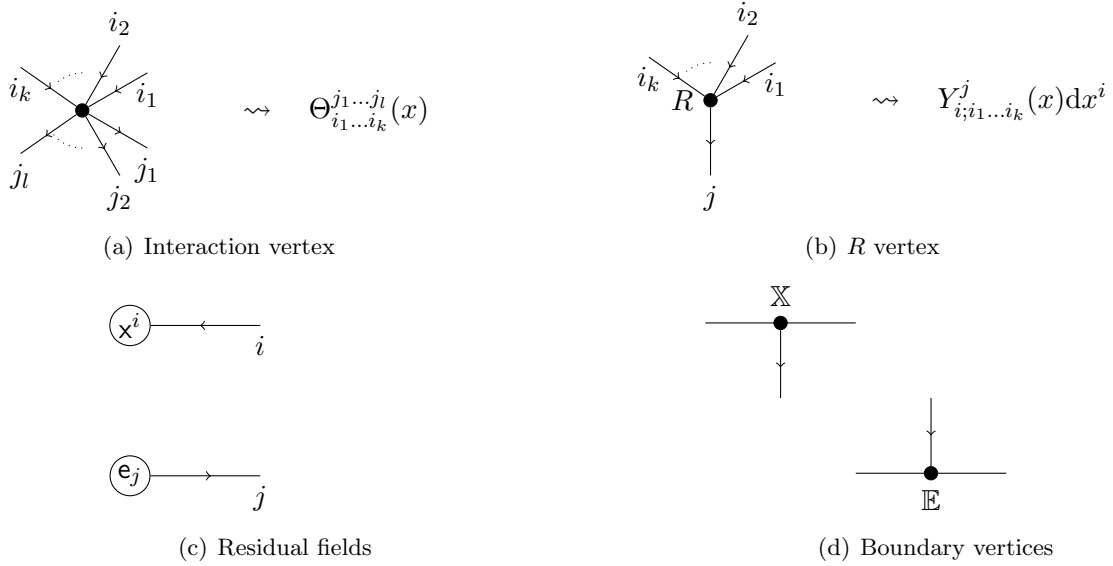
\begin{figure}
\centering
\subfigure[Interaction vertex]{
\centering
\begin{tikzpicture}
\node[vertex] (o) at (0,0) {};
\node[coordinate, label=below:{$i_1$}] at (30:1) {$i_1$}
edge[fermion] (o);
\node[coordinate, label=above:{$i_2$}] at (60:1) {$i_2$}
edge[fermion] (o);
\node[coordinate, label=below:{$i_k$}] at (145:1) {$i_k$}
edge[fermion] (o);
\draw[dotted] (90:0.5) arc (90:130:0.5);
\node[coordinate, label=below:{$j_1$}] (j1) at (-30:1) {};
\node[coordinate, label=below:{$j_2$}] (j2) at (-60:1) {};
\node[coordinate, label=below:{$j_l$}] (j3) at (-145:1) {};
\draw[dotted] (-90:0.5) arc (-90:-145:0.5); 
\draw[fermion] (o) -- (j1);
\draw[fermion] (o) -- (j2);
\draw[fermion] (o) -- (j3);
\node[coordinate,label=right:{$\leadsto\quad\Theta^{j_1\ldots j_l}_{i_1 \ldots i_k}(x)$}] at (2,0) {};
\end{tikzpicture}
}
\hspace{2cm}
\subfigure[$R$ vertex]{
\centering
\begin{tikzpicture}
\node[vertex,label=left:{$R$}] (o) at (0,0) {};
\node[coordinate, label=below:{$i_1$}] at (30:1) {$i_1$}
edge[fermion] (o);
\node[coordinate, label=above:{$i_2$}] at (60:1) {$i_2$}
edge[fermion] (o);
\node[coordinate, label=below:{$i_k$}] at (145:1) {$i_k$}
edge[fermion] (o);
\draw[dotted] (90:0.5) arc (90:130:0.5);
\node[coordinate, label=below:{$j$}] (j1) at (-90:1) {};
\draw[fermion] (o) -- (j1);
\node[coordinate,label=right:{$\leadsto\quad Y^{j}_{i;i_1\ldots i_k}(x)\dd x^i$}] at (2,0) {};
\end{tikzpicture}
}

\hspace{2cm}
\subfigure[Residual fields]{
\centering
\begin{tikzpicture}
\node[circle,draw,inner sep=1pt] (x) at (0,1) {$\mathsf{x}^i$};
\node[coordinate, label=below:{$i$}] (x2) at (30:2) {$i_1$}
edge[fermion] (x);
\node[coordinate, label=below:{$j$}] (e2)at (-30:2) {$i_1$};
\node[circle,draw,inner sep=1pt] (e) at (0,-1) {$\mathsf{e}_j$}
edge[fermion] (e2);
\node[coordinate, label={${}$}] at (4,0) {};
\end{tikzpicture}
}
\hspace{3cm}
\subfigure[Boundary vertices]{
\begin{tikzpicture}
\draw (-1,1) -- (1,1);
\node[vertex, label=above:{$\mathbb{X}$}] (x) at (0,1) {};
\node[coordinate] (b1) at (0,0) {}; 
\draw[fermion] (x) -- (b1);
\draw (1,-1) -- (3,-1);
\node[vertex, label=below:{$\mathbb{E}$}] (e) at (2,-1) {};
\node[coordinate] (b2) at (2,0) {}; 
\draw[fermion] (b2) -- (e);

\end{tikzpicture}
\hspace{2cm}
}
\caption{Summary of Feynman graphs and rules}\label{fig:FeynmanRules}
\end{figure}

\subsubsection{The full covariant state}
\label{covariant state}
As we have seen in Section \ref{sec:BVBFV}, we need to deal with composite fields in order to regularize higher functional derivatives, hence we also need a covariant version of the full state.

\begin{defn}[Full covariant quantum state]
\label{full_state_2}
We define the \textsf{full covariant quantum state} $\btpsi_{\Sigma,x}$ as in Definition \ref{full_covariant_state1}, 
using the Feynman rules in Figure \ref{fig:FeynmanRules} coming from the formal global action $\Tilde{\calS}_{\Sigma,x}$ and additionally with the rules for the boundary vertices as in Figure \ref{fig:composite_field_vertices}.
\end{defn}

\section{The modified differential Quantum Master Equation (mdQME)}
\label{sec_mdQME}
The mQME, as a condition to hold in the BV-BFV formalism on manifolds with boundary, needs to be modified for a globalized AKSZ theory. The more general condition is called the \textsf{modified differential Quantum Master Equation (mdQME)}. The classical and quantum aspects of this modification are discussed in \cite{BCM,CMR1}, and first discussed for manifolds with boundary in \cite{CMW2}. We want to think of the  operator 
\[\qtconn :=\left(\dr_x -\I\hbar \Delta_{\calV_{\Sigma,x}} + \frac{\I}{\hbar} \boldsymbol{\Omega}_{\de \Sigma}\right)  \] 
as a connection on the total bundle of spaces of states over (a part of) the moduli space of classical solutions of the theory. We call this operator the \textsf{quantum Grothendieck BFV (GBFV) operator}. 

\begin{rem}
As already mentioned, the quantum GBFV operator $\nabla_\mathsf{G}$ is an operator on forms valued in sections of the total bundle of states, which is a graded vector space and it is an operator of total degree $1$, but not of form degree $1$. We will call it an operator instead of a connection, since it can be misleading to think of $-\I\hbar\Delta_{\calV_{\Sigma,x}}+\frac{\I}{\hbar}\boldsymbol{\Omega}_{\de\Sigma}$ as a connection $1$-form. Rather, it defines a Maurer--Cartan element in the dg Lie algebra of differential forms with values in sections of the endomorphism bundle of the total state space.
\end{rem}

The goal of this section can be rephrased as showing that the state gives a well-defined $\qtconn$-cohomology class. For this we have to show that: 
\begin{enumerate}
\item The state defines a closed section with respect to $\nabla_\mathsf{G}$ (the mdQME), 
\item The operator $\qtconn$ is a coboundary operator, i.e. $\qtconn^2 = 0$, 
\item The cohomology class of $\btpsi_{\Sigma,x}$ is independent of the choices made, i.e. if we alter any of these choices, the state changes in a controlled way.
\end{enumerate}
This will be the program of this section. Heuristically, this result can be interpreted as saying that the state comes from a well-defined function on (a part of) the moduli space of classical solutions of the theory. 

\subsection{Assumptions on the Theory}
The proof for the program of this section depends on two important assumptions on our theory, which we will discuss in this section.

\subsubsection{No hidden faces anomalies}
Let $\Gamma$ be a Feynman graph and denote by $V(\Gamma)$ the set of its vertices; it decomposes into bulk vertices $V_B(\Gamma)$ and boundary vertices $V^{\de}(\Gamma)$.  The boundary of the configuration space is a union of several faces. We will denote by $F_{ij}$ the faces where two bulk vertices $i,j \in V_B(\Gamma)$ collapse in the bulk. By $F_{\geq 3}$ we denote the union of the faces where at least three bulk vertices collapse in the bulk, usually called ``hidden faces''. By $F^{\de}_{i_1,\ldots,i_k,j_1,\ldots,j_l}$ we denote faces where the bulk vertices $i_1,\ldots ,i_k \in V_B(\Gamma)$ and the boundary vertices $j_1,\ldots, j_l \in V^{\de}(\Gamma)$ collapse at a point in the boundary; the union of all these faces is denoted by $F^{\de}_{\Gamma}$. 

\begin{defn} 
\label{hidden_faces}
We say that a theory is \textsf{(hidden faces) anomaly free} if for every graph $\Gamma$ we have that 
\begin{equation}
\int_{F_{\geq 3}} \omega_{\Gamma} = 0, 
\end{equation}
i.e. all possible contributions of hidden faces vanish.
\end{defn}

\begin{rem}
A theory that is famously not anomaly free is Chern--Simons theory, see \cite{AS,AS2} and \cite{BC}, where the first ansatz for the quantum theory depends on the choice of gauge fixing. In this case one can get away of the anomaly with introducing a framing and a framing-dependent counterterm for the dependence on the gauge fixing. On the other hand, there are many examples of anomaly free theories. In particular, Kontsevich's result \cite{K} implies that any 2-dimensional theory is anomaly-free, e.g. the Poisson Sigma Model (\cite{I,SS1,SS2,CF1}).
\end{rem}

\begin{rem}[Counterterms]\label{rem:counterterms}
A general ansatz to deal with theories with anomalies is the addition of counterterms to the action. If the differential form which results from integrating over a hidden face is exact, one can add the corresponding primitive to the action, thus producing new vertices cancelling the anomaly. In Chern--Simons theory, this produces the ``framing'' anomaly, since the only hidden faces contribution comes from faces where all vertices in a graph collapse. The resulting differential form is a representative of the relative Pontryagin class of $M \times I$, whose primitive is the Chern--Simons form of the flat connection used to construct the propagator.
\end{rem} 

\subsubsection{Unimodularity}
In the quantization of general AKSZ theories one can have tadpoles, also called short loops, i.e. arrows starting and ending at the same vertex. They need to be treated seperately and can in principle spoil the mdQME. The best way to get around them is to assume that the theory satisfies a ``unimodularity'' condition.

\begin{defn}[Unimodularity]
We say that a given theory is \textsf{unimodular} if any
 contraction of the vertex tensor $\Theta$ with itself is zero.
\end{defn}

\subsection{The modified differential Quantum Master Equation} 
One of the main results, and the first point of the program is the following theorem:

\begin{thm}[mdQME for split AKSZ theories]
\label{thm:mdQME}
Consider the full covariant perturbative state $\btpsi_{\Sigma,x}$ as a quantization of an anomaly free and unimodular split AKSZ theory with target $T^*[d-1]M$, where $M$ is a graded manifold. Then
\begin{equation}
\label{AKSZ_mdQME}
\left(\dr_x -\I\hbar \Delta_{\calV_{\Sigma,x}} + \frac{\I}{\hbar} \boldsymbol{\Omega}_{\de \Sigma}\right) \btpsi_{\Sigma,x}=0,
\end{equation}
where we denote by $\dr_x$ the de Rham differential on $\Bar{M}$, the body of the graded manifold $M$.
\end{thm}We will prove this by considering the Feynman graphs of the theory analogously to the proof of the mQME in \cite{CMR2}.

\begin{proof}

\label{setup}
For the following computation we consider Feynman graphs which also have vertices, of any possible valency, on the boundary deriving the functions attached there. Let $\mathscr{G}$ denote the set of Feynman graphs of the theory. Then $\btpsi_{\Sigma,x}$ can be written as 
\begin{equation}
\btpsi_{\Sigma,x}=T_\Sigma\sum_{\Gamma \in \mathscr{G}}\int_{\mathsf{C}_{\Gamma}}\omega_{\Gamma}, \label{eq:Defbtpsi}
\end{equation}
where we include the combinatorial prefactor $\frac{(-\I\hbar)^{\text{loops($\Gamma$)}}}{\vert \Aut(\Gamma)\vert}$ in $\omega_{\Gamma}$ (here loops($\Gamma$) denotes the number of loops of a graph $\Gamma$). Moreover, we denote the configuration space $\mathsf{C}_\Gamma(\Sigma)$ by $\mathsf{C}_\Gamma$ for simplicity. Note that $\omega_{\Gamma}$ is a ($\calV_{\Sigma,x}$-dependent) differential form on $\mathsf{C}_{\Gamma} \times \Bar{M}$. Now recall Stokes' theorem for integration along a compact fiber with corners:
\begin{equation}
\dr \pi_* = \pi_*\dr - \pi^{\de}_* .\label{eq:fiberstokes}
\end{equation} The integrals in \eqref{eq:Defbtpsi} are fiber integrals, hence we can apply \eqref{eq:fiberstokes} to yield 
\begin{equation}
\dr_x \int_{\mathsf{C}_{\Gamma}}\omega_{\Gamma} = \int_{\mathsf{C}_{\Gamma}}\dr\omega_{\Gamma} - \int_{\de \mathsf{C}_{\Gamma}} \omega_{\Gamma}. \label{eq:dxintomega}
\end{equation}
Here $\dr$ inside the integral is the total differential on $\Bar{M} \times \mathsf{C}_{\Gamma}$, and thus we can split it as 
\begin{equation}
\dr = \dr_x + \dr_{1} + \dr_{2} \label{eq:diffsplit}
\end{equation}
Here $\dr_1$ denotes the part of the de Rham differential acting on the propagators in $\omega_{\Gamma}$, and $\dr_2$ the part acting on $\mathbb{X}$ and $\E$ fields. Let us introduce some more notation: The set of edges of $\Gamma$ will be denoted by $E(\Gamma)$. We denote by $E_k(\Gamma)$ the set of edges $e$ whose removal increases the number of connected components by $k$. Clearly $E(\Gamma) = E_0(\Gamma) \sqcup E_1(\Gamma)$, and $e \in E_1(\Gamma)$ if and only if $e$ is not part of a loop in $\Gamma$.


\begin{prop} 
\label{prop1}
The following hold:
\begin{enumerate}[$(i)$]
\item The action of the BV Laplacian on the state is given by 
\begin{equation}
\I\hbar\Delta_{\calV_{\Sigma,x}} \btpsi_{\Sigma,x} = T_{\Sigma}\sum_{\Gamma \in \mathscr{G}} \int_{\mathsf{C}_{\Gamma}}\dr_1\omega_{\Gamma}  \label{eq:Deltaaction}
\end{equation}
\item The action of $\Omega_0$ on the state is given by 
\begin{equation}
-\frac{\I}{\hbar}\Omega_0\btpsi_{\Sigma,x} = T_{\Sigma}\sum_{\Gamma \in \mathscr{G}} \int_{\mathsf{C}_{\Gamma}}\dr_2\omega_{\Gamma}  \label{eq:Omega0action}
\end{equation}
\item For the boundary contributions we have 
\begin{align}
 \sum_{\Gamma \in \mathscr{G}}\sum_{i,j \in V_B(\Gamma)} \int_{F_{ij}}\omega_{\Gamma} &= \sum_{\Gamma \in \mathscr{G}} \int_{\mathsf{C}_{\Gamma}}\dr_x \omega_{\Gamma}, \label{eq:dxaction}
\end{align}
where we denote by $F_{ij}$ the boundary faces where two bulk vertices $i,j \in V_B(\Gamma)$ collapse in the bulk, and furthermore 
\begin{equation}
 T_{\Sigma}\sum_{\Gamma \in \mathscr{G}}\int_{F^{\de}_{\Gamma}}\omega_{\Gamma} = \frac{\I}{\hbar}\boldsymbol{\Omega}_{\textnormal{pert}}\btpsi_{\Sigma,x},\label{eq:def_Omega_pert}
\end{equation}
where $F_\Gamma^\de$ is the union of the $F^{\de}_{i_1,\ldots,i_k,j_1,\ldots,j_l}$, which denote the boundary faces where the bulk vertices $i_1,\ldots ,i_k \in V_B(\Gamma)$ and the boundary vertices $j_1,\ldots, j_l \in V^{\de}(\Gamma)$ collapse at a point in the boundary.
\end{enumerate}
\end{prop} 

Proposition \ref{prop1} immediately implies the mdQME for anomaly free theories by the following simple computation:
\begin{align*}
\dr_x\btpsi_{\Sigma,x}&= T_{\Sigma}\sum_{\Gamma \in \mathscr{G}}\dr_x \int_{\mathsf{C}_{\Gamma}}\omega_{\Gamma} = T_{\Sigma} \sum_{\Gamma \in\mathscr{G}} \left(\int_{\mathsf{C}_{\Gamma}}\dr\omega_{\Gamma} - \int_{\de \mathsf{C}_{\Gamma}} \omega_{\Gamma}\right) \\ 
&= T_{\Sigma} \sum_{\Gamma \in \mathscr{G}} \left(\int_{\mathsf{C}_{\Gamma}}\dr_x\omega_{\Gamma} + \dr_1\omega_{\Gamma} + \dr_{2}\omega_{\Gamma} - \int_{\de \mathsf{C}_{\Gamma}} \omega_{\Gamma}\right) \\ 
&= \I\hbar \Delta_{\calV_{\Sigma,x}} \btpsi_{\Sigma,x} - \frac{\I}{\hbar}\Omega_0\btpsi_{\Sigma,x} - \frac{\I}{\hbar}\boldsymbol{\Omega}_{\textnormal{pert}}\btpsi_{\Sigma,x}
\end{align*}
\end{proof}

\subsection{Proof of Proposition \ref{prop1}}
The proof of Theorem \ref{thm:mdQME} is relied on the proof of Proposition \ref{prop1}. We split the proof into four parts. Namely, we show the Equations \eqref{eq:Deltaaction}, \eqref{eq:Omega0action}, \eqref{eq:dxaction} and \eqref{eq:def_Omega_pert} seperately and conclude.

\subsubsection{Proof of Equation \eqref{eq:Deltaaction}}
Consider a propagator\footnote{This means that we want $\zeta$ to be given by $\pm\frac{1}{T_\Sigma}\frac{1}{\I\hbar}\int_\calL \ee^{\frac{\I}{\hbar}\mathscr{S}_\Sigma}\mathscr{X}\land\mathscr{E}$, where $T_\Sigma=\int_\calL \ee^{\frac{\I}{\hbar}\mathscr{S}_\Sigma}$ and $\mathscr{S}_\Sigma=\int_\Sigma \mathscr{E}\land \dd\mathscr{X}$. Here the Lagrangian $\calL$ is given by the direct sum of the image of the Hodge-theoretic chain contraction and the image of its dual (with the correct shift).} $\zeta$ as in \cite{CMR2} and denote by $\zeta_{ij}:=\zeta(u_i,u_j)$ for $u_i,u_j\in\Sigma$. Moreover, let $\chi_i$ denote a representative for the basis of $\calV_{\Sigma,x}$ and denote by $\chi^{i}$ a representative of its dual basis. Hence, in local coordinates, we can write the residual fields, $\sfx$ and $\sfe$ as $\sfx=\sum_iz^{i}\chi_i$ and $\sfe=\sum_iz^+_i\chi^{i}$.
Then we have the identity 
\begin{equation}
\dr\zeta_{12} = \pm\sum_{k}\pi_1^*\chi_k\pi^*_2\chi^k = \pm\Delta_{\calV_{\Sigma,x}}(\sfx_1\sfe_2),
\end{equation} 
where $\pi_1$ and $\pi_2$ denote the projections to the first and second factor of $\mathsf{C}_2(\Sigma)$ respectively and where $\sfx_i:=\sfx(u_i)$ and $\sfe_j:=\sfe(u_j)$.
Recall also that the BV Laplacian is given by 
$$\Delta_{\calV_{\Sigma,x}} =\pm \sum_k \frac{\de^2}{\de z^{k}\de z_k^+} $$ 
where $\deg \frac{\de}{\de z_k^+} = - \deg z_k^+ = \deg z^k -1 = 1 - \deg \chi_k - 1 = - \deg \chi_k$. Since $\deg \sfx = 1 $ we get 
$$ \Delta_{\calV_{\Sigma,x}}(\pi^*_1\sfx\pi^*_2\sfe) = \pm\sum_k\pi^*_1\chi_k\pi^*_2\chi^k.$$   

Let us introduce some more notation for certain operations on graphs. For any graph $\Gamma$, let $\bullet$ denote either an edge $e = (i,j)$ or a pair of residual fields $\sfx_i,\sfe_j$\footnote{Note that an edge denotes a contracted $\mathscr{X}$-$\mathscr{E}$-pair, so $\bullet$ denotes either an $\mathscr{X}$-$\mathscr{E}$ or an $\sfx$-$\sfe$-pair.}. Denote by $\Gamma_\bullet'$ the graph resulting from removing the component labeled by $\bullet$ and replacing it with a diagonal class between points $i$ and $j$, i.e.  the sum $\sum_k \pm \pi^*_i\chi_k\pi_j^*\chi^k$ (see also Figure \ref{graphs_resid}). 

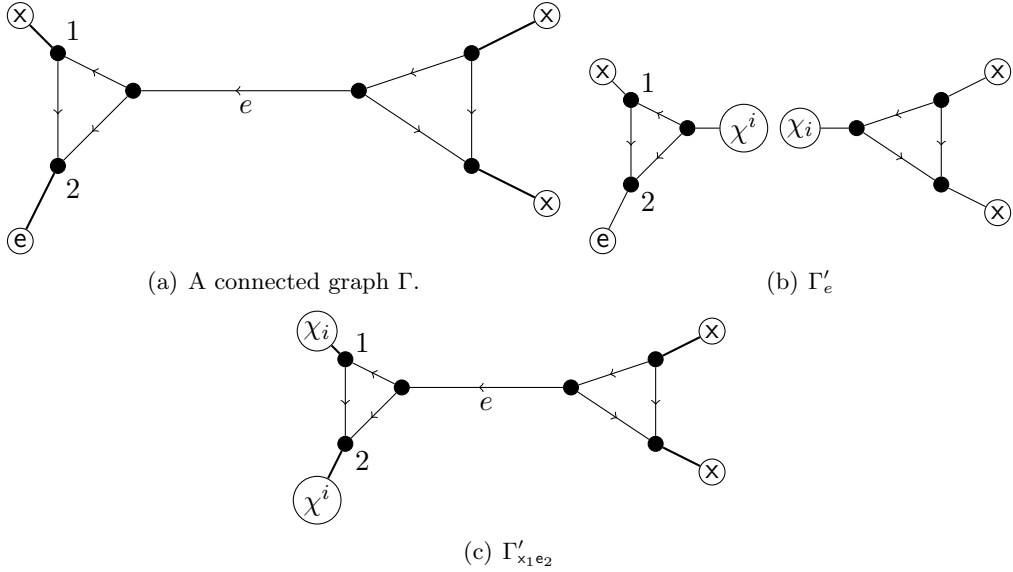
\begin{figure}[h!]
\subfigure[A connected graph $\Gamma$.]{
\begin{tikzpicture}
\node [internal] (v1) at (-1.5,0.5) {};
\node [internal] (v4) at (1.5,0.5) {};
\node [internal] (v6) at (3,-0.5) {};
\node [internal] (v5) at (3,1) {};
\node [internal] (v2) at (-2.5,-0.5) {};
\node [internal] (v3) at (-2.5,1) {};
\node [] (11) at (-2.3,1.3) {$1$};
\node [] (22) at (-2.3,-0.8) {$2$};
\draw [fermion] (v1)--(v2);
\draw [fermion] (v1)--(v3);
\draw [fermion] (v4)--node[below] {$e$} (v1);
\draw [fermion] (v5)--node[above] {}(v4);
\draw [fermion] (v4)--(v6);
\node [residual ] (v8) at (-3,1.5) {$\mathsf{x}$};
\node [residual ] (v7) at (-3,-1.5) {$\mathsf{e}$};
\node [residual ] (v9) at (4,1.5) {$\mathsf{x}$};
\node [residual ] (v10) at (4,-1) {$\mathsf{x}$};
\draw [thick, ] (v2) edge (v7);
\draw [thick, ] (v8) edge (v3);
\draw [thick, ] (v5) edge (v9);
\draw [thick, ] (v6) edge (v10);
\draw [fermion] (v3)--(v2);
\draw [fermion] (v5)--(v6);
\end{tikzpicture}
}
\subfigure[$\Gamma'_{e}$]{
\begin{tikzpicture}[scale=0.75]
\node [internal] (v1) at (-1.5,0.5) {};
\node [internal] (v4) at (1.5,0.5) {};
\node [internal] (v6) at (3,-0.5) {};
\node [internal] (v5) at (3,1) {};
\node [internal] (v2) at (-2.5,-0.5) {};
\node [internal] (v3) at (-2.5,1) {};
\node [] (11) at (-2.2,1.3) {$1$};
\node [] (22) at (-2.2,-0.8) {$2$};
\draw [fermion] (v1)--(v2);
\draw [fermion] (v1)--(v3);
\draw [fermion] (v5)--node[above] {}(v4);
\draw [fermion] (v4)--(v6);
\draw [fermion] (v3)--(v2);
\draw [fermion] (v5)--(v6);
\node [residual ] (v8) at (-3,1.5) {$\mathsf{x}$};
\node [residual ] (v7) at (-3,-1.5) {$\mathsf{e}$};
\node [residual ] (v9) at (4,1.5) {$\mathsf{x}$};
\node [residual ] (v10) at (4,-1) {$\mathsf{x}$};
\node [residual] (v11) at (-0.5,0.5) {$\chi^i$} edge (v1);
\node [residual] (v12) at (0.5,0.5) {$\chi_i$} edge (v4);
\draw (v2) edge (v7);
\draw (v8) edge (v3);
\draw (v5) edge (v9);
\draw (v6) edge (v10);
\end{tikzpicture}
}
\subfigure[$\Gamma'_{\sfx_1\sfe_2}$]{
\begin{tikzpicture}[scale=0.75]
\node [internal] (v1) at (-1.5,0.5) {};
\node [internal] (v4) at (1.5,0.5) {};
\node [internal] (v6) at (3,-0.5) {};
\node [internal] (v5) at (3,1) {};
\node [internal] (v2) at (-2.5,-0.5) {};
\node [internal] (v3) at (-2.5,1) {};
\node [] (11) at (-2.2,1.3) {$1$};
\node [] (22) at (-2.2,-0.8) {$2$};
\draw [fermion] (v1)--(v2);
\draw [fermion] (v1)--(v3);
\draw [fermion] (v4)--node[below] {$e$} (v1);
\draw [fermion] (v5)--node[above] {}(v4);
\draw [fermion] (v4)--(v6);
\node [residual ] (v8) at (-3,1.5) {$\chi_i$};
\node [residual ] (v7) at (-3,-1.5) {$\chi^i$};
\node [residual ] (v9) at (4,1.5) {$\mathsf{x}$};
\node [residual ] (v10) at (4,-1) {$\mathsf{x}$};
\draw [thick, ] (v2) edge (v7);
\draw [thick, ] (v8) edge (v3);
\draw [thick, ] (v5) edge (v9);
\draw [thick, ] (v6) edge (v10);
\draw [fermion] (v3)--(v2);
\draw [fermion] (v5)--(v6);
\end{tikzpicture}
}
\caption{Explanation of the operation $\Gamma'_\bullet$.}
\label{graphs_resid}
\end{figure}

Clearly, we have 
\begin{equation}
\sum_{\Gamma \in \mathscr{G}} \int_{\mathsf{C}_{\Gamma}}\dr_1\omega_{\Gamma}  = \sum_{\Gamma \in \mathscr{G}}\sum_{e \in E(\Gamma)} \int_{\mathsf{C}_{\Gamma}}\omega_{\Gamma_e'}\label{eq:domega}
\end{equation}
On the other hand, the properties of the BV Laplacian imply 
$$\Delta_{\calV_{\Sigma,x}} \omega_{\Gamma} = \sum_{\text{pairs of} \atop \text{residual fields }(\sfx_i,\sfe_j) \text{ in } \Gamma} \omega_{\Gamma_{\sfx_i\sfe_j}'},$$
which we can interpret as a first order differential operator on a product $z_i^+z^i$.  
By construction we get that if the edge $e$ starts at $i$ and ends at $j$, we have 
\begin{equation}\omega_{\Gamma_e'} = (\I\hbar)\omega_{\Gamma_{\sfx_i\sfe_j}'}, 
 \label{omegaequal} \end{equation} \
since each edge comes with a factor of $(-\I\hbar)$. 
Now consider the action of $\Delta_{\calV_{\Sigma,x}}$ on $\btpsi_{\Sigma,x}$, and note that 
\begin{equation}\Delta_{\calV_{\Sigma,x}}\btpsi_{\Sigma,x} =T_{\Sigma} \sum_{\Gamma \in \mathscr{G}} \int_{\mathsf{C}_{\Gamma}}\Delta_{\calV_{\Sigma,x}}\omega_{\Gamma} = T_{\Sigma}\sum_{\Gamma \in \mathscr{G}} \int_{\mathsf{C}_{\Gamma}}\sum_{\text{pairs of} \atop \text{residual fields }(\sfx_i,\sfe_j) \text{ in } \Gamma} \omega_{\Gamma_{\sfx_i\sfe_j}'}
\label{eq:Deltapsiproof}
\end{equation}
The sum in Equation \eqref{eq:domega} above can be seen as summing over all graphs with one egde marked - we will denote this set by $\mathscr{G}^{E}$. In the sum in Equation \eqref{eq:Deltapsiproof} above we sum over all  graphs with one pair of residual fields marked - we will denote this set by $\mathscr{G}^{\text{pair}}$. Now define a map 
\begin{align*}
\mathscr{G}^{E} &\to \mathscr{G}^{\text{pair}} \\
\Gamma &\mapsto \widetilde{\Gamma} 
\end{align*} which exchanges the marked edge for a marked pair of residual fields. Clearly this map is invertible and its inverse exchanges the marked pair of residual fields for a marked edge. The contributions to sum labeled by $\Gamma$ and $\widetilde{\Gamma}$ agree up to a factor by Equation \eqref{omegaequal}. We conclude the proof of \eqref{eq:Deltaaction}.

\subsubsection{Proof of Equation \eqref{eq:Omega0action}}
Recall from Section \ref{sec:BVBFV} that $\Omega_0$ is given by 
$$\Omega_0=(-1)^{\dim\Sigma}\I\hbar\left(\int_{\de_1\Sigma}\dd \mathbb{X}^{i}\frac{\delta}{\delta\mathbb{X}^{i}}+\int_{\de_2\Sigma}\dd \E_i\frac{\delta}{\delta\E_i}\right).$$
Hence Equation \eqref{eq:Omega0action} follows immediatly from the definition of the de Rham differential on the $\mathbb{X}$ and $\E$ fields. Moreover, note that $\omega_\Gamma$ is given as a product of propagators, residual fields, boundary fields and vertex tensors as in Section \ref{sec:BVBFV} Equation \eqref{eq:def_state_pert}.

\subsubsection{Proof of Equation \eqref{eq:dxaction}}\label{sec:bdry_terms}
First notice that all the $x$-dependence of $\omega_{\Gamma}$ lies in the bulk vertex tensors. There are two types of bulk vertex tensors, arising from $\mathscr{V}_{\Sigma,x}(\hatX,\hateta):=\mathsf{T}\Tilde{\varphi}^*\Theta(\mathsf{X},\boldsymbol{\eta})$ and $\calS_{\Sigma,x,R}$ respectively, we will call them type $\romI$  and type $\romII$ vertices. Let us analyse them in more detail. \\ 
First, recall that 
\begin{equation}
\mathscr{V}_{\Sigma,x}(\hatX,\hateta) = \sum_{k,l=0}^{\infty}\Theta^{i_1\ldots i_k}_{j_1\ldots j_l}(x)\hatX^{i_1}\cdots \hatX^{i_k}\hateta_{j_1} \cdots \hateta_{j_l}
\end{equation}
(the $\Theta$'s are exactly one set of vertex tensors in the Feynman graphs). 
The fact that $(\mathscr{V}_{\Sigma,x},\mathscr{V}_{\Sigma,x}) = 0$ is equivalent to $\sum_r\pm\frac{\delta}{\delta \hatX^r}\mathscr{V}_{\Sigma,x}(\hatX,\hateta)\frac{\delta}{\delta \hateta_r}\mathscr{V}_{\Sigma,x}(\hatX,\hateta) = 0$. In terms of the vertex tensors it reads as 
\begin{equation}
\sum_{\substack{k'+k''=k\\l'+l'' =l}}\sum_{\substack{1 \leq s' \leq k' \\ 1 \leq s'' \leq k''}}\sum_r \delta^r_{i'_{s'}}\delta_r^{j''_{s''}}\Theta^{i'_1\ldots i'_{k'}}_{j'_1\ldots j'_{l'}}(x)\Theta^{i''_1\ldots i''_{k''}}_{j'_1\ldots j'_{l''}}(x)  = 0 \label{eq:crazysums}
\end{equation}
for every $k,l \geq 0$. This can be understood as follows: From a $(k',l')$-tensor $\Theta'$ and a $(k'',l'')$-tensor $\Theta''$  we can form a $(k,l)=(k'+k''-1,l'+l''-1)$-tensor $\Theta$ by contracting exactly one index. We will say $\Theta$ has been merged from $\Theta'$ and $\Theta''$.  If we sum over all possibilities of constructing $(k,l)$-tensors this way, the result vanishes. Now, suppose we have a graph $\Gamma$ with an edge $e$ between two type $\romI$ vertices $v'$ and $v''$, with vertex tensors $\Theta'$ and $\Theta''$, respectively. The boundary of $\mathsf{C}_{\Gamma}$ contains a face where these two vertices collapse; by normalization of the propagator, the integral over the corresponding fiber yields $\pm 1$. We are left with a new graph $\Gamma/e$ where the edge has been collapsed into a new marked vertex. The vertex tensor at this new vertex  has been merged from  $\Theta'$, $\Theta''$. Now, sum over all graphs, and the corresponding boundary contributions of edges between type $\romI$ vertices. Then we will sum over all ways of merging a vertex in $\Gamma/e$. Hence these contributions vanish by \eqref{eq:crazysums}.  Similarly, one can argue for edges between type $\romII$ vertices, since also  $(\calS_{\Sigma,x,R},\calS_{\Sigma,x,R}) = 0$. For edges between type $\romI$ and type $\romII$ vertices, the relation $\dd_x\mathscr{V}_{\Sigma,x} = (\calS_{\Sigma,x,R},\mathscr{V}_{\Sigma,x})$ implies \eqref{eq:dxaction}.

\subsubsection{Proof of Equation \eqref{eq:def_Omega_pert}}
Recall that we have $\boldsymbol{\Omega}_{\text{pert}} = \boldsymbol{\Omega}^{\bbX}_{\text{pert}} +  \boldsymbol{\Omega}^{\E}_{\text{pert}}$, where $\boldsymbol{\Omega}^{\bbX}_{\text{pert}}$ is constructed as follows. Denote by $\Gamma$ a Feynman graph of the theory, and let $\Gamma'$ be a subgraph of $\Gamma$ (we use the notation $\Gamma' \leq \Gamma$) which contains only bulk vertices and vertices on the boundary component, say, where we work in the $\mathbb{X}$-representation.  Then there is a corresponding contribution $\boldsymbol{\Omega}_{\Gamma'\leq \Gamma}$ to $\boldsymbol{\Omega}_{\de\Sigma}$ given as follows. If $\Gamma'$ has inward leaves (i.e. there is an arrow from some vertex in $\Gamma \setminus \Gamma'$ to a vertex in $\Gamma'$) then $\boldsymbol{\Omega}_{\Gamma' \leq \Gamma}$ vanishes. Suppose the $l$ outward leaves are labeled by $j_1, \ldots,j_l$ and suppose $\Gamma'$ has $k$ boundary vertices with boundary fields $[\bbX^{I_j}], j = 1,\ldots,k$. Then 
\begin{equation}
\boldsymbol{\Omega}_{\Gamma' \leq \Gamma} = \frac{(-\I\hbar)^{\text{loops}(\Gamma')}}{\vert\operatorname{Aut}(\Gamma')\vert}\int_{\de_1 \Sigma} \left(\sigma_{\Gamma'}\right)^J_{I^1 \ldots I^k} [\bbX^{I_1}]\cdots[\bbX^{I_k}] \frac{\delta}{\delta [\bbX^J]}
\end{equation}
where $\sigma_{\Gamma'}$ is the differential form on $\de_1\Sigma$ whose value at $x \in \de_1\Sigma$ is given by integrating the limiting propagators over the compactified configuration space $\Tilde{\mathsf{C}}_{\Gamma'}(\mathbb{H})$ in the upper half-space (see Definition \ref{full_BFV}), which we denote simply by $\Tilde{\mathsf{C}}_{\Gamma'}$ for simplicity, as in Appendix \ref{app:Conf}. Recall that in $\Tilde{\mathsf{C}}_{\Gamma'}$ we take the quotient by translation and scaling. Put differently, there is a boundary face $\de_{1,\Gamma'}\mathsf{C}_{\Gamma}$ of  $\mathsf{C}_{\Gamma}$ corresponding to the collapse of $\Gamma'$ at $\de_1\Sigma$, that face is given by 
$$\de_{1,\Gamma'}\mathsf{C}_{\Gamma} \cong \Tilde{\mathsf{C}}_{\Gamma'} \times \mathsf{C}_{\Gamma/\Gamma' }, $$
where we denote by $\Gamma/\Gamma'$ the collapse of $\Gamma'$ in $\Gamma$ to a new boundary vertex in $\de_1\Sigma$. Let $\Gamma'_{\text{amp}}$  be the ``amputated'' graph $\Gamma'$ where we cut off all the outward leaves and the ones containing residual fields. 
Then $\sigma_{\Gamma'}$ is given as follows. Let $\Tilde{\sigma}_{\Gamma'}$ be the pushforward of $\omega_{\Gamma'_{\text{amp}}}$ along the map $\pi \colon   \Tilde{\mathsf{C}}_{\Gamma'} \times \mathsf{C}_{\Gamma/\Gamma' } \to \mathsf{C}_{\Gamma/\Gamma'}$. Since we take the amputated $\Gamma'$, this pushforward is a basic form in $p\colon \mathsf{C}_{\Gamma/\Gamma'} \to \de_1\Sigma$, and the corresponding form on $\de_1\Sigma$ is $\sigma_{\Gamma'}$, i.e. $\Tilde{\sigma}_{\Gamma'} = p^*\sigma_{\Gamma'}$. Then 
\begin{equation}
\boldsymbol{\Omega}^{\bbX}_{\text{pert}} = \sum_{\Gamma' \leq \Gamma} \boldsymbol{\Omega}_{\Gamma' \leq \Gamma},
\end{equation}
where the sum runs over all Feynman graphs $\Gamma$ of the theory and all their subgraphs $\Gamma'$. 
\begin{figure}[h!]

\begin{tikzpicture}

\draw (-3,0) -- (3,0); 
\node[vertex] (o) at (0,0) {};
\node[coordinate, label=below:{$y$}] at (o.south) {};
\node[vertex] (bdry1) at (-1,0) {};
\node[coordinate, label=below:{$[\mathbb{X}^{i}]$}] at (bdry1.south) {};
\node[vertex] (bdry2) at (1,0) {};
\node[coordinate, label=below:{$[\mathbb{X}^{j}]$}] at (bdry2.south) {};


\node[vertex] (bulk1) at (30:1) {}; 
\node[vertex] (bulk2) at (60:1) {};
\node[vertex] (bulk3) at (130:1) {};


\node[coordinate, label=below:{$i_1$}] (b1) at (30:2) {$i_1$};
\node[coordinate, label=above:{$i_2$}] (b2) at (60:2) {$i_2$};
\node[coordinate, label=below:{$i_3$}] (b3) at (145:2) {$i_3$};
\draw[] (0:1.3) arc (0:180:1.3); 
\draw[fermion] (bulk1) -- (b1);
\draw[fermion] (bulk2) -- (b2);
\draw[fermion] (bulk3) -- (b3);
\draw[fermion] (bulk1) -- (bulk2);   
\draw[fermion] (bulk2) -- (bulk3); 
\draw[fermion] (bulk3) -- (bulk1);
\draw[fermion] (bdry1) -- (bulk3); 
\draw[fermion] (bdry2) -- (bulk1);
\node[coordinate,label=right:{$\leadsto\int_{\de_1\Sigma}\sigma_{\Gamma'}[\bbX^{i}][\bbX^{j}]\frac{\delta}{\delta[\bbX^{i_1}\bbX^{i_2}\bbX^{i_3}]}$}] at (3,1) {};
\end{tikzpicture} 
\caption{An example of a term in $\boldsymbol{\Omega}_{\de\Sigma}$.}
\end{figure}

We can see that 
$$\int_{\de_{\Gamma'}\mathsf{C}_\Gamma}\omega_\Gamma=\boldsymbol{\Omega}_{\Gamma'\leq \Gamma}\int_{\mathsf{C}_{\Gamma/\Gamma'}}\omega_{\Gamma/\Gamma'},$$
and hence we conclude Equation \eqref{eq:def_Omega_pert} by summing over all graphs. One can construct $\boldsymbol{\Omega}^\E_{\text{pert}}$ analogously.

\subsection{Flatness of the quantum GBFV operator}
\label{flat}
We have the following theorem:
\begin{thm}
\label{thm:flatness}
The quantum GBFV operator $\nabla_\mathsf{G}$ squares to zero, i.e.
\begin{equation}
\label{eq:flatness}
(\nabla_\mathsf{G})^2\equiv 0.
\end{equation}
\end{thm}

\begin{proof}
Note that condition \eqref{eq:flatness} is the same as saying that
\begin{equation}
\label{eq:flat}\I\hbar \Delta_{\calV_{\Sigma,x}}\left(\dd_x\boldsymbol{\boldsymbol{\Omega}}_{\de\Sigma}+\boldsymbol{\Omega}_{\de\Sigma}\dd_x\right)=\boldsymbol{\Omega}_{\de\Sigma}^2
\end{equation}
since $\Delta_{\calV_{\Sigma,x}}\boldsymbol{\Omega}_{\de\Sigma}+\boldsymbol{\Omega}_{\de\Sigma}\Delta_{\calV_{\Sigma,x}}=\dd_x\Delta_{\calV_{\Sigma,x}}+\Delta_{\calV_{\Sigma,x}}\dd_x=0$. Here we interpret $\dd_x$ and $\boldsymbol{\Omega}_{\de\Sigma}$ as operators on $\calH_{tot}$-valued differential forms. Equivalently, we can interpret $\boldsymbol{\Omega}_{\de\Sigma}$ as an element of the differential graded Lie algebra of sections of $\bigwedge^{\bullet}T^*M\otimes \mathrm{End}(\calH_{tot})$ and rewrite equation \eqref{eq:flat} in the following intriguing fashion: 
\begin{equation}
\label{eq:OmegaMC}
\I\hbar\dd_x\boldsymbol{\Omega}_{\de\Sigma} - \frac{1}{2}\left[\boldsymbol{\Omega}_{\de\Sigma},\boldsymbol{\Omega}_{\de\Sigma}\right] = 0.
\end{equation}
Equation \eqref{eq:OmegaMC} shows that $\frac{\I}{\hbar}\boldsymbol{\Omega}_{\de\Sigma}$ is a Maurer--Cartan element in this dg-Lie algebra. One can prove this equation by using Stokes' theorem for the definition with Feynman diagrams, similarly as in the proof of flatness in the mQME section of \cite{CMR2}. The crucial point is the following lemma. 
\begin{lem} \label{lem:Omegasquare} We have that 
\begin{equation}
\left(\boldsymbol{\Omega}^{\bbX}_{\textnormal{pert}}\right)^2 = \sum_{\Gamma'' \leq \Gamma' \leq \Gamma} \int_{\de_1 \Sigma} \sigma_{\Gamma''} \sigma_{\Gamma'/\Gamma''} [\bbX^{I_1}]\cdots [\bbX^{I^k}]\frac{\delta}{\delta[\bbX^J]}, 
\end{equation}
where $I_1,\ldots,I_k$ are the boundary vertices of $\Gamma'$ and the outward leaves of $\Gamma'$ are labeled by $J = {j_1, \ldots, j_l}$. An analogous statement holds in the $\E$-representation. 
\end{lem}
\begin{proof}[Proof of Lemma \ref{lem:Omegasquare}]
Since $\boldsymbol{\Omega}^{\bbX}_{\text{pert}}$ has degree 1 we can write 
$$\left(\boldsymbol{\Omega}^{\bbX}_{\text{pert}}\right)^2 = \frac{1}{2}\left[\boldsymbol{\Omega}^{\bbX}_{\text{pert}},\boldsymbol{\Omega}^{\bbX}_{\text{pert}}\right] = \frac{1}{2}\sum_{\Gamma_1' \leq \Gamma_1, \Gamma_2'\leq \Gamma_2}\left[\boldsymbol{\Omega}_{\Gamma'_1 \leq \Gamma_1}, \boldsymbol{\Omega}_{\Gamma_2'\leq \Gamma_2}\right].$$
By definition,  $\boldsymbol{\Omega}_{\de\Sigma}$ contains only first order derivatives (with respect to composite fields). Hence in the commutator the quadratic terms cancel and we are left with the terms where the derivatives act on the coefficients. The bracket 
$$\left[\boldsymbol{\Omega}_{\Gamma'_1 \leq \Gamma_1}, \boldsymbol{\Omega}_{\Gamma_2'\leq \Gamma_2}\right]$$ 
is nonzero if and only if the outward leaves of $\Gamma_1'$ exactly match the composite field at one of the vertices of $\Gamma_2'$ (or vice versa). In the first case, the corresponding contribution is (for simplicity we assume that the corresponding vertex is labeled by 1 in $\Gamma_2'$)  
$$\int_{\de_1 \Sigma}\sigma_{\Gamma_1'}\sigma_{\Gamma_2'}[\bbX^{I^1_1}]\cdots[\bbX^{I^1_{k_1}}][\bbX^{I^2_2}]\cdots [\bbX^{I^2_{k_2}}]\frac{\delta}{\delta [\bbX^J]},$$
where the composite fields at the vertices of $\Gamma_i'$ are labeled by $I^i_j, 1 \leq j \leq k_i$, and the outward leaves of $\Gamma_2'$ are labeled by $J$. ``Blowing up'' the corresponding vertex $i$ (we denote this operation by $\circ_i$) by replacing it by $\Gamma_1'$, from $\Gamma_2$ we obtain a new graph $\Gamma$, and from $\Gamma_2'$ a subgraph $\Gamma'$ of $\Gamma$. Denoting the subgraph $\Gamma_1' \leq \Gamma' \leq \Gamma$ by $\Gamma''$, we obtain that $\Gamma_2' = \Gamma'/\Gamma''$. In this way we obtain all possible graphs $\Gamma$ with all possible combinations of subgraphs $\Gamma'' \leq \Gamma' \leq \Gamma$. See also Figure \ref{fig:Omegasquare}. 
\begin{figure}[h!]

\begin{tikzpicture}
\draw (-5,0) -- (-3,0);
\node[vertex] (bdry11) at (-4,0) {};
\node[coordinate, label=below:{$[\mathbb{X}^{a}]$}] at (bdry11.south) {};
\node[coordinate, label=below:{$\Gamma_1'$}] at (-4,-1) {};
\node[vertex] (bulk11) at (-4,0.5) {}; 
\node[coordinate, label=below:{$j_2$}] (b11) at (-3,1.5) {$i_1$};
\node[coordinate, label=above:{$j_1$}] (b12) at (-5,1.5) {$i_2$};
\draw[domain=0:180] plot ({-4+0.75*cos(\x)},{0.75*sin(\x)}); 
\draw[fermion] (bdry11) -- (bulk11) -- (b11);
\draw[fermion] (bulk11) -- (b12);

\node[coordinate, label=below:{$\circ_2$}] at (-2.5,1) {}; 
\node[coordinate, label=below:{$\Gamma_2'$}] at (0,-1) {};

\draw (-2,0) -- (2,0); 
\node[vertex] (o) at (0,0) {};
\node[coordinate, label=below:{$[\bbX^{j_1}\bbX^{j_2}]$}] at (o.south) {};
\node[vertex] (bdry1) at (-1,0) {};
\node[coordinate, label=below:{$[\mathbb{X}^{i}]$}] at (bdry1.south) {};
\node[vertex] (bdry2) at (1,0) {};
\node[coordinate, label=below:{$[\mathbb{X}^{k}]$}] at (bdry2.south) {};


\node[vertex] (bulk1) at (30:1) {}; 
\node[vertex] (bulk2) at (60:1) {};
\node[vertex] (bulk3) at (130:1) {};


\node[coordinate, label=below:{$l_1$}] (b1) at (30:2) {$i_1$};
\node[coordinate, label=above:{$l_2$}] (b2) at (60:2) {$i_2$};
\node[coordinate, label=below:{$l_3$}] (b3) at (145:2) {$i_3$};
\draw[] (0:1.3) arc (0:180:1.3); 
\draw[fermion] (o) -- (bulk1);
\draw[fermion] (o) -- (bulk3);
\draw[fermion] (bulk1) -- (b1);
\draw[fermion] (bulk2) -- (b2);
\draw[fermion] (bulk3) -- (b3);
\draw[fermion] (bulk1) -- (bulk2);   
\draw[fermion] (bulk2) -- (bulk3); 
\draw[fermion] (bulk3) -- (bulk1);
\draw[fermion] (bdry1) -- (bulk3); 
\draw[fermion] (bdry2) -- (bulk1);

\node[coordinate,label=below:{$=$}] at (2.5,1) {};
 \node[coordinate, label=below:{$\Gamma'$}] at (5,-1) {};


\draw (3,0) -- (7,0); 
\node[vertex] (o2) at (5,0) {};
\node[coordinate, label=below:{$[\bbX^{a}]$}] at (o2.south) {};
\node[vertex] (bdry21) at (4,0) {};
\node[coordinate, label=below:{$[\mathbb{X}^{i}]$}] at (bdry21.south) {};
\node[vertex] (bdry22) at (6,0) {};
\node[coordinate, label=below:{$[\mathbb{X}^{k}]$}] at (bdry22.south) {};


\node[vertex] (bulk21) at ([shift={(5,0)}]30:1) {}; 
\node[vertex] (bulk22) at ([shift={(5,0)}]60:1) {};
\node[vertex] (bulk23) at ([shift={(5,0)}]130:1) {};
\node[vertex] (bulk24) at ([shift={(5,0)}]90:0.25) {};


\node[coordinate, label=below:{$l_1$}] (b21) at ([shift={(5,0)}]30:2) {$i_1$};
\node[coordinate, label=above:{$l_2$}] (b22) at ([shift={(5,0)}]60:2) {$i_2$};
\node[coordinate, label=below:{$l_3$}] (b23) at ([shift={(5,0)}]145:2) {$i_3$};
\draw[shift={(5,0)}] (0:1.3) arc (0:180:1.3); 
\draw[shift={(5,0)}] (0:0.4) arc (0:180:0.4); 

\draw[fermion] (o2) -- (bulk24);
\draw[fermion] (bulk21) -- (b21);
\draw[fermion] (bulk22) -- (b22);
\draw[fermion] (bulk23) -- (b23);
\draw[fermion] (bulk24) -- (bulk22);   
\draw[fermion] (bulk24) -- (bulk23);   

\draw[fermion] (bulk21) -- (bulk22);   
\draw[fermion] (bulk22) -- (bulk23); 
\draw[fermion] (bulk23) -- (bulk21);
\draw[fermion] (bdry21) -- (bulk23); 
\draw[fermion] (bdry22) -- (bulk21);
\end{tikzpicture} 
\caption{An example of a term in $\boldsymbol{\Omega}^2$.}\label{fig:Omegasquare}
\end{figure}
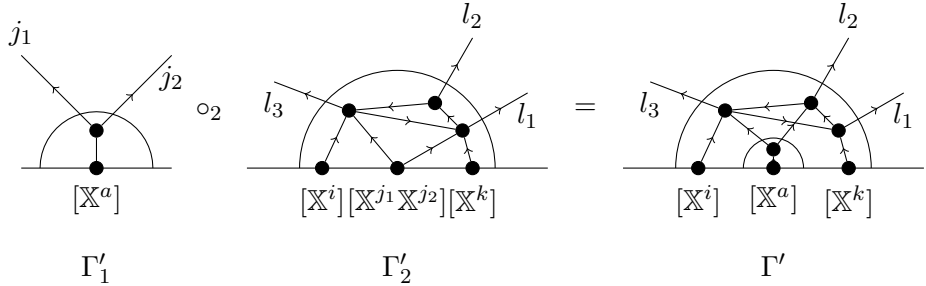
\end{proof} 
Now the proof that $\boldsymbol{\Omega}_{\de\Sigma}$ satisfies the Maurer--Cartan equation can be done very similarly to the original one in \cite{CMR2}.
\begin{proof}[Proof of \eqref{eq:OmegaMC}]
We prove the equation for $\boldsymbol{\Omega}^{\bbX}$, but the proof for $\boldsymbol{\Omega}^{\bbE}$ is analogous and then the claim follows because $\boldsymbol{\Omega}^{\bbX}$ and $\boldsymbol{\Omega}^{\bbE}$ anticommute.  We use again Stokes' theorem (twice). Suppose we apply $\dd_x$ to a summand $\boldsymbol{\Omega}_{\Gamma'\leq\Gamma} = \int_{\de_1 \Sigma}\sigma_{\Gamma'}[\bbX^{I_1}]\cdots[\bbX^{I_k}]\frac{\delta}{\delta[\bbX^J]}$. Then, applying Stokes' theorem we find 
\[\dd_x\boldsymbol{\Omega}_{\Gamma'\leq\Gamma} = \int_{\de_1 \Sigma}(\dd\sigma_{\Gamma'})[\bbX^{I_1}]\cdots[\bbX^{I_k}]\frac{\delta}{\delta[\bbX^J]} + \left[\boldsymbol{\Omega}_0,\boldsymbol{\Omega}_{\Gamma'\leq\Gamma}\right]\] 
(the second term is produced when $\dd_x$ acts on the $\bbX$ fields). Now, we have
\[\dd_x\sigma_{\Gamma'} = \dd_x\int_{\Tilde{\mathsf{C}}_{\Gamma'}}\omega_{\Gamma'} = \int_{\Tilde{\mathsf{C}}_{\Gamma'}}\dd\omega_{\Gamma'} \pm \int_{\de\Tilde{\mathsf{C}}_{\Gamma'}}\omega_{\Gamma'}. \] 
Since the limiting propagator on $\Tilde{\mathsf{C}}_{\Gamma'}$ is closed we have $\dd\omega = \dd_x\omega$. In the boundary integral, we have again three classes of faces. The faces where two bulk points collapse cancel out with $\dd_x\omega$ by the mCME. The terms where more than two bulk points collapse vanish by our assumption that the theory is anomaly free. The terms where a subgraph of $\Gamma'$ collapses at the boundary produce exactly $\frac{1}{2}\left[\boldsymbol{\Omega}^{\bbX}_{\textnormal{pert}},\boldsymbol{\Omega}^{\bbX}_{\textnormal{pert}}\right] $ by Lemma \ref{lem:Omegasquare}. \end{proof}

Now since we have shown that \eqref{eq:OmegaMC} holds, we can conclude that $\nabla_\mathsf{G}$ squares to zero.
\end{proof}

\section{Dependence on choices} 
\label{sec:dep_choices}
\subsection{Covariant gauge transformation} The definition of the state depends on the choices of 
\begin{itemize}
\item the propagator,
\item the residual fields, 
\item the formal exponential map.
\end{itemize}
In this section we will explicitly show how the state and the BFV boundary operator transform under a change in any of these choices. Similarly to Definition \ref{change_of_data} we have the following theorem:
\begin{thm}[Covariant change of data]
\label{thm:dep_choices}
Let $\boldsymbol{\Omega}_t$ be defined as in Definition \ref{full_BFV} and let $\btpsi_t$ be defined as in \ref{full_state_2} for all $t\in[0,1]$. Then we have 
\begin{align}
\frac{\dd}{\dd t}\Big\vert_{t=0}\boldsymbol{\Omega}_t &= \dd_x\tau+[\boldsymbol{\Omega}_{t=0},\tau] \\
\frac{\dd}{\dd t}\Big\vert_{t=0}\btpsi_t &= \qtconn (\btpsi_{t=0} \bullet \varrho) - \tau\btpsi_{t=0} 
\end{align}
for some operator $\tau\in\Gamma(\End(\calH_{tot}))$ and a section $\varrho \in \gm(\calH_{tot})$. Recall that $\bullet$ is the product constructed as in \eqref{bullet_prod}
\end{thm}

\begin{rem}In particular if $\tau$ is zero, the operator $\nabla_\textsf{G}$ does not change and the state changes by a $\nabla_\textsf{G}$-exact term. Theorem \ref{thm:dep_choices} shall be seen as the behaviour of the full covariant state and the full BFV boundary operator under infinitesimal gauge transformation.

\end{rem}

\subsubsection{Possible choices} We have three different choices of how we can mark the graphs according to the change of the state. One possibility is to mark the leaves of a graph $\Gamma$, which corresponds to the change of residual fields and the propagator, another one is to mark the edges which corresponds to the change of the propagator and the last choice is to mark the vertices, which corresponds to the change of the formal exponential map.

\subsection{Changing the residual fields} We have the following proposition:
\begin{prop}[Change of data: residual fields]
Fix some representatives $\chi_i$ and $\chi^{i}$ and consider their change by exact forms as 
\begin{align*}
\dot{\chi}_i&=\dd\sigma_i,\hspace{1cm}\sigma_i\in\Omega_{\text{D}_1}^{\deg \chi_i-1}(\Sigma)\\
\dot{\chi}^{i}&=\dd \sigma^{i},\hspace{1cm}\sigma^{i}\in\Omega_{\text{D}_2}^{\deg \chi^i-1}(\Sigma),
\end{align*}
where $\Omega^\bullet_{\text{D}_i}(\Sigma)=\{\gamma\in \Omega^\bullet(\Sigma)\mid \iota_i^*\gamma=0\}$ with $\iota_i$ the inclusion map from $\partial_i\Sigma$ into $\Sigma$ for $i\in\{1,2\}$ (D stands for Dirichlet). 
Then the family of states $(\btpsi_t)$ changes by 
\begin{equation}
\frac{\dd}{\dd t}\Big\vert_{t=0}\btpsi_t = \qtconn (\btpsi_{t=0} \bullet \varrho) - \tau\btpsi_{t=0} \label{eq:change_of_btpsi},
\end{equation}
where $\tau=0$ and
\begin{equation}
\varrho=\sum_{\text{$\Gamma^{\mathsf{l}}_m$ marked} \atop\text{and connected graph}}\varrho_{\Gamma^{\mathsf{l}}_m}=\sum_{\text{$\Gamma^{\mathsf{l}}_m$ marked} \atop \text{and connected graph}}\int_{\mathsf{C}_{\Gamma^{\mathsf{l}}_m}}\omega'_{\Gamma^{\mathsf{l}}_m}\in\Gamma(\mathcal{H}_{tot}),
\end{equation}
where $\Gamma^\mathsf{l}_m$ denotes a marked connnected graph with $\mathsf{l}\in\{\sfx,\sfe\}$, i.e. a graph with a labeled $\mathsf{l}$ leaf and $\omega'_{\Gamma^\textsf{l}_m}$ is the form constructed with the usual Feynman rules where we place $z^+_{i}\sigma^{i}$ at the marked leaf.
\end{prop}

\begin{proof}
The propagator changes by 
\begin{equation}
\label{change_prop}
\dot{\zeta}= \sum_i\pm \pi_1^*\sigma_i\pi_2^*\chi^{i}\pm \sum_i\pm \pi_1^*\chi_i\pi_2^*\sigma^{i}.
\end{equation}
Let $\btpsi=\sum_{\Gamma}\int_{C_\Gamma}\omega_\Gamma$. Then 
\[
\frac{\dd}{\dd t}\btpsi=\frac{\dd}{\dd t}\sum_{\Gamma}\int_{\mathsf{C}_\Gamma}\omega_\Gamma=\sum_{\Gamma}\int_{\mathsf{C}_\Gamma}\dot{\omega}_\Gamma
=\sum_{\Gamma_m^{\textsf{x}}}\omega_{\Gamma_m^{\textsf{x}}}+\sum_{\Gamma_m^\textsf{e}}\omega_{\Gamma_m^\textsf{e}}
+\sum_{\Gamma^e_m}\omega_{\Gamma^e_m}^{e_{\sigma\chi}}+\omega_{\Gamma^e_m}^{e_{\chi\sigma}},
\]
where by $\Gamma_m^{\textsf{x}}$ we mean graphs with a labeled $\mathsf{x}$ leaf and $\omega_{\Gamma_m^{\textsf{x}}}$ is the usual $\omega_{\Gamma}$ but with $z^i\dot{\chi_i}$ placed at the marked leaf. Similarly, by $\Gamma_m^{\textsf{e}}$ we mean graphs with a labeled $\mathsf{e}$ leaf and by $\omega_{\Gamma_m^{\textsf{e}}}$ the usual $\omega_{\Gamma}$ but with $z_i^+\dot{\chi}^{i}$ placed at the marked leaf. These terms arise when the time derivative hits a leaf. Finally, $\Gamma^e_m$ denotes graphs with a marked edge $e \in E(\Gamma)$, and $\omega_{\Gamma_m^e}^{e_{\sigma\chi}}$ denotes the usual $\omega_{\Gamma}$, where at the marked edge we place $\sum \pm \pi_1^*\sigma_i\pi_2^*\chi^i$.  These terms arise when the time derivative hits a propagator. We call them ``edge splits'' since the corresponding propagator in $\omega$ is split, see Figure \ref{fig:graphs_resid}. 
Define $$\varrho:=\sum_{\text{$\Gamma^{\mathsf{l}}_m$ marked} \atop\text{and connected graph}}\varrho_{\Gamma^{\mathsf{l}}_m}=\sum_{\text{$\Gamma^{\mathsf{l}}_m$ marked} \atop \text{and connected graph}}\int_{\mathsf{C}_{\Gamma^{\mathsf{l}}_m}}\omega'_{\Gamma^{\mathsf{l}}_m}\in\Gamma(\mathcal{H}_{tot}).$$
Here $\mathsf{l}  \in \{\mathsf{x},\mathsf{e}\}$. Again $\Gamma^{\mathsf{l}}_m$ denotes a graph with a marked $\mathsf{l}$ leaf. In $\omega'_{\Gamma^{\mathsf{l}}_m}$ however we replace the marked leaf by $z^i\sigma_i$ or $z_i^+\sigma^i$ respectively. 

\begin{figure}[h!]
\subfigure{
\begin{tikzpicture}
\node [] (dt) at (-4,0.5){\Large $\frac{\dd}{\dd t}$};
\node [internal] (v1) at (-1.5,0.5) {};
\node [internal] (v4) at (2,0.5) {};
\node [internal] (v6) at (1,-0.5) {};
\node [internal] (v5) at (1,1) {};
\node [internal] (v2) at (-2.5,-0.5) {};
\node [internal] (v3) at (-2.5,1) {};
\node [internal] (v11) at (6,0.5) {};
\node [internal] (v12) at (5,1) {};
\node [internal] (v13) at (5,-0.5) {};
\node [internal] (v14) at (10,0.5) {};
\node [internal] (v15) at (9,1) {};
\node [internal] (v16) at (9,-0.5) {};
\draw [fermion] (v1)--(v2);
\draw [fermion] (v3)--(v1);
\draw [fermion] (v5)--(v4);
\draw [fermion] (v4)--(v6);
\draw [fermion] (v2)--(v3);
\draw [fermion] (v6)--(v5);
\draw [fermion] (v12)--(v11);
\draw [fermion] (v13)--(v12);
\draw [fermion] (v11)--(v13);
\draw [fermion] (v15)--(v14);
\draw [fermion] (v16)--(v15);
\draw [dashed] (v14) edge (v16);
\node[] (eq) at (0,0.5){$=$};
\node[] (plus1) at (3.5,0.5){$+$};
\node[] (plus2) at (7.5,0.5){$+$}; 
\node [residual ] (v8) at (-3,1.5){$\mathsf{x}$};
\node [residual ] (v7) at (-3,-1.5) {$\mathsf{e}$};
\node [residual ] (v9) at (0,2) {$\dot{\chi}_i=\dd\sigma_i$};
\node [residual ] (v10) at (0.5,-1.5) {$\mathsf{e}$};
\node [residual ] (res1) at (4,1.5) {$\mathsf{x}$};
\node [residual ] (res2) at (4.5,-1.5) {$\dot{\chi}^{i}=\dd\sigma^{i}$};
\node [residual ] (res3) at (8.5,-1.5) {$\mathsf{e}$};
\node [residual ] (res4) at (8,1.5) {$\mathsf{x}$};
\node [residual ] (sp1) at (10,-1.5) {$\sigma^{i}$};
\node [residual ] (sp2) at (10,1.5) {$\chi_i$};e
\draw [thick, ] (v2) edge (v7);
\draw [thick, ] (v8) edge (v3);
\draw [thick, ] (v5) edge (v9);
\draw [thick, ] (v6) edge (v10);
\draw [thick, ] (v13) edge (res2);
\draw [thick, ] (v12) edge (res1);
\draw [thick, ] (res3) edge (v16);
\draw [thick, ] (res4) edge (v15);
\draw [thick, ] (sp1) edge (v16);
\draw [thick, ] (sp2) edge (v14);
\end{tikzpicture}
}
\subfigure{
\begin{tikzpicture}
\node[] (plus) at (-3,0.5){$+$};
\node [] (v1) at (3.5,0.5) {$+$ \hspace{1cm}Other edge splits};
\node [internal] (v1) at (-0.5,0.5) {};
\node [internal] (v2) at (-1.5,-0.5) {};
\node [internal] (v3) at (-1.5,1) {};
\draw [dashed] (v2)--(v1);
\draw [fermion] (v3)--(v2);
\draw [fermion] (v1)--(v3);
\node [residual ] (res1) at (-2,1.5) {$\mathsf{x}$};
\node [residual ] (res2) at (-2,-1.5) {$\mathsf{e}$};
\node [residual ] (sp1) at (-0.5,-1.5) {$\chi^{i}$};
\node [residual ] (sp2) at (-0.5,1.5) {$\sigma_i$};
\draw[thick] (res1) edge (v3);
\draw[thick] (res2) edge (v2);
\draw[thick] (sp1) edge (v2);
\draw[thick] (sp2) edge (v1);
\end{tikzpicture}
}
\caption{Graphical illustration of the time derivative, the first two terms come from time derivatives of leaves, all other terms - ``edge splits'' - from time derivatives of the propagator. }
\label{fig:graphs_resid}
\end{figure}
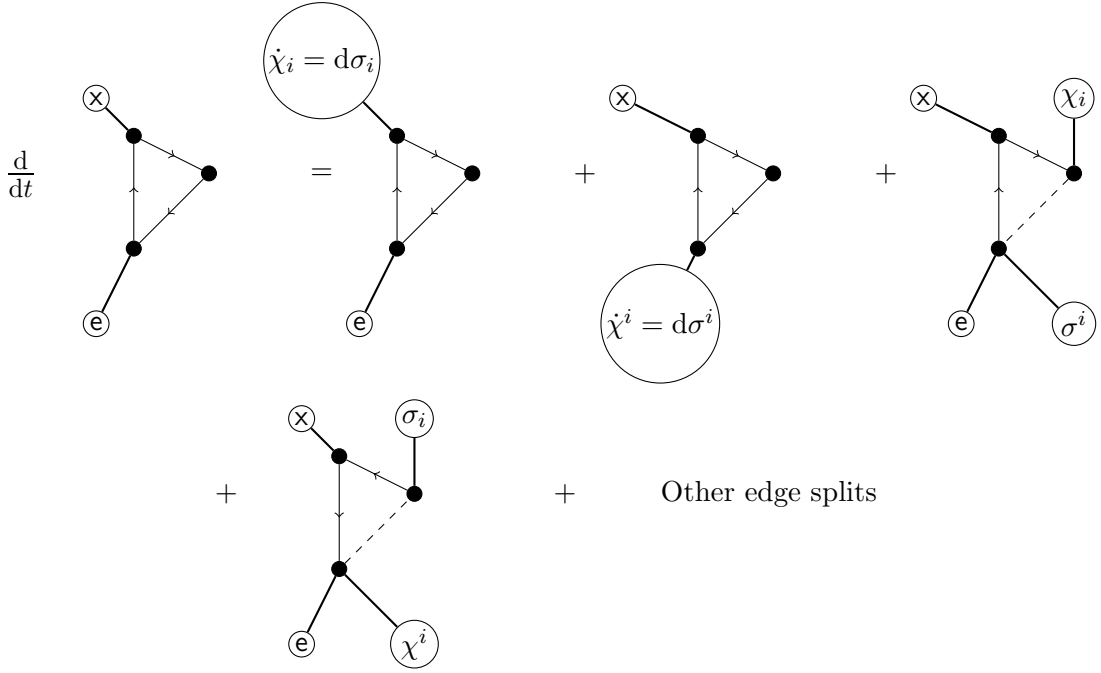

Now we show that 
\begin{equation}
\label{change_principal}
\frac{\dd}{\dd t}\btpsi=\left(\dd_x-\I\hbar\Delta_{\calV_\Sigma}+\frac{\I}{\hbar}\boldsymbol{\Omega}_{\de\Sigma}\right)(\btpsi \bullet \varrho).
\end{equation}
The easiest way to see this is by considering graphs. We can expand 
\begin{equation}
\btpsi \bullet \varrho = \sum_{\Gamma^\mathsf{l}_m}\int_{\mathsf{C}_{\Gamma^\mathsf{l}_m}}\omega'_{\Gamma^\mathsf{l}_m}
\end{equation} 
where now the sum runs over all graphs, either connected or not. Now, we again consider Stokes' theorem for integrals along the fiber 
\begin{equation}
\label{eq:stokes_bullet}
\dd_x(\btpsi \bullet \varrho) = \sum_{\Gamma^\mathsf{l}_m}\int_{\mathsf{C}_{\Gamma^\mathsf{l}_m}}
\dd\omega'_{\Gamma^\mathsf{l}_m} \pm \int_{\de \mathsf{C}_{\Gamma^\mathsf{l}_m}}\omega'_{\Gamma^\mathsf{l}_m}
\end{equation}   One can check that the edge split in $\frac{\dd}{\dd t}\btpsi$ corresponds to $\Delta_{\calV_\Sigma}$ applied to a graph (up to some signs) having the same amount of leaves and carrying the same residual fields with the difference that for the two leafs splitting an edge, where one leaf carries a residual field $\mathsf{x}=z^i\chi_{i}$ ($\mathsf{e}=z^+_i \chi^i$) and the other one a primitve field $z^+_i\sigma^{i}$ ($z^i\sigma_{i}$). Moreover, one can check that the graphs with the $\dd\sigma_i$ ($\dd\sigma^{i}$) leaves on the left hand side get produced by applying the de Rham differential on the configuration space to graphs with a $\sigma_i$ ($\sigma^{i}$) leaf. The rest of terms cancels out as in the proof of the mdQME. 
\end{proof}

\subsection{Changing the propagator} We have the following proposition: 
\begin{prop}[Change of data: propagator]
 Suppose we change the propagator $\zeta$ by an exact form $\lambda \in \Omega^{n-2}(\mathsf{C}_2(\Sigma))$ with the appropriate boundary conditions (but keep the residual fields fixed), 
\begin{equation}
\dot{\zeta} = \dd \lambda.  
\end{equation}
Then the family of states $(\btpsi_t)$ and the familiy of BFV boundary operators $(\boldsymbol{\Omega}_t)$ change by 
\begin{align}
\frac{\dd}{\dd t}\Big\vert_{t=0}\boldsymbol{\Omega}_t &= \dd_x\tau+[\boldsymbol{\Omega}_{t=0},\tau] \label{eq:changeprop1}\\
\frac{\dd}{\dd t}\Big\vert_{t=0}\btpsi_t &= \qtconn (\btpsi_{t=0} \bullet \varrho) - \tau\btpsi_{t=0} \label{eq:changeprop2}
\end{align}
where $\tau=\tau^\mathbb{X}+\tau^\E$ with
\begin{align}
\tau^\mathbb{X}&= \sum_{n,k\geq 0}\sum_{\Gamma_m^{e'}}\frac{(\I\hbar)^{\textnormal{loops}(\Gamma_m^{e'})}}{\vert \textnormal{Aut}(\Gamma_m^{e'})\vert}\int_{\de_1\Sigma}\left(\sigma_{\Gamma_m^{e'}}\right)_{I_1....I_n}^{J_1...J_k}\land\left[\mathbb{X}^{I_1}\right]\land \dotsm \land\left[\mathbb{X}^{I_n}\right] \left((-1)^{kd}(\I\hbar)^k\frac{\delta^{\vert J_1\vert+\dotsm +\vert J_k\vert}}{\delta\left[\mathbb{X}^{J_1}\dotsm \mathbb{X}^{J_k}\right]}\right),\\
\tau^\E&=\sum_{n,k\geq 0}\sum_{\Gamma_m^{e'}}\frac{(\I\hbar)^{\textnormal{loops}(\Gamma_{m}^{e'})}}{\vert \textnormal{Aut}(\Gamma_m^{e'})\vert}\int_{\de_2\Sigma}\left(\sigma_{\Gamma_{m}^{e'}}\right)^{I_1....I_n}_{J_1...J_k}\land \left[\mathbb{E}_{I_1}\right]\land\dotsm\land \left[\mathbb{E}_{I_n}\right] \left((-1)^{kd}(\I\hbar)^k\frac{\delta^{\vert J_1\vert+\dotsm +\vert J_k\vert}}{\delta\left[\mathbb{E}_{J_1}\dotsm \mathbb{E}_{J_k}\right]}\right),
\end{align}
where $\sigma_{\Gamma_m^{e'}}$ is given similarly as in Definition \ref{full_BFV}, with the difference that we place $\lambda$ at the marked edge, and
\begin{equation}
\varrho=\sum_{\text{$\Gamma^{e}_m$ marked} \atop\text{and connected graph}}\varrho_{\Gamma^{e}_m}=\sum_{\text{$\Gamma^{e}_m$ marked} \atop \text{and connected graph}}\int_{\mathsf{C}_{\Gamma^{e}_m}}\omega^{e}_{\Gamma^{e}_m}\in\Gamma(\mathcal{H}_{tot}),
\end{equation}
where $\Gamma^e_m$ denotes a marked connnected graph with edge $e$ labeled by $\lambda$. and $\omega^{e}_{\Gamma^e_m}$ is the form constructed with the usual Feynman rules where we place $\lambda$ at the marked edge $e$.
\end{prop}

\begin{proof}
Let us consider $\boldsymbol{\Omega}$ first. Let $\zeta_t$ be a family of propagators with $\dot{\zeta} = \frac{\dd}{\dd t} \vert_{t=0} \zeta = \dd\lambda$ and $\boldsymbol{\Omega}_t$ corresponding family of BFV boundary operators, which we loosely write as $\boldsymbol{\Omega}_t = \sum_{\Gamma'}\int_{\de\Sigma}\Omega_{\Gamma',t}$. (Here the prime on $\Gamma'$ should remember us that we are taking ``boundary graphs''.)  $\Omega_{\Gamma'}$ is constructed from the propagators, boundary composite fields and derivatives with respect to composite fields. The only thing that depends on $t$ is the propagator, hence the time derivative satisfies 
$$\ddt\boldsymbol{\Omega}_t = \sum_{\Gamma^{e'}_m}\int_{\de\Sigma}\Omega_{\Gamma^{e'}_m,t}.$$ Here $\Gamma^{e'}_m$ is a graph with a marked edge $e$ and in $\Omega_{\Gamma^{e'}}$ we evaluate the marked edge to $\dd\lambda$. Using Stokes' theorem for fiber integration $\dd_x\int = \int \dd \pm \int_{\de}$, we can pull out the de Rham differential on $\lambda$ of the integration. This gives the term $\dd_x\tau$. The other terms from Stokes' theorem are of three kinds: The first kind are terms where the de Rham differential hits a propagator. These vanishe, because the limiting propagator is closed. The second kind are terms where the de Rham differential hits a boundary field, this corresponds to $[\tau,\boldsymbol{\Omega}_0]$. Finally, the boundary terms assemble to  $[\tau,\boldsymbol{\Omega}_{\textnormal{pert},t=0}]$, similar to the proof of Lemma \ref{lem:Omegasquare}. This proves Equation \eqref{eq:changeprop1}.\\
The proof for the derivative of the state works in a similar way. In this case $\varrho$ is given by the sum of all connected Feynman graphs with one marked edge, evaluated using the usual Feynman rules, but placing $\lambda$ at the marked edge. Now, observe that
$$\frac{\dd}{\dd t} \omega_{\Gamma} = \sum_{e \in E(\Gamma)}\omega_{\Gamma}^e,$$
where $\omega_{\Gamma}^e$ is the form obtained by placing $\dd \lambda$ at the edge $e$. Again, we integrate by parts using Stokes' theorem. We get eight different types of terms. 
\begin{enumerate}
\item First, $\dd$ can come out of the integral. This corresponds to $\dd_x(\btpsi \bullet \rho)$ ($(\btpsi \bullet \rho)$ is precisely given by summing over all graphs (not necessarily connected) with a single marked edge that evaluates to $\lambda$). 
\item Terms where $\dd$ hits a propagator correspond to the action of $-\ii\hbar\Delta_{\calV_\Sigma}$ on $(\btpsi \bullet \rho)$. 
\item Terms where $\dd$ hits a boundary field correspond to the action of $\boldsymbol{\Omega}_0$ on $(\btpsi \bullet \rho)$. 
\item Terms where $\dd$ hits a vertex, they cancel with boundary terms just below:
\item Boundary terms corresponding to the collapse of two vertices in the bulk with a single edge between them. If this edge is marked, the pushforward over the boundary sphere vanishes\footnote{The propagators $\zeta_t$ are normalised to integrate to $\pm 1$ over this sphere, hence the integral of $\lambda$ must vanish. }. The terms without marked edges cancel out with the terms where $\dd$ hits a vertex. 
\item Boundary terms where a subgraph with more that two vertices collapses in the bulk, these vanish by assumption. 
\item Boundary terms where a subgraph \emph{without} a marked edge collapses on the boundary. These correspond to the action of $\boldsymbol{\Omega}_{\textnormal{pert},t=0}$ on $(\btpsi \bullet \rho)$. 
\item Finally, in this case we can have graphs with a marked edge collapsing at the boundary, which correspond to the action of $\tau$ to $\btpsi$. 
\end{enumerate}
This completes the proof of Equation \eqref{eq:changeprop2}.  
\end{proof}

\subsection{Changing the formal exponential map}
We can also change the connection on the graded manifold $M$ used to construct the formal exponential map, which is described in Appendix \ref{app:change_of_phi} and \ref{app:formalexpo_graded}. From a multivector field $Y$ on $M$ we can construct a functional 
\begin{equation}
\label{general_action}
\mathcal{S}_{\Sigma,Y}=\frac{1}{k!}\int_\Sigma Y^{i_1,...,i_k}(\sfX)\boldsymbol{\eta}_{i_1}\land\dotsm\land\boldsymbol{\eta}_{i_k}.
\end{equation}
We can do the same construction pointwise for formal vertical multivector fields $\Hat{Y}$, yielding 
\begin{equation}
\calS_{\Sigma,\Hat{Y}} = \frac{1}{k!}\int_{\Sigma}\Hat{Y}^{i_1,\ldots,i_k}(x;\Hat{\sfX})\Hat{\boldsymbol{\eta}}_{i_1}\wedge \cdots \wedge \Hat{\boldsymbol{\eta}}_{i_k}. 
\end{equation} Moreover, writing $\btpsi=\int_\mathcal{L}\ee^{\frac{\I}{\hbar}\Tilde{\mathcal{S}}_{\Sigma,x}}$, for some Lagrangian submanifold $\calL$ of the space of fields, one formally obtains (\cite{BCM}) that
\begin{equation}
\label{eq:pathintegral_change}
\frac{\dd}{\dd t}\btpsi=(\dd_x-\I\hbar\Delta_{\calV_{\Sigma,x}})\int_{\mathcal{L}}\ee^{\frac{\I}{\hbar}\Tilde{S}_{\Sigma,x}}\frac{\I}{\hbar}\mathcal{S}_{\Sigma,C}
\end{equation}
if $\Sigma$ is a closed manifold, where $C \in \Gamma(M,TM \otimes \Hat{S}T^*M)$ is a generator of the gauge transformation applied to the formal exponential map; see Appendix \ref{app:change_of_phi} and \ref{app:formalexpo_graded}.  Formula \eqref{eq:pathintegral_change} motivates to introduce graphs with one marked vertex, labeled by $C$, with vertex tensor coming from the formal Taylor expansion of 
\begin{equation}
\calS_{\Sigma,C}(\sfX,\Hat{\boldsymbol\eta}) = \int_{\Sigma}C^k(x;\sfX)\Hat{\boldsymbol\eta}_k = \int_{\Sigma} C^k_{i_1\ldots i_k}(x)\Hat{\sfX}^{i_1}\cdots\Hat{\sfX}^{i_k}\Hat{\boldsymbol\eta}_k.
\end{equation}

\begin{prop}[Change of data: formal exponential map]
Let $C$ be the generator of a gauge transformation of the Grothendieck connection as in Appendix \ref{app:formal_geometry}.
Then the family of states $(\btpsi_t)$ and the familiy of BFV boundary operators $(\boldsymbol{\Omega}_t)$ change by 
\begin{align}
\frac{\dd}{\dd t}\Big\vert_{t=0}\boldsymbol{\Omega}_t &= \dd_x\tau+[\boldsymbol{\Omega}_{t=0},\tau] \label{eq:expchange1} \\
\frac{\dd}{\dd t}\Big\vert_{t=0}\btpsi_t &= \qtconn (\btpsi_{t=0} \bullet \varrho) - \tau\btpsi_{t=0}  \label{eq:expchange2}
\end{align}
where $\tau=\tau^\mathbb{X}+\tau^\E$ with
\begin{align}
\tau^\mathbb{X}&= \sum_{n,k\geq 0}\sum_{\Gamma_m^{v'}}\frac{(\I\hbar)^{\textnormal{loops}(\Gamma_m^{v'})}}{\vert \textnormal{Aut}(\Gamma_m^{v'})\vert}\int_{\de_1\Sigma}\left(\sigma_{\Gamma_m^{v'}}\right)_{I_1....I_n}^{J_1...J_k}\land\left[\mathbb{X}^{I_1}\right]\land \dotsm \land\left[\mathbb{X}^{I_n}\right] \left((-1)^{kd}(\I\hbar)^k\frac{\delta^{\vert J_1\vert+\dotsm +\vert J_k\vert}}{\delta\left[\mathbb{X}^{J_1}\dotsm \mathbb{X}^{J_k}\right]}\right),\\
\tau^\E&=\sum_{n,k\geq 0}\sum_{\Gamma_m^{v'}}\frac{(\I\hbar)^{\textnormal{loops}(\Gamma_{m}^{v'})}}{\vert \textnormal{Aut}(\Gamma_m^{v'})\vert}\int_{\de_2\Sigma}\left(\sigma_{\Gamma_{m}^{v'}}\right)^{I_1....I_n}_{J_1...J_k}\land \left[\mathbb{E}_{I_1}\right]\land\dotsm\land \left[\mathbb{E}_{I_n}\right] \left((-1)^{kd}(\I\hbar)^k\frac{\delta^{\vert J_1\vert+\dotsm +\vert J_k\vert}}{\delta\left[\mathbb{E}_{J_1}\dotsm \mathbb{E}_{J_k}\right]}\right),
\end{align}
where $\sigma_{\Gamma_m^{v'}}$ is given similarly as in Definition \ref{full_BFV}, with the difference that we place $C$ at the marked vertex, and
\begin{equation}
\label{rho_residual}
\varrho=\sum_{\text{$\Gamma^{v}_m$ marked} \atop\text{and connected graph}}\varrho_{\Gamma^{v}_m}=\sum_{\text{$\Gamma^{v}_m$ marked} \atop \text{and connected graph}}\int_{\mathsf{C}_{\Gamma^{v}_m}}\omega'_{\Gamma^{v}_m}\in\Gamma(\mathcal{H}_{tot}),
\end{equation}
where $\Gamma^v_m$ denotes a marked connnected graph with vertex $v$ labeled by $C$, i.e. a graph and $\omega^{v}_{\Gamma^v_m}$ is the form constructed with the usual Feynman rules where we place $C$ at the marked vertex $v$.
\end{prop}

\begin{proof}
If we vary the formal exponential map, the vertex tensors at interaction and $R$ vertices change according to the formulas \begin{align}
 \dot{\mathcal{S}}_{\Sigma,x} &=-L_C\mathcal{S}_{\Sigma,x} \\ 
 \dot{R} &= \dd_xC + [R,C] \label{eq:R_change}
\end{align}
Since we have $-L_CS_{\Sigma,x} = (S_{\Sigma,C},S_{\Sigma,x})$ and $S_{[R,C]}= (S_{\Sigma,R},S_{\Sigma,C})$, in terms of Taylor expansions the time derivatives are obtained by contracting the terms in the Taylor expansion $(\Theta^{i_1\ldots i_k}_{j_1 \ldots j_l}$ or $Y^{i}_{j,i_1\ldots i_k}$) with $C^{i}_{j,i_1\ldots i_k}$ in all possible ways (plus taking the differential in the case of $Y$'s).
Keeping this in mind, the proof proceeds completely analogously to the proofs before: In   the term $\qtconn (\btpsi \bullet \varrho)$ on left hand side of \eqref{eq:expchange2}, the usual cancellations apply. Terms which survive are: 
\begin{enumerate}
\item Terms where $\dd_x$ hits the $C$ vertex. 
\item Boundary terms corresponding to the collapse of a single edge with the $C$ vertex at one endpoint, and an $R$ vertex at the other endpoint.
\item Boundary terms corresponding to the collapse of a single edge with the $C$ vertex at one endpoint, and an  interaction ($\Theta$) vertex at the other endpoint.
\item Boundary terms where a subgraph containing the $C$ vertex collapses. 
\end{enumerate}
The first and the second type of terms yield the time derivative of $R$ vertices. The third type of terms yield time the derivative of an interaction vertex. Finally, the last type of terms yield the action of $\tau$ on $\btpsi$. This completes the proof of \eqref{eq:expchange2}. The proof of \eqref{eq:expchange1} is entirely the same, using the method of the proof of \eqref{eq:changeprop1}. 
\end{proof}

\begin{rem}
In this paper we only considered free boundaries. We could also consider a boundary component $\de_{\text{fix}}\Sigma$ where we put boundary condition. As explained in \cite{CF4} boundary conditions compatible with the BV formalism are $Q$-invariant Lagrangian submanifolds of the boundary space of fields. As we prove the dQME (and the mdQME) by Stokes' theorem, we have also to take boundary contributions on $\de_\text{fix}\Sigma$ into account. The classical boundary conditions mentioned above make the contributions corresponding to a single bulk point approaching $\de_\text{fix}\Sigma$ vanish. If the terms where two or more bulk points collapsing at $\de_\text{fix}\Sigma$ do not vanish, the theory needs quantum boundary corrections (similarly to what happens in the Landau--Ginzburg model \cite{KL,BHLS,Laz}). We will consider this in the case of the Poisson Sigma Model in \cite{CMW3}.
\end{rem}

\begin{appendix}

\section{Configuration spaces and their compactifications}\label{app:Conf}

To define the quantum state, we need to recall the notion of configuration spaces and their compactification as in \cite{AS2,FulMacPh} due to Fulton--MacPherson and Axelrod--Singer. 
\subsection{FMAS compactification}
We start with the definition of the configuration space. 
\begin{defn}
Let $M$ be a manifold and $S$ a finite set. The \textsf{open configuration space} of $S$ in $M$ is defined as
\begin{equation}
\mathsf{Conf}_S(M) := \{\iota\colon S \hookrightarrow M |\iota \hspace{0.2cm}\textnormal{injection}\}
\end{equation}
\end{defn}
Elements of $\textsf{Conf}_{S}(M)$ are called $S$-configurations. To give an explicit definition of the compactification that can be extended to manifolds with boundaries and corners, we introduce the concept of \emph{collapsed configurations}. Intuitively, a collapsed $S$-configuration is the result of a collapse of a subset of the points in the $S$-configuration. However, we remember the relative configuration of the points before the collapse by directions in the tangent space. This is a configuration in the tangent space that is well-defined only up to translations and scaling. The difficulty is that one can imagine a limiting configuration where two points collapse first together and then with a third (see Figure \ref{fig:collapsing_conf}).
This explains the recursive nature of the following definition. Recall that if $X$ is a vector space, then $X\times \R_{>0}$ acts on $X$ by translations and scaling. 
\begin{defn}[Collapsed configuration in $M$]
Let $M$ be a manifold, $S$ a finite set and $\mathfrak{P} = \{S_1,\ldots,S_k\}$ be a partition of $S$. A \emph{$\mathfrak{P}$-collapsed configuration in $M$} is a $k$-tuple $(p_{\sigma},c_{\sigma})$ such that 
$((p_{\sigma},c_{\sigma}))_{\sigma = 1}^k$ satisfies 
\begin{enumerate}
\item $p_{\sigma} \in M$ and $p_{\sigma} \neq p_{\sigma'}$, for $\sigma \neq \sigma'$, 
\item $c_{\sigma} \in \Tilde{\mathsf{C}}_{S_{\sigma}}(T_{p_{\sigma}}M)$, where for $|S| = 1$, 
$\Tilde{\mathsf{C}}_S(X) := \{pt\}$  and for $|S| \geq 2$
\begin{equation}
\Tilde{\mathsf{C}}_S(X) := \coprod_{\substack{\mathfrak{P}=\{S_1,\ldots,S_k\} \\ S = \sqcup_\sigma S_\sigma, k\geq 2}}\left\lbrace \left(x_\sigma, c_\sigma )\right)_{1\leq \sigma \leq k}\ \bigg|\ (x_\sigma, c_\sigma)\ \mathfrak{P}\text{-collapsed $S$-configuration in $X$}\right\rbrace \bigg/(X \times \R_{>0})
\end{equation}
\end{enumerate}

Here, $\varphi \in X \times \R_{>0}$ acts on $(x_\sigma,c_\sigma)$ by $(x_\sigma,c_\sigma) \mapsto (\varphi(x_\sigma), \dr\varphi_{x_\sigma}c_\sigma)$. 
\end{defn}
Intuitively, given a partition $\mathfrak{P}= \{S_1,\ldots,S_k\}$, a $k$-tuple $(p_\sigma,c_\sigma)$ describes the collapse of the points in $S_\sigma$ to $p_\sigma$. $c_\sigma$ remembers the relative configuration of the collapsing points. This relative configuration can itself be the result of a collapse of some points. See Figure \ref{fig:collapsing_conf}. 
\begin{figure}[h!]
\begingroup%
  \makeatletter%
  \providecommand\color[2][]{%
    \errmessage{(Inkscape) Color is used for the text in Inkscape, but the package 'color.sty' is not loaded}%
    \renewcommand\color[2][]{}%
  }%
  \providecommand\transparent[1]{%
    \errmessage{(Inkscape) Transparency is used (non-zero) for the text in Inkscape, but the package 'transparent.sty' is not loaded}%
    \renewcommand\transparent[1]{}%
  }%
  \providecommand\rotatebox[2]{#2}%
  \ifx\svgwidth\undefined%
    \setlength{\unitlength}{177.26514562bp}%
    \ifx\svgscale\undefined%
      \relax%
    \else%
      \setlength{\unitlength}{\unitlength * \real{\svgscale}}%
    \fi%
  \else%
    \setlength{\unitlength}{\svgwidth}%
  \fi%
  \global\let\svgwidth\undefined%
  \global\let\svgscale\undefined%
  \makeatother%
  \begin{picture}(1,1.22602258)%
    \put(0,0){\includegraphics[width=\unitlength]{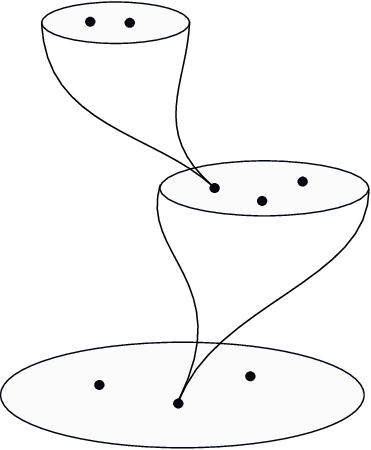}}%
    \put(0.28504216,0.22615565){\color[rgb]{0,0,0}\makebox(0,0)[lb]{\smash{$p_1$}}}%
    \put(7.45141911,-5.32081804){\color[rgb]{0,0,0}\makebox(0,0)[lt]{\begin{minipage}{2.35581904\unitlength}\raggedright \end{minipage}}}%
    \put(0.69011264,0.72199957){\color[rgb]{0,0,0}\makebox(0,0)[lb]{\smash{$p_3$}}}%
    \put(0.7144357,0.19405763){\color[rgb]{0,0,0}\makebox(0,0)[lb]{\smash{$p_2$}}}%
    \put(0.8659863,0.0058981){\color[rgb]{0,0,0}\makebox(0,0)[lb]{\smash{$M$}}}%
    \put(0.84598042,0.72266449){\color[rgb]{0,0,0}\makebox(0,0)[lb]{\smash{$p_4$}}}%
    \put(0.37765289,1.15553808){\color[rgb]{0,0,0}\makebox(0,0)[lb]{\smash{$p_5$}}}%
    \put(0.17364444,1.16429723){\color[rgb]{0,0,0}\makebox(0,0)[lb]{\smash{$p_6$}}}%
  \end{picture}%
\endgroup%

\caption{An element of $\mathsf{C}_S(M)$.}\label{fig:collapsing_conf}
\end{figure}
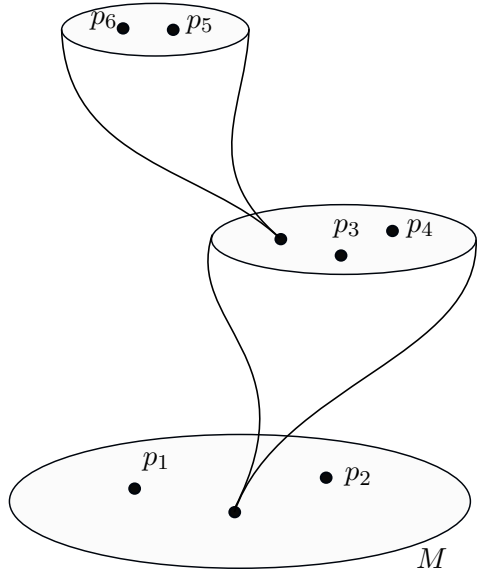
\begin{defn}[FMAS compactification]
The \textsf{compactified configuration space} $\mathsf{C}_S(M)$ of $S$ in $M$ is given by 
\begin{equation}
\mathsf{C}_S(M) := \coprod_{\substack{S_1,\ldots,S_k \\ S = \sqcup_\sigma S_\sigma}}\left\lbrace (p_\sigma, c_\sigma )_{1\leq \sigma \leq k}\ \bigg|\ (p_\sigma, c_\sigma)\ \mathfrak{P}\text{-collapsed $S$-configuration in $M$}\right\rbrace.
\end{equation}
%
\end{defn}

\subsection{Boundary strata}
A precise description of the combinatorics of the stratification can be found in \cite{FulMacPh}, where it is also shown that $\mathsf{C}_S(M)$ is a manifold with corners and is compact if $M$ is compact. For us, only strata in low codimensions are interesting. Let $S=\{s_1,\ldots,s_k\}$. 
The stratum of codimension 0 corresponds to the partition $\mathfrak{P}= \{\{s_1\},\ldots,\{s_k\}\}$. Strata of codimension 1 correspond to the collapse of exactly one subset $S'=\{s_1,\ldots,s_\ell\} \subset S$ with no further collapses, i.e a partition $\mathfrak{P}=\{\{s_1,\ldots,s_\ell\},\{s_{\ell+1}\},\ldots,\{s_k\}\}$ and configuration $(p_\sigma,c_\sigma)$ with $c_\sigma$ in the component of $\Tilde{\mathsf{C}}_{S'}(X)$ given by the partition $\mathfrak{P} =  \{\{s_1\},\ldots,\{s_\ell\}\}$. This boundary stratum will be denoted by $\de_{S'}\mathsf{C}_S(M)$, in particular, we have
\begin{equation}
\de\mathsf{C}_S(M) = \coprod_{S' \subset S}\de_{S'}\mathsf{C}_S(M).
\end{equation} 
There is a natural fibration $\de_{S'}\mathsf{C}_S(M) \to \mathsf{C}_{S \setminus S'\cup \{pt\}}(M)$ whose fiber is $\Tilde{\mathsf{C}}_S(\R^{\dim M})$. Finally, we note that if $|S| = 2$, then $\mathsf{C}_{S}(M) \cong Bl_{\overline{\Delta}}(M \times M)$, the \textsf{differential-geometric blow-up} of the diagonal $\overline{\Delta} \subset M \times M$, and $\Tilde{\mathsf{C}}_S(X) \cong S^{\dim X-1}$.
\\

\subsection{Configuration spaces on manifolds with boundary}
We proceed to recall the definition of a compactified configuration space on manifolds with boundary. Let $M$ be a compact manifold with boundary $\de M$. Recall that for a manifold $M$ with boundary $\de M$, at points $p \in \de M$ there is a well-defined notion of inward and outward half-space in $T_pM$. If $H \subset X$ is a half-space, then $\de H \subset X$ is a hyperplane. $\de H \times \R_{>0}$ acts on $H$ by translations and scaling. 
%
%
\begin{defn}[Configuration spaces on manifolds with boundary]
Let $M$ be a manifold with boundary $\de M$. For $S,T$ finite sets, we define the \textsf{open configuration space} by
\begin{equation}
\mathsf{Conf}_{S,T}(M,\de M) := \{(\iota,\iota')\colon S \times T \hookrightarrow M \times \de M\}
\end{equation}
\end{defn}
\begin{defn}[Collapsed configuration on manifolds with boundary]
Let $(M, \de M)$ be a manifold with boundary. Let $S, T$ be finite sets and $\mathfrak{P}=\{S_1, \ldots, S_k\}$ a partition of $S \sqcup T$.  Then, a $\mathfrak{P}$-collapsed $(S,T)$-configuration in $M$ is a $k$-tuple of pairs $(p_\sigma,c_\sigma)$ such that
\begin{enumerate}
\item $p_\sigma \in M$ and $p_\sigma\neq p_{\sigma'}$, for all $\sigma\neq \sigma'$, 
\item $S_\sigma \cap T \neq \varnothing \Rightarrow p_\sigma \in \de M$,
\item $$c_\sigma \in \begin{cases} 
\Tilde{\mathsf{C}}_{S_\sigma}(T_{p_\sigma}M)  & p_\sigma \in M \setminus \de M \\  \Tilde{\mathsf{C}}_{S \cap S_\sigma, T \cap S_\sigma}(\mathbb{H}(T_{p_\sigma}M)) & p_\sigma \in \de M\end{cases}$$
\end{enumerate}
where $\mathbb{H}(T_{p_\sigma}M) \subset T_{p_\sigma}M$ denotes the inward half-space in $T_{p_\sigma}M$. Here, for a vector space $X$ and a half-space $H \subset X$, $\Tilde{\mathsf{C}}_{\varnothing,\{pt\}}(H) := \Tilde{\mathsf{C}}_{\{pt\},\varnothing}(H) := \{pt\}$, and for $S \sqcup T| \geq 2$, 
$$\Tilde{\mathsf{C}}_{S,T}(H) := \coprod_{\substack{\mathfrak{P}=\{S_1,\ldots,S_k\}\\ S\sqcup T = \sqcup_\sigma S_\sigma, k \geq 2 }}\left\lbrace (v_\sigma,c_\sigma)\ \bigg|\ (v_\sigma, c_\sigma)\ \mathfrak{P}\text{-collapsed $(S,T)$-configuration in $H$} \right\rbrace \bigg/ (\de H \times \R_{>0}) $$

\end{defn}
\begin{defn}[FMAS compactification for manifolds with boundary]
We define the  \textsf{compactification} $\mathsf{C}_{S,T}(M,\de M)$ of $\mathsf{Conf}_{S,T}(M,\de M)$ by 
\begin{equation}
\mathsf{C}_{S,T}(M,\de M) = 
\coprod_{\substack{ \mathfrak{P} =\{S_1,\ldots,S_k\} \\ S\sqcup T= \sqcup_\sigma S_\sigma}}\left\lbrace \left(p_\sigma , c_\sigma\right)_{1\leq \sigma \leq k}\ \bigg|\ (p_\sigma,c_\sigma)\  \mathfrak{P}\text{-collapsed $(S,T)$-configuration}\right\rbrace
\end{equation}

\end{defn}
 Again, this is a manifold with corners and is compact if $M$ is compact. We proceed to describe the strata of low codimension. Let $U= \{u_1,\ldots,u_k\}, V = \{v_1,\ldots,v_k\}.$ The codimension 0 stratum again is given by the partition $\mathfrak{P} = \{\{u_1\},\ldots,\{u_k\},\{v_1\},\ldots, \{v_\ell\}\}.$ Let us describe the strata of codimension 1. We denote by $\de^\romI_S\mathsf{C}_{U,V}(M,\de M)$ a boundary stratum where a subset $S \subset U$ collapses in the bulk, described in the same way as above.  On manifolds with boundary, there are new boundary strata in the compactified configuration space given by the collapse of a subset of points to a point in the boundary. Concretely, given a subset $S=\{u_1,\ldots,u_{k'},v_1,\ldots,v_{\ell'}\} \subset U \sqcup V$, there is a boundary stratum $\de^{\romII}_S\mathsf{C}_{U,V}(M,\de M)$ corresponding to the partition $\mathfrak{P}=\{S,\{u_{k'+1}\},\ldots, \{u_k\},\{v_{\ell'+1}\},\ldots,\{v_\ell\}\}$ and collapsed configurations $(p_\sigma,c_\sigma)$ with $p_{\sigma} \in \de M$ and $c_\sigma$ corresponding to the partition $\mathfrak{P'} = \{\{u_1\},\ldots,\{u_k\},\{v_1\},\ldots,\{v_\ell\}\}$. The boundary decomposes as 
\begin{equation}
\de \mathsf{C}_{U,V}(M,\de M) = \coprod_{S \subseteq U} \de^\romI_S\mathsf{C}_{U,V}(M,\de M) \amalg\coprod_{S \subseteq U \sqcup V} \de^{\romII}_S\mathsf{C}_{U,V}(M,\de M)
\end{equation}

\section{Formal geometry}\label{app:formal_geometry}
We are interested in how perturbative expansions change if one changes the point of expansion. The language of formal geometry (\cite{GK,B}) provides adequate tools to study how the coefficients of Taylor expansions change if one changes coordinates. In this appendix we recollect some notions of formal geometry. We follow the expositions of \cite{CF3} and \cite{BCM}, and refer to these papers for proofs of the statemets. Another good reference is \cite{D}. 
\subsection{Formal power series on vector spaces} 
We begin with a very short review of formal power series on vector spaces. 
If $V$ is a finite-dimensional vector space, the polynomial algebra on $V$ is 
the symmetric algebra of the dual vector space \[S^{\bullet}V^* = \bigoplus_{k=0}^{\infty} S^kV^*.\] If $e_1,\ldots,e_n$ is a basis of $V$, with dual basis $y^1,\ldots,y^n$, then elements $f \in S^{\bullet}V^*$ can be represented by 
\[ f(y) = \sum_{i_1,\ldots,i_n=1}^{\infty} f_{i_1,\ldots,i_k}y_1^{i_1}\cdots y_n^{i_n} = \sum_{I}f_Iy^I,\] 
with only finitely many non-vanishing $f_I$. Here $I = \{i_1,\ldots,i_k\}$ is a multi-index and we understand $y^I = y^{i_1}\cdots y^{i_k}, y^\varnothing:=1$.
\\ We can complete this algebra to the algebra of formal power series $\Hat{S}V^*$, where infinitely many coefficents $f_I$ can be nonzero. Both $S^{\bullet}V^*$ and $\Hat{S}V^*$ are commutative algebras with the multiplication of polynomials or formal power series respectively, generated by $V^*$. Derivations of these algebras are specified by their value on these generators, hence the map 
\begin{align}
V \otimes S^{\bullet}V^* &\to \Der(S^{\bullet}V^*) \nonumber \\ 
v \otimes f &\mapsto \left(V^* \ni \alpha \mapsto \alpha(v) \cdot f\right) \label{eq:derivationiso} 
\end{align}
is an isomorphism with inverse 
\begin{align*}
 \Der(S^{\bullet}V^*)  &\to V \otimes S^{\bullet}V^*  \\ 
D &\mapsto \sum_{i=1}^n e_i \otimes D(y^i) \label{eq:derivationiso} 
\end{align*}
 In coordinates, it simply amounts to sending $e^i \mapsto \frac{\partial}{\partial y^i}$. 

\subsection{Formal exponential maps}
\label{App:formal_exp}
Let $M$ be a smooth manifold. Let $\varphi \colon U \to M$ where $U \subset TM$ is a open neighbourhood of the zero section. For $x \in M, y \in T_xM \cap U$ we write $\varphi(x,y) = \varphi_x(y)$. We say that $\varphi$ is a \textsf{generalized exponential map} if for all $x \in M$ we have that $\varphi_x(0) = x, \dr\varphi_x(0) = \mathrm{id}_{T_xM}$. In local coordinates we can write 
\[
\varphi_x^{i}(y)=x^{i}+y^{i}+\frac{1}{2}\varphi_{x,jk}^{i}y^jy^k+\frac{1}{3!}\varphi^{i}_{x,jk\ell}y^jy^ky^\ell+\dotsm
\]
where the $x^i$ are coordinates on the base and the $y^i$ are coordinates on the fibers. 
We identify two generalised exponential maps if their jets agree to all orders. A \textsf{formal exponential map} is an equivalence class of generalised exponential maps. It is completely specified by the sequence of functions $\left(\varphi^i_{x,i_1\ldots i_k}\right)_{k=0}^{\infty}$. By abuse of notation, we will denote equivalence classes and their representatives by $\varphi$. From a formal exponential map $\varphi$ and a function $f \in C^{\infty}(M)$, we can produce a section $\sigma \in \Gamma(\Hat{S}T^*M)$ by defining $\sigma_x = \T\varphi_x^*f$, where $\T$ denotes the Taylor expansion in the fiber coordinates around $y=0$ and we use any representative of $\varphi$ to define the pullback. We denote this section by $\T\varphi^*f$, it is independent of the choice of representative, since it only depends on the jets  of the representative. 

\subsection{Grothendieck connection}\label{sec_GrConn}
One can define a flat connection $D_\textsf{G}$ on $\Hat{S}T^*M$ with the property that $D_{\textsf{G}}\sigma = 0$ if and only if $\sigma = \T\varphi^*f$ for some $f\in C^\infty(M)$. Namely, $D_\textsf{G} = \dd + R$ where $R \in \Gamma(T^*M \otimes TM \otimes \Hat{S}T^*M)$ is a one-form with values in derivations of $\Hat{S}T^*M$, which we identify with $\Gamma(TM \otimes \Hat{S}T^*M)$ using the isomorphism \eqref{eq:derivationiso}\footnote{This is slightly confusing, since the basis of $T_xM$ is usually denoted $\frac{\partial}{\partial x^i}$, which under this isomorphism gets sent to $\frac{\partial}{\partial y^i}$. }. $R$ can be defined in local coordinates by $R = R_i\dd x^i$ and 
\begin{equation}
R_i(x;y) = \left(\left(\frac{\partial \varphi_x}{\partial y}\right)^{-1}\right)^k_j\frac{\partial\varphi_x^{j}}{\partial x^{i}}\frac{\de }{\de y^k} =: Y^k_i(x;y)\frac{\de}{\de y^k}
\end{equation}
so that 
\[ R(x;y) = R_i(x;y)dx^i = Y_i^k\frac{\partial}{\partial y^k}dx^i,\]
where we use the Einstein summation convention. For $\sigma \in \Gamma(\Hat{S}T^*M)$, $R(\sigma)$ is given by the Taylor expansion (in the $y$ coordinates) of $$-\dd_y\sigma \circ (\dd_y\varphi)^{-1} \circ \dd_x\varphi \colon \Gamma(TM) \to \Gamma(\Hat{S}T^*M),$$
this shows that $R$ does not depend on the choice of coordinates. For a vector field $\xi = \xi \frac{\partial}{\partial x^i}$ we get 
\begin{equation} 
\label{AppA:Aconnection_G}
D_\textsf{G}^{\xi} = \xi + \Hat{\xi},
\end{equation}
where 
\begin{equation} 
\label{AppA:formal_vector}
\Hat{\xi}(x;y) = \iota_{\xi}R(x;y) = \xi^i(x)Y_i^k(x;y)\frac{\partial}{\partial y^k}. 
\end{equation}
The connection $D_{\textsf{G}}$ is called a \textsf{Grothendieck connection}. Its flatness can be expressed by 
\[\dd_xR + \frac{1}{2}[R,R] = 0.\] 
It can be shown that its cohomology is concentrated in degree 0 and is given by $$H^0_{D_\mathsf{G}}(\Gamma(\Hat{S}T^*M))=\mathsf{T}\varphi^*C^\infty(M)\cong C^{\infty}(M).$$
\subsection{Formal vertical tensor fields}
Now, let $E \to M$ be any tensorial\footnote{I.e. any bundle which is a tensor product or antisymmetric or symmetric product of the tangent or cotangent bundle, or a direct sum thereof.} bundle, e.g. $E = \bigwedge^kTM$. Its sections are called \textsf{tensor fields of type $E$}. Then its associated \textsf{formal vertical bundle} is  $\Hat{E}:= E\otimes \Hat{S}T^*M$, and sections of this bundle are called \textsf{formal vertical tensors of type $E$}. One can think of these bundles as tensors of the same type on $TM$ where the dependence on fiber directions is formal. The formal exponential map defines an injective map $\T\varphi^* \colon E \to \Hat{E}$ by taking the Taylor expansion of a tensor field pulled back to $U$ by $\varphi$\footnote{Since $\varphi$ is a local diffeomorphism we can define the pullback of contravariant tensors as the pushforward of the inverse.}.We can let $R$ act on formal vertical tensors by Lie derivative. Thus we get a Grothendieck connection $D_\textsf{G} = \dd + R$ on any formal vertical tensor bundle. Again, it is flat, and the flat sections are precisely the ones in the image of $\T\varphi^*$. Moreover the cohomology is concentrated in degree 0 and given by the flat sections, i.e. $\Hat{E}$-valued 0-forms. 

\subsection{Changing the formal exponential map}\label{app:change_of_phi}
Let $\phi$ be a family of formal exponential maps depending on a parameter $t$ belonging to an open interval $I$. Then we can associate this family a formal exponential map $\psi$ for the manifold $M\times I$ by $\psi(x,t,y,\tau):=((\phi)_{x,t}(y),t+\tau)$, where $\tau$ denotes the tangent variable to $t$. We want to define the associated connection $\Tilde{R}$: on a section $\Tilde{\sigma}$ of $\Hat{S}T^*(M\times I)$ we have, by definition 
\begin{equation}
\label{change_R}
\Tilde{R}(\Tilde{\sigma})=-(\dd_y\Tilde{\sigma},\dd_\tau\Tilde{\sigma})\circ \begin{pmatrix}(\dd_y\phi)^{-1}& 0\\ 0 &1\end{pmatrix}\circ\begin{pmatrix}\dd_x\phi&\dot{\phi}\\0&1\end{pmatrix}.
\end{equation}
So we can write $\Tilde{R}=R+C\dd t+T$ with $R$ defined as in Appendix \ref{sec_GrConn} (but now $t$-dependent),
\[
C(\Tilde{\sigma})=-\dd_y\Tilde{\sigma}\circ(\dd_y\phi)^{-1}\circ\dot{\phi},
\]
and $T=-\dd t\frac{\partial}{\partial\tau}$. We can formulate the MC equation for $\Tilde{R}$ observing that $\dd_xT=\dd_tT=0$ and that $T$ commutes with both $R$ and $C$. The $(2,0)$-form component over $M\times I$ yields again the MC equation for $R$, whereas the $(1,1)$-component reads
\[
\dot{R}=\dd_xC+[R,C].
\]
Hence, under a change of formal exponential map, $R$ changes by a gauge transformation with generator the section $C$ of $\Hat{\mathfrak{X}}(TM):=TM\otimes \Hat{S}T^*M$. Finally, if $\sigma$ is a section in the image of $\T\phi^*$, then by a simple computation one gets 
\[
\dot{\sigma}=-L_C\sigma,
\]
which can be interpreted as the associated gauge transformation for sections. 

\subsection{Extension to graded manifolds} \label{app:formalexpo_graded}
The results of the previous subsections can be generalised to the category of graded manifolds using the algebraic reformulation of formal exponential maps developed in \cite{LS}. \\ 
Namely, given a formal exponential map $\varphi$ on a smooth manifold $\Bar{M}$, one can construct a map 
\begin{equation}
\operatorname{pbw} \colon \Gamma(\Hat{S}T\Bar{M}) \to \mathcal{D}(\Bar{M})
\end{equation} 
from sections of the completed symmetric algebra of the tangent bundle to the algebra of differential operators $\mathcal{D}(\Bar{M})$ by defining 
\begin{equation}
\operatorname{pbw}(X_1 \odot \cdots \odot X_n)(f) = \frac{\dd}{\dd t_1}\bigg|_{t_1 = 0}\cdots \frac{\dd}{\dd t_n}\bigg|_{t_n = 0}f(\varphi(t_1X_1 + \ldots + t_nX_n)).
\end{equation}
This map can be defined also in the category of graded manifolds by choosing a torsion-free connection $\nabla$ on the tangent bundle of a graded manifold $M$. In particular, there still exists an element $R^{\nabla} \in \Omega^1(M,  TM \otimes \Hat{S}T^*M)$ with the property that 
$D_\textsf{G} = \dd_{M} + R^{\nabla}$ is a flat connection on $\Hat{S}T^*M$, namely 
\begin{equation}
R^{\nabla} = -\delta + \Gamma + A^{\nabla},
\end{equation}
where  in local coordinates $x^i$ on $M$ and corresponding coordinates $y^i$ on $TM$ we have (see \cite{D,LS}) 
\begin{align}
\delta &= \dd x^i \frac{\partial}{\partial y^i}, \\
\Gamma &= -\dd x^i \Gamma^k_{ij}(x)y^j\frac{\partial}{\partial y^k}, \\
A^{\nabla} &= \dd x^i \sum_{|J| \geq 2} A_{i,J}^k(x) y^J \frac{\partial}{\partial y^k}.
\end{align}
Here $\Gamma^k_{ij}$ are the Christoffel symbols of $\nabla$. We define $R_i \in \Gamma(M, \Hat{S}T^*M \otimes TM)$ and $Y_i^k \in \Gamma(M,\Hat{S}T^*M)$ by  
\begin{equation}
R^{\nabla} = R_i(x;y) \dd x^i = Y_i^k(x;y) \dd x^i \frac{\partial}{\partial y^k}. \label{eq:def_graded_Y}
\end{equation} $D_\textsf{G}$ extends to a differential on $\Omega^{\bullet}(M,\Hat{S}T^*M)$. The ``Taylor expansion'' of a function $f \in C^{\infty}(M)$ can be defined as (\cite{LS})
\begin{equation}
\T\varphi^*f := \sum_I \frac{1}{I!}y^I\operatorname{pbw}\left(\underleftarrow{\partial_x^I}\right)(f) \label{eq:graded_Taylor}
\end{equation}
where 
$$\underleftarrow{\partial_x^I} = \underbrace{\partial_{x_1} \odot \cdots \odot \partial_{x_1}}_{i_1} \odot \dotsm \odot \underbrace{\partial_{x_n} \odot \cdots \odot \partial_{x_n}}_{i_n}.$$ 
One can prove that \eqref{eq:graded_Taylor} still has the same properties, namely, the image of $\T\varphi^*$ consists precisely of the $D_\textsf{G}$-closed sections of $\Hat{S}T^*M$. \\ 
One can describe how the formal exponential map varies under the choice of connection mimicking the construction for the smooth case described in \ref{app:change_of_phi}. Namely, assume we have a smooth family $\nabla^t$ of connections on $TM$, then we can associate to that family a connection $\widetilde{\nabla}$ on $M \times I$. The corresponding $R^{\widetilde{\nabla}}$ still can be split as 
$$ R^{\widetilde{\nabla}} = R^{\nabla^t} + C^{\nabla^t} \dd t + T $$
as in subsection \ref{app:change_of_phi}, where $C \in \Gamma(M,\Hat{S}T^*M)$. The fact that $(D_\textsf{G})^2 = 0$ translates into the same equations as before, namely, we have 
\begin{equation}
\dot{R^{\nabla^t}} = \dd_M R^{\nabla^t} + \left[C^{\nabla^t},R^{\nabla^t}\right]
\end{equation}
and for any section $\sigma$ in the image of $\T\varphi^*$ we have  
\begin{equation} 
\dot{\sigma} = -L_{C^{\nabla^t}}\sigma.
\end{equation}
\end{appendix}

\printbibliography

\end{document}